%% file: main.tex
\documentclass[11pt]{article}
\usepackage{fullpage}

\usepackage{geometry}
\usepackage[utf8]{inputenc}
\usepackage{authblk}

\usepackage{graphicx}
\usepackage{amsmath} 
\usepackage{amssymb} 
\usepackage{amsthm} 
\usepackage{xcolor}
\usepackage{hyperref}
\usepackage{enumitem}
\usepackage{thm-restate}
\usepackage{comment}
\usepackage{algorithm}
\usepackage{mdframed} 
\usepackage{multirow}
\usepackage{makecell}
\usepackage{color, colortbl}

\hypersetup{
    colorlinks=true,
    citecolor=blue,
    linkcolor=red,
    filecolor=cyan,
    urlcolor=black,
}
\mdfdefinestyle{figstyle}{ %
  linecolor=black!7, %
  backgroundcolor=black!7, %
  innertopmargin=10pt, %
  innerleftmargin=\textwidth, %
  innerrightmargin=\textwidth, %
  innerbottommargin=5pt %
}

\newtheorem{theorem}{Theorem}[section]
\newtheorem{lemma}[theorem]{Lemma}
\newtheorem{claim}[theorem]{Claim}

\newtheorem{corollary}[theorem]{Corollary}
\newtheorem{definition}[theorem]{Definition}
\newtheorem{observation}{Observation}
\newtheorem{remark}{Remark}

\newtheorem{openproblem}{Open Problem}
\newtheorem{conjecture}{Conjecture}

\newtheorem{notation}{Notation}

\newcommand{\fH}{{\sf H}}

\newcommand{\fX}{{\sf X}}

\newcommand{\fI}{{\sf I}}

\newcommand{\fCNOT}{{\sf CNOT}}
\newcommand{\fToffoli}{{\sf Toffoli}}

\newcommand{\fcSWAP}{{\text{c-}\sf SWAP}}
\newcommand{\cS}{{\mathcal{S}}}

\newcommand*{\bra}[1]{\langle#1|}
\newcommand*{\ket}[1]{|#1\rangle}
\newcommand*{\opro}[2]{|#1\rangle\langle#2|}
\newcommand*{\ipro}[2]{\langle #1|#2\rangle}
\newcommand{\g}{\mathcal{G}} 
\newcommand{\Circ}{\mathcal{C}} 

\newcommand{\N}{\mathbb{N}} 
\newcommand{\T}{\mathsf{tt}} 
\newcommand{\C}{\mathbb{C}} 
\newcommand{\E}{\mathop{\mathbb{E}}} 
\newcommand{\A}{\mathcal{A}} 

\newcommand{\class}[1]{\mathsf{#1}}

\newcommand{\MCSP}{\mathsf{MCSP}}
\newcommand{\MQCSP}{\mathsf{MQCSP}}
\newcommand{\UMCSP}{\mathsf{UMCSP}}
\newcommand{\SMCSP}{\mathsf{SMCSP}}

\newcommand{\mBMCSP}{\mathsf{multiMQCSP}}
\newcommand{\negl}{\mathsf{negl}} 
\newcommand{\mcsps}{{\mathsf{MCSP}^\star}}
\newcommand{\mqcsps}{{\mathsf{MQCSP}^\star}}
\newcommand{\iO}{{i{\cal O}}}

\DeclareMathOperator{\poly}{\mathsf{poly}}
\DeclareMathOperator{\polylog}{\mathsf{polylog}}

\usepackage{mathtools}
\DeclarePairedDelimiter\ceil{\lceil}{\rceil}

\newif\ifdraft
\newif\ifarxiv
\arxivtrue
\drafttrue

\definecolor{mygreen}{RGB}{0,180,180}
\definecolor{myblue}{RGB}{0,0,180}
\definecolor{mycrimson}{RGB}{165,28,48}

\newcommand{\nai}[1]{\ifarxiv{\color{mygreen}[Nai: #1]}\fi}

\renewcommand{\thefootnote}{\fnsymbol{footnote}}

\usepackage{algorithmicx}
\usepackage{algorithm}
\usepackage[noend]{algpseudocode}
\newcounter{algsubstate}
\renewcommand{\thealgsubstate}{\alph{algsubstate}}

\algnewcommand\algorithmicinput{\textbf{Input:}}
\algnewcommand\Input{\item[\algorithmicinput]}

\algnewcommand\algorithmicoutput{\textbf{Output:}}
\algnewcommand\Output{\item[\algorithmicoutput]}

\algnewcommand\algorithmicgoal{\textbf{Goal:}}
\algnewcommand\Goal{\item[\algorithmicgoal]}

\makeatletter
\newcounter{phase}[algorithm]
\newlength{\phaserulewidth}

\makeatother

\title{Quantum Meets the Minimum Circuit Size Problem} 

\ifdraft
\author[1]{Nai-Hui Chia\thanks{\texttt{naichia@iu.edu}.}}
\author[2]{Chi-Ning Chou\thanks{\texttt{chiningchou@g.harvard.edu}.}}
\author[3]{Jiayu Zhang\thanks{\texttt{jyz16@bu.edu}.}}
\author[4]{Ruizhe Zhang\thanks{\texttt{ruizhe@utexas.edu}.}}
\affil[1]{Luddy School of Informatics, Computing, and Engineering, Indiana University Bloomington}
\affil[2]{School of Engineering and Applied Sciences, Harvard University}
\affil[3]{Department of Computer Science, Boston University}
\affil[4]{Department of Computer Science, The University of Texas at Austin}
\else
\author{}
\fi
\date{}

\begin{document}

\begin{titlepage}
\maketitle
\begin{abstract}
In this work, we initiate the study of the Minimum Circuit Size Problem (MCSP) in the quantum setting. MCSP is a problem to compute the circuit complexity of Boolean functions. It is a fascinating problem in complexity theory --- its hardness is mysterious, and a better understanding of its hardness can have surprising implications to many fields in computer science.  

We first define and investigate the basic complexity-theoretic properties of minimum quantum circuit size problems for three natural objects: Boolean functions, unitaries, and quantum states. We show that these problems are not trivially in NP but in QCMA (or have QCMA protocols). Next, we explore the relations between the three quantum MCSPs and their variants. We discover that some reductions that are not known for classical MCSP exist for quantum MCSPs for unitaries and states, e.g., search-to-decision reductions and self-reductions. Finally, we systematically generalize results known for classical MCSP to the quantum setting (including quantum cryptography, quantum learning theory, quantum circuit lower bounds, and quantum fine-grained complexity) and also find new connections to tomography and quantum gravity. Due to the fundamental differences between classical and quantum circuits, most of our results require extra care and reveal properties and phenomena unique to the quantum setting.
Our findings could be of interest for future studies, and we post several open problems for further exploration along this direction.

\end{abstract}
\thispagestyle{empty}
\end{titlepage}


\renewcommand{\thefootnote}{\arabic{footnote}}
\setcounter{footnote}{0}
\ifdraft
\newpage
{
  \hypersetup{linkcolor=black}
  \tableofcontents
}
\newpage
\fi
\input{intro}

\input{main-open}

\newpage
\input{main-prelim}

\input{main-MQCSP}

\input{main-connections}
\input{UandS}

\input{main-ack}

\addcontentsline{toc}{section}{References}
\bibliographystyle{alpha}
\bibliography{qmcsp}

\appendix\label{appendix}

\input{SZK}
\input{learning}

\input{circuit_lowerbound}
\input{qeth_hardness}
\input{UandS_appx}
\input{QCircuit}
\end{document}

%% file: intro.tex
\section{Introduction} 
The Minimum Circuit Size Problem ($\MCSP$) is one of the central computational problems in complexity theory. Given the truth table of a Boolean function $f:\{0,1\}^n\rightarrow \{0,1\}$ and a size parameter $s$ (in unary) as inputs, $\MCSP$ asks whether there exists a circuit of size at most $s$ for $f$. While $\MCSP$ has been studied as early as the 1950s in the Russian cybernetics program~\cite{trakhtenbrot1984survey}, its complexity remains mysterious: we do not know whether it is in $\class{P}$ or $\class{NP}$-hard. Meanwhile, besides being a natural computational problem, in recent years, researchers have discovered many surprising connections of $\MCSP$ to other areas such as cryptography~\cite{RR97}, learning theory~\cite{CIKK16}, circuit complexity~\cite{KC00}, average-case complexity~\cite{hirahara2018non}, and others. 


Quantum computing is of growing interest, 
with applications to cryptography~\cite{shor1994algorithms}, machine learning~\cite{biamonte2017quantum}, and complexity theory~\cite{ji2020mip}, etc. Inspired by the great success of $\MCSP$ in classical computation and the flourishing of quantum computers, we propose a new research program of studying quantum computation through the lens of~$\MCSP$. We envision $\MCSP$ as a central problem that connects different quantum computation applications and provides deeper insights into the complexity-theoretic foundation of quantum circuits.



\subsection{The classical \texorpdfstring{$\MCSP$}{MCSP} and its connections to other problems}\label{sec:intro MCSP}

It is immediate that $\MCSP\in\class{NP}$ because the input size is $2^n$ so one can verify if a circuit (given as the certificate/proof) computes the input truth table in time $2^{O(n)}$. However, there is no consensus on the complexity status of this problem -- $\MCSP$ could be in $\class{P}$, $\class{NP}$-complete, or $\class{NP}$-intermediate. Several works~\cite{murray2017non,KC00} showed 
negative evidence for proving the $\class{NP}$-hardness of $\MCSP$ using standard reduction techniques. 
We also do not know whether there is an algorithm better than brute force search (see Perebor conjecture for $\MCSP$~\cite{trakhtenbrot1984survey}) or whether there is a search-to-decision reduction or a self-reduction\footnote{Roughly, a problem is self-reducible if one can solve the problem with size $n$ by algorithms for smaller size.} for $\MCSP$\footnote{It is worth noting that every $\class{NP}$-complete problem has search-to-decision reductions and self-reductions.}. On the other hand, several variants of $\MCSP$ are $\class{NP}$-hard under either deterministic reductions~\cite{masek79, hirahara2018np} or randomized reductions~\cite{Ilango2019AC0pLB, ILO20}. 



Researchers have discovered many surprising connections of $\MCSP$ to other fields in Theoretical Computer Science including cryptography, learning theory, and circuit lower bounds. To name a few, Razborov and Rudich~\cite{RR97} related natural properties against $\class{P/poly}$ 
with circuit lower bounds and pseudorandomness. 
Kabanets and Cai~\cite{KC00} showed that $\MCSP\in\class{P}$ implies new circuit lower bounds, and that $\MCSP\in\class{BPP}$ implies that any one-way function can be inverted. Allender and Das~\cite{AB14} related the complexity class $\class{SZK}$ (Statistical Zero Knowledge) to $\MCSP$.
Carmosino et al.~\cite{CIKK16} showed that $\MCSP\in\class{BPP}$ gives efficient PAC-learning algorithms.
Impagliazzo et al.~\cite{impagliazzo2018power} showed that the existence of indistinguishable obfuscation implies that $\class{SAT}$ reduces to $\MCSP$ under a randomized reduction.
Hirahara~\cite{hirahara2018non} showed that if an approximation version of $\MCSP$ is $\class{NP}$-hard, then the average-case and worst-case hardness of $\class{NP}$ are equivalent.
Arunachalam et al.~\cite{agg20} proved that $\MCSP\in\class{BQP}$ implies new circuit lower bounds.
All these results indicate that the $\class{MCSP}$ serves as a ``hub'' that connects many fundamental problems in different fields. Therefore, a deeper understanding of this problem could lead to significant progress in Theoretical Computer Science.

\subsection{Main results and technical overview}

In this work, we consider three different natural objects that a quantum circuit can compute: Boolean functions, unitaries, and quantum states. We start with giving the informal definitions of the minimum circuit size problem for each of them. See Section~\ref{sec:main-QMCSP} and Section~\ref{sec:quantum obj} for the formal definitions.

\begin{definition}[$\MQCSP$, informal]
Given the truth table of a Boolean function $f$ and a size parameter $s$ in unary, decide if there exists a quantum circuit $C$ which has size at most $s$ and uses at most $s$ ancilla qubits such that $C$ computes $f$ with high probability. 
\end{definition}

\begin{definition}[$\UMCSP$, informal]
Given the full description of a $2^n$-dimensional unitary matrix $U$ and a size parameter $s$ in unary, decide if there exists a quantum circuit $C$ which has size at most $s$ and uses at most $s$ ancilla qubits such that $C$ and $U$ are close\footnote{We say $C$ and $U$ are close if $|(\bra{\psi}\otimes I)U^{\dag}C(\ket{\psi}\ket{0})|$ is large for all $\ket{\psi}$. }.
\end{definition}

\begin{definition}[$\SMCSP$, informal]
\label{def:smcsp_informal}
Let $\ket{\psi}$ be an $n$-qubit state. Given size parameters $s$ and $n$ in unary and access to arbitrarily many copies of $\ket{\psi}$ (or the classical description of $\ket{\psi}$), decide if there exists a quantum circuit $C$ which has size at most $s$ using at most $s$ ancilla qubits such that $C\ket{0^n}$ and $\ket{\psi}$ are close in terms of fidelity.
\end{definition}

In the rest of this subsection, we first discuss several challenges and difficulties we encountered in the study of MCSP when moving from the classical setting to the quantum setting. Next, we give an overview of all the results and techniques. In particular, we focus on both interpreting the new connections we establish as well as the technical subtleties when quantizing the previous works in the classical setting. For a quick summary of the results, please take a look at Table~\ref{tab:UMCSP SMCSP}.


\subsubsection{Challenges and difficulties when moving to the quantum setting}
In the following, we summarize several fundamental properties of quantum circuits, unitaries, and quantum states that induce problems and difficulties that would not appear in the classical setting. 

\paragraph{Quantum computation is generally random and erroneous.} It is natural to consider quantum circuits that approximate (rather than exactly computing) the desired unitary. One immediate consequence is that we have to define the quantum $\MCSP$s as promise problems (with respect to the error)\footnote{The definitions above are not promise problems for simplicity. Check Section~\ref{sec:main-QMCSP} and~\ref{sec:quantum obj} for formal definitions.}, which is more challenging to deal with. Moreover, since unitaries and quantum states are specified by complex numbers, we also need to properly tackle the precision issue. These quantum properties make generalizing classical results to the quantum setting non-trivial. For instance, some classical analyses (see~\cite{agg20} for an example) rely on the fact that the classical circuits are deterministic after the random string is made public, while any intermediate computation of a quantum circuit is inherently not deterministic. 
\vspace{-2mm}
    
\vspace{-2mm}    
\paragraph{Quantum circuits are reversible.} 
This follows from the fact that every quantum gate is reversible. While this seems to be a restriction for quantum circuits, we observe that this enables search-to-decision reductions for $\UMCSP$ and $\SMCSP$. Note that the existence of such reduction is a longstanding open question for classical $\MCSP$. This suggests that quantum $\MCSP$s can provide a new angle to leverage the reversibility of quantum circuits.

\vspace{-2mm}
\paragraph{The introduction of ancilla qubits.}
As quantum circuits are reversible, every intermediate computation has to happen on the input qubits. Thus, it is very common to introduce \textit{ancilla qubits} which are extra qubits initialized to all zero and can be regarded as additional registers for intermediate computation. Ancilla qubits introduce complications in quantum $\MCSP$s. First, the quantum circuit complexity of an object could be very different when the allowed number of ancilla qubits is different. Second, the classical simulation time of a quantum circuit scales exponentially in the number of input qubits plus the number of ancilla qubits. Namely, when the number of ancilla qubits is super-linear, classical simulations would require super-polynomial time\footnote{The running time is measured with respect to the size of the truth table or the size of the unitary/quantum state.}. An immediate consequence is that, unlike classical $\MCSP$, $\MQCSP$ is not trivially in $\class{NP}$ when allowing a super-linear number of ancilla qubits. In addition, the output of quantum circuits on ancilla qubits can be arbitrary quantum states in general. This property makes certain reductions for quantum $\MCSP$s fail when considering many ancilla qubits.


\vspace{-2mm}

\paragraph{Various universal quantum gate sets.} The choice of the gate set affects the circuit complexity of the given Boolean functions (and unitaries and states). There are various universal quantum gate sets, and transforming from one to the other results in additional polylogarithmic overhead to the circuit complexity by the Solovay-Kitaev Theorem. We note that when considering certain hardness results, the choice of the gate set might matter. Take the approximate self-reduction for $\SMCSP$ (in Theorem~\ref{thm:informal_reduction}) as an example, we start from constructing such reductions for a particular gate set. We then generalize the result to an arbitrary gate set via the Solovay-Kitaev Theorem; however, it introduces additional overhead to the approximation ratio. Another example is proving NP-hardness for multi-output $\MQCSP$, where we show that the problem is $\class{NP}$-hard when considering particular gate sets, and it is still open whether the problem is $\class{NP}$-hard for all universal gate sets. 
\vspace{-2mm}

\subsubsection{The Hardness of \texorpdfstring{$\MQCSP$}{MQCSP} and cryptography} 
\label{sec:techintro_mqcsp_crypto}

We start with stating the hardness results of $\MQCSP$ and its implications in cryptography. 
\begin{theorem}[Informal]
\label{thm:informal_mqcsp_crypto}
\mbox{}
\begin{enumerate}
    \item $\MQCSP$ is in $\class{QCMA}\subseteq \class{QMA}$. 
    \item If $\class{MQCSP}$ can be solved in quantum polynomial time, then quantum-secure one-way function ($\class{qOWF}$) does not exist.
    \item If one can solve $\MQCSP$ efficiently, then all problems in $\class{SZK}$ have efficient algorithms.    
    \item Suppose that quantum-secure indistinguishability obfuscator ($\iO$) for polynomial-size circuits exists. Then, $\class{MQCSP}\in \class{BQP}$ implies $\class{NP}\subseteq \class{coRQP}$\footnote{$
    \class{coRQP}$ is a complexity class of quantumly solvable problems with perfect soundness and bounded-error completeness.}. 
    \item Multiple-output $\MQCSP$ (under a gate set with some natural properties) is $\class{NP}$-hard under randomized reductions.
\end{enumerate}

\end{theorem}


We have discussed why $\MQCSP$ is not trivially in $\class{NP}$ earlier.  So, it is natural to wonder what can be a tighter upper bound for $\MQCSP$. Instead of considering classical verifier, we allow the verifier to check the given witness circuit quantumly and thus are able to prove that $\MQCSP$ is in $\class{QCMA}$ (which is a quantum analogue of $\class{MA}$ allowing efficient quantum verifiers but classical witness).

For item $2$ -- $5$, we study whether some hard problems reduce to $\MQCSP$. Classically, many results use the fact that an $\MCSP$ oracle can break certain \emph{pseudorandom generators} to show reductions from hard problems to $\MCSP$.  A distinguisher can break a pseudorandom generator by viewing that the string is a truth table of some Boolean function and using the $\MCSP$ oracle to decide if the function has small circuit complexity\footnote{If the truth table is truly random, it corresponds to a random function and must have large circuit complexity with high probability. }. We generalize this idea to the quantum setting by observing that if the Boolean function has small classical circuit complexity, then its quantum circuit complexity is also small. It is worth noting that the second result implies efficient algorithms for some lattice problems if $\MQCSP$ is in $\class{BQP}$. 


For item $5$, we generalize the recent breakthrough of Ilango et al.~\cite{ILO20} on the $\class{NP}$-hardness of $\MCSP$. We note that the formal theorem statement depends on the gate set choices of $\MQCSP$. To prove this theorem, we follow the proof ideas in ~\cite{ILO20} and overcome some additional obstacles that appear in the quantum world. The new obstacle comes from (i) the quantum gate set is different from the one in the classical case; (ii) in the quantum world, we need to deal with error terms. We carefully handle these issues and extend the proof to the quantum setting.

\subsubsection{\texorpdfstring{$\MQCSP$}{MQCSP} and learning theory}
A central learning theory setting is (approximately) reconstructing a circuit for an unknown function given a limited number of samples. Learning Boolean functions in the classical setting was extensively studied (see, for example, a survey by Hellerstein and Servedio~\cite{hellerstein2007pac}); however, relatively few explorations have been made under the quantum setting. There are two natural quantum extensions: (i) learning a quantum circuit and (ii) adding quantumness in the learning algorithm. We study both scenarios and provide generic connections between $\MQCSP$ and the two settings
\vspace{-2mm}
\paragraph{PAC learning for quantum circuits.}
Probabilistic approximately correct (PAC) learning~\cite{valiant1984theory} is a standard theoretical framework in learning theory. There are several variants, but for simplicity, we focus on the query model where a classical learning algorithm can query an unknown $n$-bit Boolean function $f$ on inputs $x_1,\dots,x_m\in\{0,1\}^n$ and aim to output a circuit approximating $f$ with high probability. To have efficient PAC learning algorithms for polynomial-size quantum circuits, we show that it is necessary and sufficient to have efficient algorithms for $\MQCSP$ or its variants.

\begin{theorem}[Informal]\label{thm:pac_informal}
The existence of an efficient PAC learning algorithm for $\class{BQP/poly}$ is equivalent to the existence of an efficient randomized algorithm for $\MQCSP$.
\end{theorem}


\vspace{-2mm}
\paragraph{Quantum learning.}
In the past two decades, there has been increased interest in quantum learning (see a survey by Arunachalam and de Wolf~\cite{arunachalam2017guest}) due to the success of machine learning and quantum computing. While there have been interesting quantum speed-ups for specific learning problems such as Principal Component Analysis~\cite{lloyd2014quantum} and quantum recommendation system~\cite{kerenidis2017quantum}, it is unclear whether the quantumness can provide a generic speed-up in learning theory. A recent result of Arunachalam et al.~\cite{agg20} suggested that this might be difficult by showing that the existence of efficient quantum learning algorithms for a circuit class would imply a breakthrough circuit lower bound. We further generalize their result by showing the equivalence of efficient quantum PAC learning and the non-trivial upper bound for $\MQCSP$.

\begin{theorem}[Informal]\label{thm:learning_informal}
The existence of efficient quantum learning algorithms for PAC learning a circuit class $\class{C}$ is equivalent to the existence of efficient quantum algorithms for $\class{C}$-$\MQCSP$\footnote{$\class{C}$-$\MQCSP$ is $\MQCSP$ with respect to circuit class $\class{C}$.}.
\end{theorem}

The proof idea is to quantize the ``learning from a natural property'' paradigm of~\cite{CIKK16}. Briefly speaking, the converse direction ``algorithms for $\MQCSP$ imply learning algorithms'' follows from the idea that one can use the Boolean function (the object to be learned) to construct a PRG with the property that breaking the PRG implies a reconstructing algorithm for $f$. Then, since an algorithm for $\MQCSP$ can break PRG, we obtain an algorithm for $f$. Another direction follows from the observation that we can still apply the learning algorithm given the truth table of the function. Specifically, for Theorem~\ref{thm:pac_informal}, it turns out that the converse direction is straightforward because $\class{P/poly}\subset\class{BQP/poly}$ while the forward direction requires the number of ancilla bits to be $O(n)$ due to the overhead from a classical simulation for quantum circuits. For Theorem~\ref{thm:learning_informal}, the difficulty lies in the fact that a quantum circuit is \textit{inherently random} and one cannot arbitrarily compose quantum circuits as their wishes. To circumvent these issues, we invoke the techniques in~\cite{agg20} which built up composable tools for \textit{reconstructing} a circuit from a quantum distinguisher. See Theorem~\ref{thm:quantum learning MCSP}, Theorem~\ref{thm:PAC learning and MQCSP}, and Section~\ref{sec:main learning} for more details.

\subsubsection{\texorpdfstring{$\MQCSP$}{MQCSP} and quantum circuit lower bounds}

The classical $\MCSP$ is tightly connected to circuit lower bounds. We generalize the results of Oliveira and Santhanam~\cite{os16}, Arunachalam et al.~\cite{agg20}, and Kabanets and Cai~\cite{KC00} to $\MQCSP$.

\begin{theorem}[Informal]\label{thm:informal_circuit_lb}
Suppose that $\class{MQCSP}\in \class{BQP}$. Then 
\begin{enumerate}
    \item $\class{BQE}\not\subset\class{BQC}[n^k]$ for any constant $k\in\mathbb{N}$\footnote{$\class{BQC}[n^k]$ is the complexity class for problems that can be solved by $O(n^k)$-size quantum circuits with bounded fan-in, and $\class{BQE}$ in the set of problems that can be solved in $2^{O(n)}$ time by quantum computers. Previously, Aaronson \cite{aar06} showed that $\mathsf{P}^{\mathsf{PP}}\not\subset \class{BQC}[n^k]$ unconditionally. However, the relations between $\class{P^{PP}}$, $\class{BQE}$, and $\class{BQP^{QCMA}}$ are still unclear.
    }; and
    \item $\class{BQP^{QCMA}}\not\subset\class{BQC}[n^k]$ for any constant $k\in\mathbb{N}$.
\end{enumerate}
\end{theorem}



For item $1$, we use $\MQCSP$ to construct a $\class{BQP}$-natural property against quantum circuit classes. Then, with a quantum-secure pseudorandom generator, we can use a ``win-win argument'' to show that $\class{BQE}\not\subset\class{BQC}[n^k]$ for any $k>0$.  The proof mainly follows from \cite{agg20,os16}. However, we extend their proofs to the quantum natural properties against \emph{quantum} circuit classes. One technical contribution is a diagonalization lemma for quantum circuits.


For item $2$, we follow the idea in~\cite{KC00} to show that the maximum quantum circuit complexity problem\footnote{The problem is, given $1^n$, ask for a Boolean function $f:\{0,1\}^n\rightarrow \{0,1\}$ that has the maximum complexity.} can be solved in exponential time with a $\class{QCMA}$ oracle. The main difference from the classical case is that we require a $\class{QCMA}$ oracle instead of an $\class{NP}$ one, which follows from the fact that we assume $\MQCSP$ is in $\class{BQP}$\footnote{Along this line, the result still holds if we consider $\MCSP\in \class{BQP}$ and maximum classical circuit complexity.}. Then, the statement follows from the standard padding argument. 

Another aspect of quantum circuit complexity is \emph{hardness amplification}. Kabanets and Cai \cite{KC00} showed that $\MCSP$ can be used as an amplifier to generate many hard Boolean functions. In this part, we show that with an $\MQCSP$ oracle, given one quantum extremely hard Boolean function, there is an efficient quantum algorithm that outputs many quantum-hard functions.

\begin{theorem}[Hardness amplification by $\MQCSP$, informal]\label{thm:amp_intro}
Assume $\class{MQCSP}\in \class{BQP}$. There exists a $\class{BQP}$ algorithm that, given the truth table of a Boolean function with quantum circuit complexity $2^{\Omega(n)}$, outputs $2^{\Omega(n)}$ Boolean functions with $m=\Omega(n)$ variables such that each function has quantum circuit complexity greater than $2^m/(c+1)m$ for $c$ some constant. 
\end{theorem}

The proof of Theorem~\ref{thm:amp_intro} closely follows the proof in \cite{KC00}. The key ingredient is a quantum Impagliazzo-Wigderson generator, which ``quantizes'' the construction in \cite{iw97}. The quantum Impagliazzo-Wigderson generator can transform the given quantum extremely hard function to a quantum pseudorandom generator that fools quantum circuits of size $2^{O(n)}$. Since we assume $\MQCSP\in \class{BQP}$, it means that we can construct a small quantum distinguishing circuit to accept the truth tables of hard functions. And we can show that our quantum Impagliazzo-Wigderson generator can fool the distinguishing circuit. Hence, most of the outputs of the quantum pseudorandom generator will have high quantum circuit complexity. 

To quantize the Impagliazzo-Wigderson generator, we construct a quantum-secure direct-product generator, and also use the quantum  Goldreich-Levin Theorem and quantum-secure Nisan-Wigderson generator developed in \cite{agg20}.

\emph{Hardness magnification} is an interesting phenomenon in classical circuit complexity defined by \cite{os18}. It shows that a weak worst-case lower bound can be ``magnified'' into a strong worst-case lower bound for another problem. (See a recent talk by Oliveira~\cite{olinote}.) In this part, we show that $\MQCSP$ also has a quantum hardness magnification.

\begin{theorem}[Hardness magnification for $\MQCSP$, informal]\label{thm:magnify_informal}
If a gap version of $\MQCSP$ does not have nearly-linear size quantum circuit, then $\class{QCMA}$ cannot be computed by polynomial size quantum circuits.
\end{theorem}

We note that this is a nontrivial theorem because even if we assume $\class{QCMA}\subseteq \class{BQC}[\poly(n)]$, we can only show $\MQCSP\in \class{BQC}[\poly(2^n)]$, i.e., $\MQCSP$ has a polynomial-size quantum circuit by the fact that $\MQCSP \in \class{QCMA}$. But the theorem implies that some gap-version of $\class{MQCSP}$ has nearly-linear size circuit! 

We prove the above theorem via a quantum antichecker lemma, whose classical version was given by \cite{ops19,chopr20}. 
And we observe that the two key ingredients: a delicate design of a Boolean circuit and a counting argument can be quantized.


\subsubsection{\texorpdfstring{$\MQCSP$}{MQCSP} and quantum fine-grained complexity}
Fine-grained complexity theory aims to study the \emph{exact} lower/upper bounds of some problems. For example, most theorists believe 3-SAT is not in $\class{P}$, but we do not know if it can be solved in $2^{o(n)}$ time. Exponential Time Hypothesis ($\class{ETH}$) is a commonly used conjecture in this area which rules out this possibility (see a survey by Williams~\cite{wvv18}).
Very recently, \cite{ila20} showed the fine-grained hardness of $\MCSP$ for partial function based on $\class{ETH}$. In the quantum setting, \cite{aclwz20,bps21} proposed quantum fine-grained reductions and quantum strong exponential time hypothesis ($\class{QSETH}$) to study the quantum hardness of problems in $\class{BQP}$. In this part, we follow the works of \cite{ila20,aclwz20} and prove the quantum hardness of $\MQCSP$ for partial functions based on the quantum $\class{ETH}$ conjecture,
which conjectures that there does not exist a $2^{o(n)}$-time quantum algorithm for solving 3-SAT\footnote{Existing quantum SAT solvers are not much faster than Grover's search; they need $2^{\Omega(n)}$-time even for 3-SAT.}. 
 \begin{theorem}[Fine-grained hardness of $\MQCSP^\star$, informal]\label{thm:mqcsps_informal}
Quantum \textsf{ETH} implies $N^{o(\log \log N)}$-quantum hardness of $\MQCSP$ for partial functions.
\end{theorem}

To prove the above theorem, we basically follow the reduction path in \cite{ila20}, which gave a reduction from a fine-grained problem studied by \cite{lms11} to $\MQCSP$ for partial functions. But we need to bypass two subtleties:
\begin{itemize}
    \item The proof of \cite{ila20} relies on the structure of the classical read-once formula, but there is no direct correspondence with quantum;
    \item \cite{lms11} only proved the classical hardness of the bipartite permutation independent set problem, but we need quantum hardness result. 
\end{itemize}


For the first issue, we prove an unconditional quantum circuit lower bound for that function in the reduction. More specifically, we first show that if a small quantum circuit can compute the partial function $\gamma$ in the reduction, then that circuit is a quantum read-once formula (defined by \cite{yao93}); and vice versa. And then, we apply a ``dequantization'' result by \cite{ckp13} to show that the quantum read-once formula can be converted to a classical read-once formula with the same size. Then, by the structure of the ``dequantized'' read-once formula, we finally conclude that deciding $\MQCSP$ for $\gamma$ is equivalent to solving the bipartite permutation independent set problem.

For the second issue, we use the quantum fine-grained reduction framework and give a reduction from 3-SAT to the bipartite permutation independent set problem. Therefore, the quantum hardness of $\MQCSP$ for partial function follows from the quantum hardness of deciding 3-SAT conjectured by the quantum $\class{ETH}$.

\subsubsection{Quantum circuit complexity for states and unitaries}
In this section, we study $\UMCSP$ and $\SMCSP$. For $\SMCSP$ in Definition~\ref{def:smcsp_informal}, we consider two types of inputs: quantum states and the classical description of the state. We consider the inputs as quantum states since we generally cannot have the classical description of the quantum state in the real world, and many related problems (such as shadow tomography~\cite{aaronson18}, quantum gravity~\cite{BFV19}, and quantum pseudorandom state~\cite{JLF18}) have multiple copies of states as inputs. Although this input format makes $\SMCSP$ harder, we are able to show that $\SMCSP$ has a $\class{QCMA}$ protocol\footnote{Note that since $\SMCSP$ has quantum inputs, the problem is not in $\class{QCMA}$ under the standard definition.}. Furthermore, the search-to-decision reduction and the self-reduction in Theorem~\ref{thm:informal_reduction} hold for both versions of $\SMCSP$. We first show hardness upper bounds for $\UMCSP$ and $\SMCSP$.  
\begin{theorem}[Informal]\label{thm:smcspumcspqcma}
\emph{(1)} $\UMCSP\in \class{QCMA}$. \emph{(2)} $\SMCSP$ can be verified by $\class{QCMA}$ protocols.
\end{theorem}

To prove Theorem~\ref{thm:smcspumcspqcma}, we use the \emph{swap test} to test whether the witness circuit $C$ outputs the correct states. This suffices to show that $\SMCSP$ has a $\class{QCMA}$ protocol. To show that $\UMCSP$ is in $\class{QCMA}$, checking if the circuit $C$ and $U$ agree on all inputs by using swap test is infeasible since there are infinitely many quantum states in the $2^n$-dimensional Hilbert space. If one only checked all the computational basis states (i.e., $\{\ket{x}:\, x\in \{0,1\}^n\}$), it is possible that the circuit $C$ and the given unitary $U$ are not close on inputs in the form of superposition states. This can come from the following two sources. (a) $C$ can introduce different phases on different computational basis states; (b) using ancilla qubits to implement $U$ results in entanglement between the output qubits and ancilla qubits, which may fail the swap test.

To deal with these difficulties, we introduce an additional step in the test called ``coherency test''. This step tests the circuit output on all the initial states in the form of $\ket{a}+\ket{b}$, where $\ket{a},\ket{b}$ are different computational basis states. We can prove that it forces the behavior of $C$ to be coherent on all the computational basis states, and forces the phases to be roughly the same.\par
\vspace{-2mm}
\paragraph{Reductions for $\UMCSP$ and $\SMCSP$ that are unknown to the classical $\MCSP$.}
In addition to the upper bounds, we also show interesting reductions for $\UMCSP$ and $\SMCSP$. 

\begin{theorem}[Informal]
\label{thm:informal_reduction}
\mbox{}
\begin{itemize}
    \item \textbf{Search-to-decision reductions:} There exist search-to-decision reductions for $\UMCSP$ and $\SMCSP$ when no ancilla qubits are allowed.
    \item \textbf{Self-reduction:} $\SMCSP$ is approximately self-reducible. 
    \item A gap version of $\MQCSP$ reduces to $\UMCSP$. 
\end{itemize}

\end{theorem}

Classically, it is unknown whether $\MCSP$ is self-reducible or has search-to-decision reductions. Ilango~\cite{ilango20_CCC} proved that some variants of $\MCSP$ have search-to-decision reductions. 
Recently, Ren and Santhanam \cite{rs21} showed that a relativization barrier applies to the deterministic search-to-decision reduction and self-reduction of \textsf{MCSP}. We prove the existence of search-to-decision reductions by using the property that ``\emph{quantum circuits are reversible}''. In particular, we guess the $i$-th gate, uncompute the gate from the state or the unitary, and use the decision oracles to check whether the complexity of the new state or the new unitary reduces. By repeating this process for all gates, we can find the desired circuits. This approach suffices for the case where the quantum circuits use no ancilla qubits. On the other hand, when the quantum circuits use ancilla qubits and are not forced to turn ancilla qubits back to the all-zero state, this approach does not work. Consider $\UMCSP$. The quantum circuit may implement a unitary $U\otimes V$. To find the circuit, the approach above needs to start from $U\otimes V$ and do the uncomputation iteratively. However, $V$ is unknown. $\SMCSP$ has the similar issues.

For the self-reducibility of $\SMCSP$, we show that one can approximate the circuit complexity of an $n$-qubit state by computing the circuit complexities of ($n-1$)-qubit states. Roughly, we find a ``win-win decomposition'' of an $n$-qubit state such that its circuit complexity is either close to the circuit complexity of an ($n-1$)-qubit state or can be approximated by two ($n-1$)-qubit states.


Finally, we show a reduction related to $\class{MQCSP}$ and $\class{UMCSP}$. The proof is by encoding a Boolean function into a particular unitary and showing that the circuit complexity of that unitary gives both upper and lower bounds for the circuit complexity of the Boolean function.  


\paragraph{Implications of Hardness of \texorpdfstring{$\SMCSP$}{SMCSP} and \texorpdfstring{$\UMCSP$}{UMCSP}} 

For $\UMCSP$, one application is related to a question Aaronson asked in~\cite{aaronson2016complexity}: does there exist an efficient quantum process that generates a family of unitaries that are indistinguishable from random unitaries given the full description of the unitary? If there is an efficient algorithm for $\UMCSP$, then there is no efficient quantum process that generates unitaries indistinguishable from random unitaries given the full unitary.

Moreover, several implications of $\MCSP$ carry to $\UMCSP$ by Theorem~\ref{thm:informal_reduction}. This follows from the fact that the gap version of $\MQCSP$ suffices to break certain pseudorandom generators.


For $\SMCSP$, we focus on the version where the inputs are copies of quantum states and present its relationships to quantum cryptography, tomography, and quantum gravity.
\begin{theorem}[Informal]\label{thm:smcsp_app_informal}
\mbox{}
\begin{enumerate}
    \item If $\SMCSP$ has quantum polynomial-time algorithms, then there are no pseudorandom states, and thus no quantum-secure one-way functions.
    \item Assuming additional conjectures from physics and complexity theory, the existence of an efficient algorithm for $\SMCSP$ implies the existence of an efficient algorithm for estimating the wormhole's volume
    \item If $\SMCSP$ can be solved efficiently, then one can solve the succinct state tomography problem\footnote{The succinct state tomography problem is that given many copies of a state with the promise that its circuit complexity is at most certain $s$, the problem is to find a circuit that computes the state.} in quantum polynomial time. 
\end{enumerate}
\end{theorem}
The first result in Theorem~\ref{thm:smcsp_app_informal} follows from the observation that we can use $\SMCSP$ algorithms to distinguish whether the given states have large circuit complexities. This results in algorithms for breaking pseudorandom states, and thus algorithms for inverting quantum-secure one-way functions by~\cite{JLF18}. It is worth noting that a recent work by Kretschmer~\cite{kretschmer2021} showed some relativized results for the problem of breaking pseudorandom states. Since that problem reduces to $\class{SMCSP}$, his results would provide another angle for understanding the hardness of $\class{SMCSP}$. We show the second result under the model and assumptions considered in~\cite{BFV19}. Roughly speaking, the volumes of wormholes correspond to circuit complexities of particular quantum states. Thus efficient algorithms for one implies solving the other one efficiently if the correspondence can be computed efficiently. The third result mainly uses the \emph{search-to-decision reduction} in Theorem~\ref{thm:informal_reduction} to find the circuit that computes the state.   
\vspace{-2mm}

%% file: main-open.tex
\subsection{Discussion and open questions}
\label{sec:intro quantum}



We lay out the following three-aspect road map for the quantum $\MCSP$ program. For each aspect, we present several results and also propose many open directions to explore. We have also summarized all results in this work in Table~\ref{tab:UMCSP SMCSP}.    

First, we define the Minimum Quantum Circuit Size Problem ($\MQCSP$) and study upper bounds and lower bounds for its complexity. Furthermore, we explore the connections between $\MQCSP$ and other areas of quantum computing such as quantum cryptography, quantum learning, quantum circuit lower bounds, and quantum fine-grained complexity.

Then, we further extend $\MQCSP$ to study the quantum circuit complexities for quantum objects, including unitaries and states.\footnote{Aaronson has raised questions about quantum circuit complexity for unitaries or states in~\cite{aaronson2016complexity}.} We want to investigate their hardness and connections to other areas in TCS. In this work, we show upper bounds and lower bounds for their complexities, search-to-decision reductions (for $\UMCSP$ and $\SMCSP$), a self-reduction (for $\SMCSP$), and reductions from $\MQCSP$ to $\UMCSP$. In addition to connections generalized from classical analogues (such as cryptography, learning, and circuit lower bounds), we also find connections that might be unique in the quantum setting, such as tomography and quantum gravity.

For the last part, we want to turn around and ask what could happen when considering quantum algorithms or quantum reductions for $\MCSP$ (and also for $\MQCSP$, $\UMCSP$, and $\SMCSP$)? In the previous two parts, we have already observed that efficient quantum algorithms for these problems result in surprising implications to other fields. One can further consider other influences of quantum algorithms to study quantum and classical $\MCSP$s. For example, can $\class{SAT}$ reduce to $\MCSP$ under quantum reductions?



Following the three-aspect road map for the quantum $\MCSP$ program, there are many open directions to explore. In particular, we are interested to understand the hardness of these problems, the relationships between them, and their connections to other fields in computer science. 

\subsubsection{Open problems: the complexity of quantum circuits}\label{sec:intro open probs}
We start with open problems related to the hardness and relationships between quantum $\MCSP$s. The most basic questions are to understand the complexity of different quantum $\MCSP$s. As we have already seen, it is unclear if quantum $\MCSP$s are in $\class{NP}$. 
Besides, we do not know if $\class{NP}$- or $\class{QCMA}$-hard problems reduce to them.
\begin{openproblem}\label{open:MQCSP in NP}
Are $\UMCSP$, $\MQCSP$, and $\SMCSP$ in $\class{NP}$? Are these problems $\class{NP}$-hard, $\class{QCMA}$-hard, or $\class{C}$-hard for some complexity class $\class{C}$ that is between $\class{QCMA}$ and $\class{SZK}$? 
\end{openproblem} 

We note that the case that makes these problems not known to be in $\class{NP}$ is when there are more than linearly many ancilla qubits. Therefore, if one can show that adding superpolynomially many ancilla qubits does not lead to significant improvement on quantum circuit complexity, then we are likely to put these problems in $\class{NP}$ directly. Along this line, we pose the following open question: 

\begin{restatable}{openproblem}{linearancilla}\label{conj:linear ancilla}
For every $n,s,t\in\N$ with $t\leq s\leq2^{O(n)}$, is $\class{BQC}(s,t)\subset\class{BQC}(\poly(s,t),O(n))$?
\end{restatable}


For the hardness of $\UMCSP$ and $\SMCSP$, One potential approach for proving $\class{NP}$-hardness of $\UMCSP$ is as follows: Prove the $\class{NP}$-hardness of the gap version of certain variants of $\MQCSP$ (such as sparse $\MQCSP$ or $\mBMCSP$), and then reduce it to $\UMCSP$ via the last reduction in Theorem~\ref{thm:informal_reduction}. The hardness of $\SMCSP$ seems to be slightly more mysterious than $\UMCSP$. One reason for this is that we do not know any relationship between $\SMCSP$ and other quantum $\MCSP$s, and thus the approach of reducing particular variants of quantum $\MCSP$ to $\SMCSP$ does not directly work. This leads to another important open question: 
\begin{openproblem}
What are the relationships between $\UMCSP$, $\MQCSP$, and $\SMCSP$?  
\end{openproblem}

To answer whether quantum $\MCSP$s are $\class{NP}$-complete, we can also study these problems from another angle, that is, check if quantum $\MCSP$s have particular reductions that all $\class{NP}$-complete problems have. In the previous section, we observed that quantum circuits have some properties leading to search-to-decision reductions for $\UMCSP$ and $\SMCSP$ without ancilla qubits and an approximate self-reduction for $\SMCSP$. Therefore, we ask whether we can have search-to-decision reductions and self-reductions for these quantum $\MCSP$s. 
\begin{openproblem}
Are there search-to-decision reductions and self-reductions for quantum $\MCSP$s?
\end{openproblem}
It is worth noting that our search-to-decision reductions fail when ancilla qubits are allowed. This mainly follows from the fact that the circuit of the solution can be an non-identity operator on the ancilla qubits in general. This could possibly be addressed by iterating all possible unitaries or states on an $\epsilon$-net when the number of ancilla qubits are not large (e.g., at most $\log\log n$). However, we need new ideas when considering more ancilla qubits.

Moreover, it would be interesting to investigate the applications of these reductions. For instance, we have seen that the search-to-decision reductions give algorithms with $\UMCSP$ or $\SMCSP$ oracle additional power to obtain the circuits. This power may lead to interesting applications. 
\begin{openproblem}
Is there any application of search-to-decision reductions or self-reductions for quantum $\MCSP$s?
\end{openproblem}

The hardness of average-case quantum $\MCSP$s (which inputs are given randomly) is another interesting topic to explore. Hirahara~\cite{hirahara2018non} showed that there is a worst-case to average-case reduction for the (gap version of) classical $\MCSP$. We wonder if we can prove that quantum $\MCSP$s have worst-case to average-case reductions. 
\begin{openproblem}
\label{open:w2a}
Are there worst-case to average-case reductions for quantum $\MCSP$s?
\end{openproblem}
Note that there is negative evidence~\cite{BT06} showing that such classical reductions might not exist for $\class{NP}$-complete problems\footnote{However, there is no evidence for the existence of quantum worst-case to average-case reductions for $\class{NP}$-complete since the analysis in~\cite{BT06} fails in the quantum setting. See~\cite{Chia2020basingoneway} for related discussion. }. The existence of such reduction could result in important applications in cryptography, which we will discuss later.

Finally, we can also try to prove the hardness of quantum $\MCSP$s under stronger assumptions or more powerful reductions. 
\begin{openproblem}
Assuming $\class{QETH}$ or $\class{QSETH}$, is $\MQCSP$, $\UMCSP$, or $\SMCSP$ quantumly hard?
\end{openproblem}

\begin{openproblem}
Does quantum reduction provide more power to show the hardness of $\MCSP$? Specifically, is $\class{NP}\subseteq\class{BQP}^{\MCSP}$ or $\class{NP}\subseteq\class{BQP}^{\MQCSP}$?
\end{openproblem}

\vspace{-4mm}
\subsubsection{Open problems: potential connections to other areas}

In this work, in addition to generalizing several known connections for $\MCSP$ to quantum $\MCSP$s, we have also discovered several connections which could be unique for quantum $\MCSP$s. There are still many classically existing or unknown connections that we can explore. One fascinating question is whether we can base the security of one-way functions on any of these problems. 
\begin{openproblem}
\label{open:mcsp2owf}
Can we base the security of cryptographic primitives on $\MQCSP$, $\UMCSP$, $\SMCSP$, or some variants of these problems? 
\end{openproblem}
Note that since quantum $\MCSP$s considered in this work are all worst-case problems, to answer Problem~\ref{open:mcsp2owf}, we probably need worst-case to average-case reductions discussed in Problem~\ref{open:w2a}. Moreover, Liu and Pass~\cite{LP20} recently showed that the existence of classical one-way function is equivalent to the average-case hardness of a type of Kolmogorov complexity on uniform distribution. However, the average-case hardness of $\MCSP$ on uniform distribution is not known to imply one-wayness even classically, and the quantum version faces a similar obstacle. Very recently, Ilango, Ren, and Santhanam \cite{irs21} showed that the average-case hardness of \textsf{Gap-MCSP} on a locally samplable distribution is equivalent to the existence of one-way function.  Liu and Pass \cite{lp21} further generalized this result to show equivalence between the existence of one-way functions and the existence of sparse languages that are hard-on-average (including Kolmogorov complexity, $k$-\textsf{SAT}, and $t$-\textsf{Clique}). It is natural to ask whether their results can be generalized to quantum $\MCSP$s. In addition to one-way functions, We are interested in connections between quantum $\MCSP$s and ``quantum-only'' primitives, e.g., quantum $\iO$, copy protection, quantum process learning, etc. 

Along this line, as many quantum problems have quantum inputs, it is natural to consider quantum $\MCSP$s with quantum inputs. We have shown how $\SMCSP$ connects to problems in quantum cryptography, quantum gravity, and tomography given quantum states as inputs. This fact gives the possibility that $\MQCSP$, $\UMCSP$, and $\SMCSP$ with ``succinct'' quantum or classical inputs may have surprising connections to other problems in quantum computing.  For instance, one can consider inputs which are quantum circuits that encode some objects (e.g., unitaries). Then, the problem is to find another significantly smaller circuit. In~\cite{CCCW21}, Chakrabarti et al. have studied this problem and show applications to quantum supremacy.

\input{table-3}

%% file: table-3.tex
\definecolor{LightCyan}{rgb}{0.8,0.9,1}
\definecolor{LightYellow}{rgb}{1,1,0.8}
\definecolor{LightRed}{rgb}{1,0.8,0.8}

\begin{table}[H]
\centering
\begin{tabular}{|c|c|c|}
\hline
                          & Results                                                                                                                    & \makecell{Informal Theorem Index\\ (Formal Theorem Index)}                          \\ \hline
 \multirow{12}{*}{\makecell{$\MQCSP$\\(Def.~\ref{def:main_bmcsp})}} & \cellcolor{LightCyan} $\MQCSP\in\class{QCMA}$                                                                                                     & \cellcolor{LightCyan} Theorem~\ref{thm:informal_mqcsp_crypto}  (Theorem~\ref{thm:QCMA_MCSP})                           \\ \cline{2-3} 
& \cellcolor{LightCyan} $\MQCSP\in\class{BQP}$ $\Rightarrow$ No \textsf{qOWF}                                                                       & \cellcolor{LightCyan} Theorem~\ref{thm:informal_mqcsp_crypto}  (Theorem~\ref{thm:mqcsp_qowf})                           \\ \cline{2-3}
& \cellcolor{LightCyan} $\class{SZK}\leq \MQCSP$                                                            & \cellcolor{LightCyan} Theorem~\ref{thm:informal_mqcsp_crypto}  (Theorem~\ref{thm:szk_mcsp})                           \\ \cline{2-3}
& \cellcolor{LightCyan} $\mBMCSP$ is $\class{NP}$-hard    under a natural gate set                                                                 & \cellcolor{LightCyan} Theorem~\ref{thm:informal_mqcsp_crypto}  (Theorem~\ref{thm:mqcsp_multi hard})                           \\ \cline{2-3}
& \cellcolor{LightCyan} $i{\cal O}+\MQCSP\in\class{BQP}$ $\Rightarrow$ $\class{NP}\subseteq\class{coRQP}$   & \cellcolor{LightCyan} Theorem~\ref{thm:informal_mqcsp_crypto}  (Theorem~\ref{thm:qioimplication})                           \\ \cline{2-3}
& \cellcolor{LightYellow} PAC learning for $\class{BQP/poly}$ $\Leftrightarrow$ $\MQCSP\in\class{BPP}$                                                & \cellcolor{LightYellow} Theorem~\ref{thm:pac_informal}  (Theorem~\ref{thm:PAC learning and MQCSP})                           \\ \cline{2-3}
& \cellcolor{LightYellow} $\class{BQP}$ learning $\Leftrightarrow$ $\MQCSP\in \class{BQP}$                                                           & \cellcolor{LightYellow} Theorem~\ref{thm:learning_informal}  (Theorem~\ref{thm:quantum learning MCSP})                           \\ \cline{2-3}
& \cellcolor{LightYellow} $\MQCSP\in\class{BQP}\Rightarrow \class{BQE}\not\subset\class{BQC}[n^k],\ \forall k\in\mathbb{N}_+$ & \cellcolor{LightYellow} Theorem~\ref{thm:informal_circuit_lb}  (Theorem~\ref{thm:ckt lb from MQCSP in BQP max})                           \\ \cline{2-3}
& \cellcolor{LightCyan} $\MQCSP\in\class{BQP}\Rightarrow \class{BQP^{QCMA}}\not\subset\class{BQC}[n^k],\ \forall k\in\mathbb{N}_+$ & \cellcolor{LightCyan} Theorem~\ref{thm:informal_circuit_lb}  (Theorem~\ref{thm:magnification})                           \\ \cline{2-3}
& \cellcolor{LightYellow} $\MQCSP\in\class{BQP}$ $\Rightarrow$ Hardness amplification                                                                 & \cellcolor{LightYellow} Theorem~\ref{thm:amp_intro}  (Theorem~\ref{thm:mqcsp_hardness_amp})                           \\ \cline{2-3}
& \cellcolor{LightCyan} Hardness magnification for $\class{MQCSP}$                                                                & \cellcolor{LightCyan} Theorem~\ref{thm:magnify_informal}  (Theorem~\ref{thm:magnification})                           \\ \cline{2-3}
& \cellcolor{LightYellow} $\class{QETH}\Rightarrow$ quantum hardness of $\MQCSP^\star$                                   & \cellcolor{LightYellow} Theorem~\ref{thm:mqcsps_informal}  (Theorem~\ref{thm:qeth_hard_mqcsps})                           \\ \hline                          
                          
\multirow{8}{*}{\makecell{$\UMCSP$\\(Def.~\ref{def:umcsp})}} & \cellcolor{LightRed} $\UMCSP\in\class{QCMA}$                                                                                                    & \cellcolor{LightRed} Theorem~\ref{thm:smcspumcspqcma}  (Theorem~\ref{thm:umcspqcma})                                  \\ \cline{2-3} 
                          & \cellcolor{LightRed} Search-to-decision reduction for $\UMCSP$                                                                      & \cellcolor{LightRed} Theorem~\ref{thm:informal_reduction}  (Theorem~\ref{thm:search2decision_umcsp})                              \\ \cline{2-3} 
                          & \cellcolor{LightRed} $\textsf{gap}$-$\MQCSP\leq\UMCSP$                                                                                          & \cellcolor{LightRed} Theorem~\ref{thm:informal_reduction}  (Theorem~\ref{thm:b2u})                              \\ \cline{2-3} 
                          & \cellcolor{LightYellow} $\UMCSP\in\class{BQP}$                          &  \cellcolor{LightYellow} \\
                          & \cellcolor{LightYellow} $\Rightarrow$ No pseudorandom unitaries and no \textsf{qOWF} & \multirow{-2}{*}{\cellcolor{LightYellow} (Theorem~\ref{thm:pru}, Corollary~\ref{cor:umcsp_owf})} \\ \cline{2-3} 
                          & \cellcolor{LightYellow} $i\mathcal{O}+\UMCSP\in\class{BQP}$ $\Rightarrow$ $\class{NP}\subseteq\class{coRQP}$ & \cellcolor{LightYellow} (Corollary~\ref{cor:umcsp_io})                                                      \\ \cline{2-3} 
                          & \cellcolor{LightYellow} $\UMCSP\in\class{BQP}$ $\Rightarrow$ Hardness amplification for $\class{BQP}$                                              & \cellcolor{LightYellow} (Corollary~\ref{cor:umcsp_amp})                                                     \\ \cline{2-3} 
                          & \cellcolor{LightYellow} $\UMCSP\in\class{BQP}$ $\Rightarrow$ $\class{BQE}\not\subset\class{BQP}[n^k],\ \forall k\in\mathbb{N}$                     & \cellcolor{LightYellow} (Corollary~\ref{cor:umcsp_bqe})                                                     \\ \hline
\multirow{8}{*}{\makecell{$\SMCSP$\\(Def.~\ref{def:smcsp})}} & \cellcolor{LightCyan} $\SMCSP$ can be verified via $\class{QCMA}$                                                                                & \cellcolor{LightCyan} Theorem~\ref{thm:smcspumcspqcma}  (Theorem~\ref{thm:SQCMA_UMCSP})                                  \\ \cline{2-3} 
                          & \cellcolor{LightRed} Search-to-decision reduction for $\SMCSP$                                                                      & \cellcolor{LightRed} Theorem~\ref{thm:informal_reduction}  (Theorem~\ref{thm:s2d_smcsp})                              \\ \cline{2-3} 
                          & \cellcolor{LightRed} Self-reduction for $\SMCSP$                                                                                    & \cellcolor{LightRed} Theorem~\ref{thm:informal_reduction}  (Theorem~\ref{thm:smcsp_self_reduction})                              \\ \cline{2-3} 
                          & \cellcolor{LightRed} $\SMCSP\in\class{BQP}$                              & \cellcolor{LightRed} \\
                          & \cellcolor{LightRed} $\Rightarrow$ No pseudorandom states and no \textsf{qOWF} & \multirow{-2}{*}{\cellcolor{LightRed} Theorem~\ref{thm:smcsp_app_informal}  (Theorem~\ref{thm:smcsp_no qOWF})}                              \\ \cline{2-3} 
                          & \cellcolor{LightRed} Assume conjectures from physics & \cellcolor{LightRed}  \\ & \cellcolor{LightRed} $\SMCSP$ $\Rightarrow$ Estimating wormhole's volume                       & \multirow{-2}{*}{\cellcolor{LightRed} Theorem~\ref{thm:smcsp_app_informal}  (Theorem~\ref{thm:smcsp_wormwhole})}                                \\ \cline{2-3} 
                          & \cellcolor{LightRed} Succinct state tomography $\leq \SMCSP$                            & \cellcolor{LightRed} Theorem~\ref{thm:smcsp_app_informal}  (Theorem~\ref{thm:smcsp_succinct tomo})                              \\ \hline
\end{tabular}
\caption{Summary of our results. A result with color \colorbox{LightCyan}{Blue} is a direct extension from its classical analog. A result with color \colorbox{LightYellow}{Yellow} requires additional techniques. A result with color \colorbox{LightRed}{Red} is unique in the quantum setting.}
\label{tab:UMCSP SMCSP}
\end{table}

%% file: main-prelim.tex
\section{Preliminaries}\label{sec:main-prelim}
We start with a brief overview of quantum computation and complexity theory. We recommend the standard textbook~\cite{NC00} for a more comprehensive treatment.

\subsection{Quantum states, unitary transformations, and quantum circuits}

To give a brief introduction to the quantum computing, we divide the computation into three parts: \textit{input}, \textit{process}, and \textit{output}.

\paragraph{Input.} In quantum computing, we represent information in quantum states using qubits. 
\begin{definition}[Pure quantum state]
A pure quantum state $\ket{\psi}$ on $n$ qubits is represented as a unit vector in $\C^{2^n}$, $\ket{\psi} = (c_0,c_1,\dots,c_{2^n-1})^T$, where $c_i \in \C$ for $i\in\{0,\dots,2^n-1\}$ and $\sum_{i} |c_i|^2 = 1$.
\end{definition}

For example, $\ket{0},\ket{1},\dots,\ket{2^n-1}$ represent $n$-bit classical messages, $0,1,\dots,2^n-1$. For convenience, we sometimes denote $N=2^n$. Mathematically, one can think of $\ket{i}$ as the column vector with the $(i+1)$-th entry being $1$ and $0$ elsewhere. The input to quantum computers can be any quantum state. For classical problems, we can encode the classical input $x\in \{0,1\}^n$ as the quantum state $\ket{x}\in \C^{2^n}$. In general, any pure quantum state $\ket{\psi}$ can be represented as $\sum_{i} c_i\ket{i}$ for some $c_0,\dots,c_{2^n-1}\in \C$ with $\sum_i|c_i|^2=1$. The complex conjugate of $\ket{\psi}$ is denoted as a row vector $\bra{\psi} = (c_0^*,\cdots,c_{2^n-1}^*)$.\footnote{In general, people consider mixed state for quantum information. However, pure states suffice for our purpose. }

\paragraph{Quantum process.} Quantum process for quantum states is defined as a unitary transformation. 
\begin{definition}[Unitary transformation]
A unitary transformation $U$ for an $n$-qubit quantum state is an isomorphism in the $2^n$-dimensional Hilbert space. For convenience, we view $U$ as a $2^n\times2^n$ matrix satisfying that $U\in \C^{N\times N}$ with $UU^\dagger = U^\dagger U = I$ where $U^\dagger$ is the Hermitian adjoint of $U$. 
\end{definition}

We can represent an $n$-qubit quantum process acting on an $n$-qubit state $\ket{\psi}$ as $\ket{\psi}\mapsto U\ket{\psi}$ as a unitary matrix $U$ in $\C^{2^n\times 2^n}$. Note that a unitary matrix must preserve the norm of the input state. Thus any unitary transformation is reversible. To implement a unitary transformation, we pick a set of \textit{local unitary operations} that can generate any unitary transformation with arbitrary precision.

\begin{definition}[Universal quantum gate set]
A quantum gate set $\g$ is a set of unitaries such that for any unitary transformation $U$, $U$ can be approximated by a finite sequence of gates in $\g$. 
\end{definition}
For example, $\{\class{Toffoli},\class{H}\}$ is a universal gate set~\cite{shi2002both}. In this work, we consider gate sets which only contain unitaries with constant dimensions. 

Note that choosing different universal gate sets may cause the circuit complexity of the same object to be different.  However, the Solovay-Kitaev theorem shows that one universal gate set can approximate another one at a modest cost. 

\begin{theorem}[Solovay-Kitaev Theorem]\label{thm:sk}
Let $\g$ and $\g'$ be two universal gate sets. Then, any $s$-gate circuit $\Circ$ using gates from $\g$ can be approximated to precision $\epsilon$ by a $s\poly\log\frac{s}{\epsilon}$-gates circuit $\Circ'$ using gates from $\g'$. We say $\Circ$ approximates to $\Circ'$ with precision $\epsilon$ if 
\begin{align*}
    \|\Circ - \Circ'\| \leq \epsilon, 
\end{align*}
where $\|\cdot\|$ is $L_2$ norm. 
\end{theorem}
We will formally state Solovay-Kitaev Theorem when defining the problems of quantum circuit complexity.

We can represent quantum algorithms as quantum circuits by using a sequence of quantum gates from a universal quantum gate set.   
\begin{definition}[Quantum circuit $\class{QC}(s,t,\g)$]\label{def:main_d_qc}
Let $s,t:\N\rightarrow \N$ and $\g$ be a universal quantum gate set. A quantum circuit family $\{\Circ_n: n>0\}$ is in $\class{QC}(s,t,\g)$ if the following holds: For all $n>0$, 
\vspace{-1mm}
\begin{itemize}
\setlength\itemsep{-1mm}
    \item the input to $\Circ_n$ is an $n$-qubit quantum state $\ket{\psi}$;
    \item $\Circ_n$ extends the input layer with $t(n)$ ancilla qubits, where these ancilla qubits are initiated to $\ket{0^{t(n)}}$; 
    \item $\Circ_n$ applies $s(n)$ gates from $\g$ on the initial state $\ket{\psi}\ket{0^{t(n)}}$. 
\end{itemize}
\end{definition}
Here, in addition to the qubits for the input, the circuit can also have ancilla qubits as its working space. We say that a quantum algorithm is efficient if its corresponding circuit has circuit size at most polynomial in the input size. In the rest of the paper, we may write $\class{QC}(s,t,\g)$ as $\class{QC}(s)$ if the number of ancilla qubits is at most $O(s)$.

\paragraph{Output.} The outputs of quantum circuits defined in Definition~\ref{def:main_d_qc} are quantum states. To extract useful information from a quantum state $\ket{\psi}$, one can \emph{measure} the state. Mathematically, a measurement is simply a sampling process. For example, if we measure $\ket{\psi}$ in the \textit{computational basis}, i.e., $\{\opro{0}{0},\dots,\opro{2^n-1}{2^n-1}\}$, we get the output being index $i$ with probability $|c_i|^2$. In general, we can measure a state $\ket{\psi}$ on any orthogonal basis $\mathcal{B}$ for $\C^{2^n}$. Mathematically, this is equivalent to a change of basis via a \textit{unitary transformation}. 

In summary, a quantum algorithm for some Boolean function is as follows: Given $\ket{x}$, apply a quantum circuit $\Circ$ on state $\ket{x,0^{t(n)}}$, and then measure the state $\Circ\ket{x,0^{t(n)}}$ in the computational basis. If $\Circ$ computes $f$, then the measurement outcome will be $f(x)$ with probability good enough (e.g., $\geq 2/3$).  Note that a quantum process can have measurements in the middle of the computation in general. In this case, the process is not reversible any more. However, we can always defer these measurements until all the unitaries have been applied by adding ancilla qubits.  Therefore, for simplicity, we will only consider processes represented as unitaries followed by a computational-basis measurement.

\begin{remark}[Deferring measurements]
Let $\mathcal{M}_i$ be the computational-basis measurement on the $i$-th qubit. Let $\ket{\psi}$ be any $n$-qubit state and $U,V$ be any $n$-qubit unitaries. Then, the process $U\circ \mathcal{M}_i \circ V$  operating on $\ket{\psi}$ is equivalent to $\mathcal{M}_{n+1}\circ U\circ \mathsf{CNOT}_{i,n+1} \circ V$, where $\mathsf{CNOT}_{i,n+1}$ has the $i$-th qubit as the control qubit and the $n+1$-th qubit as the target qubit.    
\end{remark}

\subsection{Quantum complexity classes}
\label{sec:complex}
We introduce quantum complexity classes that are related to our study on the quantum MCSP. The classes we define in below are actually $\class{PromiseBQP}$ and $\class{PromiseQCMA}$. To avoid abuse of notation, we just denote them as $\class{BQP}$ and $\class{QCMA}$. 

We first give the definition of the quantum analogue of $\class{BPP}$ and $\class{P}$.  
\begin{definition}[$\class{BQP}$] 
A promise problem $P=(P_Y,P_N)$ is in $\class{BQP}$ if there exists a polynomial-time classical Turing Machine that on input $1^n$ for any $n\in \N$ outputs the description of a quantum circuit $\Circ_n$ with $\poly(n)$ gates and $\poly(n)$ ancilla qubits such that for $x\in \{0,1\}^n$ the following holds: 
\begin{enumerate}
    \item if $x\in P_Y$, $\Pr[M_1\circ\Circ_n\ket{x,0^{t}}=1]\geq 2/3$;  
    \item if $x\in P_N$, $\Pr[M_1\circ\Circ_n\ket{x,0^{t}}=1]\leq 1/3$, 
\end{enumerate}
where $M_1$ is the computational-basis measurement on the first qubit of the given state.  
\end{definition}

We also consider the quantum analogue of $\class{NP}$ and $\class{MA}$ in this work. 
\begin{definition}[$\class{QCMA}$]
A promise problem $P=(P_Y,P_N)$ is in $\class{QCMA}$ if there exists a quantum polynomial-time (QPT) algorithm $V$ such that
\begin{enumerate}
    \item for $x\in P_Y$, there exists $w\in \{0,1\}^{\poly(n)}$ such that $\Pr[V(x,w)=1]\geq 2/3$;
    \item for $x\in P_N$, for all $w\in \{0,1\}^{\poly(n)}$, $\Pr[V(x,w)=1]\leq 1/3$. 
\end{enumerate}  
\end{definition}
Another quantum analogue of $\class{MA}$ and $\class{NP}$ is called $\class{QMA}$. The difference between $\class{QMA}$ and $\class{QCMA}$ is that $\class{QMA}$ allows the certificates to be quantum states. This difference makes $\class{QCMA}\subseteq \class{QMA}$\footnote{One may expect that the quantum certificate gives the malicious prover more power to cheat in the soundness case. However, it can be shown that the existence of such a cheating prover in $\class{QMA}$ would also imply a cheating prover in $\class{QCMA}$ by the convexity of quantum states.}.

We also consider the class $\class{RQP}$, which is the one-sided error version of $\class{BQP}$:
\begin{definition}[$\class{RQP}$]
A promise problem $P=(P_Y, P_N)$ is in $\class{RQP}$ if there exists a QPT algorithm ${\cal A}$ such that
\begin{enumerate}
    \item for $x\in P_Y$, then $\Pr[{\cal A}(x)=1]\geq \frac{1}{2}$;
    \item for $x\in P_N$, then $\Pr[{\cal A}(x)=1]=0$.
\end{enumerate}
\end{definition}

\subsection{Nonuniform quantum circuit complexity classes}
With the mathematical background of quantum computing, we can define nonuniform quantum circuit complexity classes. We define the quantum analogues of $\MCSP$ as promise problems. (We will justify the reason later in Section~\ref{sec:main-QMCSP}.) Therefore, we also define complexity classes for promise problems. A promise problem is defined as $P = \{P^n\}$, where $P^n = (P^n_Y,P^n_N)$ satisfying $P^n_Y\cap P^n_N = \emptyset$ and $P^n_Y\cup P^n_N \subseteq \{0,1\}^n$. We say a promise problem $P$ is in some class $\class{C}$ if there exists a language $L\in\class{C}$ such that $P_Y\subseteq L$ and $P_N\subseteq \{0,1\}^*\setminus L$. In other words, for $x\in\{0,1\}^* \setminus P$, the answer could be arbitrary. Note that promise problems are naturally considered in quantum computing; for example, the local Hamiltonian problem~\cite{KSV02} (which is $\class{QMA}$-complete) and Identity check on basis states~\cite{wjb03} (which is $\class{QCMA}$-complete.) 

\begin{definition}[$\class{BQC}(s,t,\g)$]\label{def:main_bqc}
Let $s,t:\N\rightarrow \N$ and $\g$ be a quantum gate set. $\class{BQC}(s,t,\g)$ is the set of promise problems $P=\{P^{n}: n>0\}$ for which there exists a circuit family $\{\Circ_n: n>0\}\in \class{QC}(s,t,\g)$ such that for $n>0$, for any $x$ where $|x| =n$, 
\vspace{-1mm}
\begin{itemize}
\setlength\itemsep{-1mm}
    \item if $x\in P^{n}_Y$, then $\Pr[M_1\circ\Circ_n\ket{x,0^t}=1]\geq 2/3$; 
    \item if $x\in P^{n}_N$, $\Pr[M_1\circ\Circ_n\ket{x,0^t}=1]\leq 1/3$.
\end{itemize}
Here, $M_1$ is the computational-basis measurement on the first qubit. 
\end{definition}
In the rest of the paper, we will write $\class{BQC}(s,t,\g)$ as $\class{BQC}(s)$ for simplicity if the number of ancilla qubits is at most $O(s)$. 

In addition to $\class{BQC}$, we will also consider quantum complexity classes such as $\class{QMCA}$ and $\class{BQP}$. For the same reason, the classes we consider are actually $\class{PromiseBQP}$ and $\class{PromiseQCMA}$. To avoid abuse of notation, we just denote them as $\class{BQP}$ and $\class{QCMA}$. Also, when $\class{NP}$ is mentioned, we are actually considering $\class{PromiseNP}$. The formal definitions of these classes are given in Appendix~\ref{sec:complex}. 




%% file: main-MQCSP.tex
\section{Minimum Quantum Circuit Size Problems}\label{sec:main-QMCSP}

We start off the quantum $\MCSP$ program by giving the definitions of various quantum analogs of the classical $\MCSP$ in Section~\ref{sec:definitions} and investigating some basic complexity-theoretic results in Section~\ref{sec:MQCSP upper bound} and Section~\ref{sec:MQCSP unconditional hardness}.

\subsection{Problem definitions}\label{sec:definitions}
While classical computation works on Boolean strings, quantum computation works on unit complex vectors. Thus, there are multiple natural notions of $\MCSP$ that can be defined and studied in the quantum realm. But first let us formally define the classical $\class{MCSP}$ as follows.

\begin{definition}[Classical $\class{MCSP}$]~\label{def:main_mcsp}
Let $n,s\in \N$\footnote{For every Boolean function, there is a circuit with size at most $O(2^n/n)$. Therefore, one can suppose $s$ is at most $O(2^n/n)$. Besides, one can also consider $s$ is given in unary, such that the problem is still well-defined in the sense that it is trivially in $\class{NP}$.}. Let $f: \{0,1\}^n\rightarrow \{0,1\}$ be a Boolean function. The problem is, given the truth table $\T(f)$ of $f$ and the size parameter $s$ in unary, decide if there exists a classical Boolean circuit $C$ of size at most $s$ such that $C(x) = f(x)$ for all $x\in \{0,1\}^n$. 
\end{definition}

Note that $\MCSP\in\class{NP}$ because given a truth table $\T(f)$ a circuit $C$, we can verify whether $C(x)=f(x)$ for all $x\in\{0,1\}^n$ in $\poly(|\T(f)|,1^s)$ time. On the other hand, when $s= \Omega(n)$, the number of circuits of size at most $s$ is $2^{\Theta(s\log s)}$, which is $2^{\omega(n)}$ by the counting argument. Besides, for every Boolean function, there exists a circuit with size at most $O(2^n/n)$~\cite{Lupanov58}; therefore, we can suppose the $s = O(2^n/n)$, which implies that brute-force search takes $2^{O(2^n)}$ time to solve $\MCSP$ in the worst case and it is the best known algorithm for $\MCSP$.\\

As quantum computation is generally believed to be more powerful than classical computation, it is likely that the quantum circuit complexities for some Boolean functions are much different from their classical circuit complexities. Specifically, quantum circuits can create quantum entanglement between qubits that cannot be simulated classically. Therefore, we define the following problem for studying the quantum circuit complexity of the given Boolean function.

\begin{definition}[$\MQCSP_{\alpha,\beta}$]\label{def:main_bmcsp}
Fix a universal gate set $\g$. Let $n,s,t\in \N$ and $t\leq s$. Let $f:\{0,1\}^n\rightarrow \{0,1\}$ be a Boolean function. Let $\alpha,\beta\in (1/2,1)$ such that $\alpha-\beta\geq \frac{1}{\poly(2^n)}$. $\MQCSP$ is a promise problem defined as follows.
\vspace{-1mm}
\begin{itemize}
\setlength\itemsep{-1mm}
\item \textbf{Inputs:} the truth table $\T(f)$ of $f$, the size parameter $s$ in unary representation, and the ancilla parameter $t$. 
\item \textbf{Yes instance:} there exists a quantum circuit $\Circ$ using at most $s$ gates and operating on at most $n+t$ qubits such that for all $x\in \{0,1\}^n$,  
$\|(\bra{f(x)}\otimes I_{n+t-1})\Circ\ket{x,0^{t}}\|^2 \geq \alpha$.
\item \textbf{No instance:} for every quantum circuit $\Circ$ using at most $s$ gates and operating on at most $n+t$ qubits, there exists $x\in\{0,1\}^n$ such that $\|(\bra{f(x)}\otimes I_{n+t-1})\Circ\ket{x,0^t}\|^2 \leq \beta$. 
\end{itemize}
With the promise that the input must be either a yes instance or a no instance, the problem is to decide whether the input is a yes instance or not. 
\end{definition}

\begin{remark}
Here, we set the thresholds for the yes and no instances to be $\alpha, \beta$ such that $1/2<\beta<\alpha<1$ and $\alpha-\beta>\frac{1}{\poly(2^n)}$. We require $\alpha$ and $\beta$ to be greater than $1/2$ because a quantum circuit that outputs a uniformly random bit (e.g., measure $\ket{+}$ in the computational basis) can compute $f(x)$ with $1/2$ probability for all $x$. For simplicity, in the rest of the work, we will ignore the subscription $\alpha,\beta$ and will specify them when it is necessary.
\end{remark}


For $\MQCSP$, which gate set $\g$ is used is another important parameter to be considered. One may ask if circuit complexity can significantly change when considering different $\g$. Fortunately, according to the Solovay-Kitaev Theorem in Theorem~\ref{thm:sk}, we can conclude that any $s$-gate circuit using gates from $\g$ can be $\epsilon$-approximated by an ($s\cdot\polylog\frac{s}{\epsilon}$)-gate circuit from another universal gate set. Hence, the circuit complexity only modestly changes when considering different gate sets. 

\begin{claim}\label{claim:gateset2gateset}
Fix two universal gate sets $\g$ and $\g'$. Suppose that there exists a $s$-gate circuit $\Circ$ that uses gates from $\g$ such that for all $x$, $\|(\bra{f(x)}\otimes I_{n+t-1})\Circ\ket{x,0^{t}}\| \geq 1-\delta$. Then, there exists another circuit $\Circ'$ that uses $s\cdot\polylog\frac{s}{\epsilon}$ gates in $\g'$ such that $\|(\bra{f(x)}\otimes I_{n+t-1})\Circ\ket{x,0^{t}}\| \geq 1-\delta-\epsilon^2/2$. 
\end{claim}
\begin{proof}
The proof simply follows from the Solovay-Kitaev Theorem in Theorem~\ref{thm:sk}. The only subtlety is that the distance measure in Theorem~\ref{thm:sk} is $L_2$ norm distance. However, for any two states $\ket{\psi}$ and $\ket{\phi}$, we have $|\langle \psi|\phi\rangle| \geq 1-\frac{1}{2}\|\ket{\psi} - \ket{\phi}\|^2$. Thus, we can obtain the lower bound for $\|(\bra{f(x)}\otimes I_{n+t-1})\Circ\ket{x,0^{t}}\|$ by using the $L_2$ norm between $\Circ$ and $\Circ'$.
\end{proof}

In this work, we mainly focus on arbitrary gate sets containing one- and two-gates and $|\g| = O(1)$. However, for some applications, we may require a particular gate set such as $\{\class{Toffoli},\class{H}\}$. We will specify $\g$ when it is necessary. We assume $t\leq s$ without loss of generality since we mainly consider the gate set $\g$ to have one- and two-qubit gates. Specifically, if there are more than $s$ ancilla qubits, there must be ancilla qubits that are not used by any gate.


We define the problem as a promise problem for two reasons: first, applying measurements on quantum states generally gives probabilistic outputs. Similar to many probabilistic algorithms, we say a quantum algorithm solves a problem if it outputs the answer with high probability in general. Check the definition of $\class{BQP}$ for an example. Along this line, we expect a quantum circuit $\Circ$ to implement the given Boolean function $f$ with high probability, i.e., for each input $x$, the circuit outputs $f(x)$ with high probability. The second reason is about verifying the circuit. Consider the case where $\Circ$ only fails on one $x$ with success probability $2/3-\epsilon$, where $\epsilon$ is some extremely small number.  In this case, it is hard to verify the circuit efficiently. Therefore, we require a gap for efficient verification and say that $\Circ$ does not implement $f$ if it can only output $f(x)$ with probability with small probability for some $x$. 





\paragraph{Other variants.}

In many applications, the \textit{gap-version} of $\MCSP$ is much easier and more flexible to work with. Below we define the gap-version of $\MQCSP$ and the multi-output $\MQCSP$.


\begin{definition}[${\MQCSP_{a,b}[s,s',t]}$]\label{def:app_gap MQCSP}
Let $n,s,s',t\in \N$ such that $t\leq s<s'\leq 2^{O(n)}$. Let $a-b\geq 1/\poly(2^n,1^{|s|})$. Let $f:\{0,1\}^n\rightarrow \{0,1\}$ be a Boolean function. $\MQCSP[s,s']$ is a promise problem defined as follows.
\begin{itemize}
\item Input: the truth table $\T(f)$ of $f$, the size parameter $s$ in unary, and the ancilla parameter $t$.
\item Yes instance: there exists a quantum circuit $\Circ$ using at most $s$ gates and operating on at most $n+t$ qubits such that for all $x\in \{0,1\}^n$,  
\begin{align*}
    \|(\bra{f(x)}\otimes I_{n+t-1})\Circ\ket{x,0^{t}}\|^2 \geq \frac{2}{3} \, .
\end{align*}
\item No instance: for every quantum circuit $\Circ$ using at most $s'$ gates and operating on at most $n+t$ qubits, there exists $x\in\{0,1\}^n$ such that
\begin{align*}
    \|(\bra{f(x)}\otimes I_{n+t-1})\Circ\ket{x,0^t}\|^2 \leq \frac{1}{2} \, .
\end{align*}
\end{itemize}
With the promise that the input must be either a yes instance or a no instance, the problem is to decide whether the input is a yes instance or not. 
\end{definition}
When it is clear from the context, we may use $\MQCSP^\star$ to denote ${\MQCSP_{a,b}[s,s',t]}$.

\begin{definition}[$\g$-$\mBMCSP_{\alpha,\beta}{(s,t)}$]\label{def:BQMCSP-multi}
Let $m,s,t$ be functions of $n$ such that $t\leq s\leq 2^{o(n)}$ and $m\leq n+t$. Let $\alpha,\beta \in [2^{-m},1]$ such that $\alpha-\beta>\frac{1}{\poly(2^n)}$. Let $f:\{0,1\}^n\rightarrow \{0,1\}^m$ be a multioutput function. $\g-\mBMCSP_{\alpha,\beta}{(s,t)}$ is a promise problem that  
\begin{enumerate}
\item Input: the truth table $\T(f)$ of $f$.
    \item Yes instance: there exists a quantum circuit $\Circ$ using at most $s$ gates from $\g$ and operating on at most $n+t$ qubits such that for all $x\in \{0,1\}^n$,  
    \begin{align*}
        \|(\bra{f(x)}\otimes I_{n+t-m})\Circ\ket{x,0^{t}}\|^2 \geq \alpha, 
    \end{align*}
    \item No instance: for any quantum circuit $\Circ$ using at most $s$ gates from $\g$ and operating on at most $n+t$ qubits, there exists $x\in\{0,1\}^n$ such that
    \begin{align*}
        \|(\bra{f(x)}\otimes I_{n+t-m})\Circ\ket{x,0^t}\|^2 \leq \beta. 
    \end{align*}
\end{enumerate}
With the promise that the input must be either a yes instance or a no instance, the problem is to decide whether the input is a yes instance or not. 
\end{definition}

\paragraph{Natural property.}
It is worth noting that we can view an efficient quantum algorithm for $\class{MQCSP}$ as \emph{quantum natural property against quantum circuit classes}. Natural properties against circuit classes were first defined by Razborov and Rudich~\cite{RR97}, and recently, Arunachalam et al.~\cite{agg20} further considered quantum natural properties against circuit classes. 

\begin{definition}[Natural Property~\cite{RR97}]\label{def:main_natural_prop}
Let $C$ be a uniform complexity class and $C'$ be a circuit class. We say that a property $\Gamma = \{\Gamma_n: n\in \mathbb{N}\}$ is $C$-natural against $C'$ if the following holds.
\vspace{-1mm}
\begin{enumerate}
\setlength\itemsep{-1mm}
    \item \textbf{Constructivity:} for all $L\in \Gamma$, $L\in C$.   
    \item \textbf{Largeness:} There exists $n_0\in \N$, for $n\geq n_0$, $|\Gamma_n|/|\mathcal{F}_n|\geq \frac{1}{2}$, where $\mathcal{F}_n$ is the set of all Boolean functions with input length $n$. 
    \item \textbf{Usefulness:} There exists $n_0\in \N$, for $n\geq n_0$, $\Gamma_n\cap C'_n = \emptyset$, where $C'_n$ is the set of circuits in $C'$ on $n$ (qu)bits.   
\end{enumerate}
\end{definition}

Note that an $\class{MQCSP}$ oracle can be used to construct natural properties against quantum circuit classes $\class{BQC}[s]$ for any $s$. 
Therefore, if we suppose that $\class{MQCSP}$ is in $\class{BQP}$, then we can have properties that are $\class{BQP}$-natural against quantum circuit classes. For simplicity, we call properties that are $\class{BQP}$-natural as \emph{quantum natural properties}. 
Arunachalam et al.~\cite{agg20} first considered quantum natural properties against circuit classes, and proved circuit lower bounds for quantum complexity classes. Our work can also be viewed as a study of quantum natural properties against quantum circuit classes. The formal definition of $\class{BQP}$-natural property is in below:

\begin{definition}[$\class{BQP}$-Natural Property~\cite{agg20}]\label{defn:qnaturalprop}
We say that a  combinatorial property $\Gamma$ is $C$-natural against polynomial-size quantum circuits $(\class{BQC}[\poly])$ if the following holds. 
\begin{enumerate}
    \item \textbf{Constructivity:} for any string $L\in \Gamma$, $L$ can be accepted by a $\class{BQP}$ algorithm.   
    \item \textbf{Largeness:} There exists $n_0\in \N$, for $n\geq n_0$, $\frac{|\Gamma_n|}{|\mathcal{F}_n|}\geq \frac{1}{2}$. 
    \item \textbf{Usefulness:} There exists $n_0\in \N$, for $n\geq n_0$, any string accepted by $\class{BQC}[\poly]$ is not in $\Gamma_n$.   
\end{enumerate}
\end{definition}

Then, our observation on the connection between $\MQCSP$ and quantum natural property is formally stated as follows:
\begin{observation}\label{obs:mqcsp_nat_prop}
If $\class{MQCSP}\in \class{BQP}$, then there exists a $\class{BQP}$-natural property against quantum circuits $\class{QC}[n^k]$ for any $k\in \N_+$.
\end{observation}

\subsection{Upper bounds for \texorpdfstring{$\class{MQCSP}$}{MQCSP}}\label{sec:MQCSP upper bound}

It turns out that, unlike the classical $\MCSP$, $\MQCSP$ is not trivially in $\class{NP}$. The best upper bound we are able to get for $\MQCSP$ is $\class{QCMA}$, the quantum analogue of $\class{NP}$ (or $\class{MA}$). Before showing that $\MQCSP$ is in $\class{QCMA}$, we first discuss why it is not trivially in $\class{NP}$ like the classical $\MCSP$. One obvious reason is that $\MQCSP$ is a promise problem. Therefore, we consider $\class{PromiseNP}$, which definition is the same as $\class{NP}$ except that $\class{PromiseNP}$ relax the definition of $\class{NP}$ to contain promise problems that have $\class{NP}$ certificates. For the ease of presentation, we will use $\class{NP}$ for both $\class{NP}$ and $\class{PromiseNP}$. Then, when the number of ancilla qubits is linear, one can verify the given circuit by simply writing down the corresponding unitary.  

\begin{theorem}\label{thm:np_mcsp}
$\MQCSP$ is in $\class{NP}$ when only a linear number of ancilla qubits are allowed.
\end{theorem}

However, when the number of ancilla qubits is superlinear, e.g., $n^2$, the quantum circuit $\Circ$ operates on $2^{O(n^2)}$ qubits, and thus the corresponding unitary $U_{\Circ}$ has dimension $2^{O(n^2)}$ which is superpolynomial in $2^n$. In this case, the verifier cannot compute $U_{\Circ}$ classically in time $\poly(2^n)$. Therefore, the trivial approach does not work. 

Note that although the trivial approach fails to show that $\MQCSP$ is in $\class{NP}$, it does not rule out the possibility that $\MQCSP$ can be efficiently verified via other approaches. In the following theorem, we show that a quantum verifier can efficiently verify the given quantum circuit, and thus $\MQCSP$ is in $\class{QCMA}$.    

\begin{restatable}{theorem}{qcmamqcsp}\label{thm:QCMA_MCSP}
$\MQCSP \in \class{QCMA}$. 
\end{restatable}

We leave the proof to Appendix~\ref{sec:proof_szk} for completeness. 


\subsection{Hardness of quantum \texorpdfstring{$\class{MCSP}$}{MCSP}}\label{sec:MQCSP unconditional hardness}

It is a major open problem in complexity theory to understand the hardness of classical $\MCSP$. Here, we show that the state-of-the-art hardness results on $\MCSP$ (and its variants) can be extended to $\MQCSP$. We remark that this is actually not straightforward to see because the classical $\MCSP$ is incomparable with $\MQCSP$.

First, we show that the $\class{SZK}$-hardness result of $\MCSP$ by Allender and Das~\cite{AB14} can be extended to $\MQCSP$. Here, $\class{SZK}$ stands for the complexity class \textit{Statistical Zero Knowledge} that lies between $\class{P}$ and $\class{NP}$. We first define $\class{SZK}$ and the statistical distance as follows.

\begin{definition}[Statistical Distance $SD(X,Y)$]
Let $X$ and $Y$ be two probability distributions, the statistical distance between $X$ and $Y$ can be defined as follows:
\begin{align*}
\max_{S\subseteq \{0,1\}^{m'}} |\Pr[X\in S] - \Pr[Y\in S]|    
\end{align*}

\end{definition}

\begin{definition}[$\class{SZK}$]
A promise problem $P = (P_Y,P_N)$ is in $\class{SZK}$ if there exists a $\class{PPT}$ verifier $V$ and an interactive proof system $(P,V)$ satisfying the following properties: 
\begin{enumerate}
    \item \textbf{Completeness:} For $x\in P_Y$, there exists $P$ such that $\Pr[\langle P,V\rangle(x)=1] \geq \frac{2}{3}$. 
    \item \textbf{Soundness:} For $x\in P_N$, for all $P$, $\Pr[\langle P,V\rangle(x)=1] \leq \frac{1}{3}$.
    \item \textbf{Statistical zero-knowledge:} There exists a $\class{PPT}$ simulator $S$, for all $\class{PPT}$ verifier $V^*$, for all $x\in P_Y$, 
    \begin{align*}
        \mathsf{SD}(S(V^*)(x), \langle P,V^*\rangle(x)) \leq \negl(n). 
    \end{align*}
\end{enumerate}
\end{definition}

We introduce an $\class{SZK}$-complete problem by Ben-Or and Gutfreund~\cite{BG03}. 
\begin{definition}[Polarized Image Intersection Density ($\class{PIID}$),~\cite{BG03}]
Given two circuits $C_0,C_1: \{0,1\}^m\rightarrow \{0,1\}^{m'}$ of size $n^k$ with the promise that either
\begin{enumerate}
    \item $\max_{S\subseteq \{0,1\}^{m'}} |\Pr_x[C_0(x)\in S] - \Pr_x[C_1(x)\in S]|\leq \frac{1}{2^n}$, or 
    \item $\Pr_{x\in \{0,1\}^{m'}}[\exists y\in \mathsf{Im}(C_0)\mbox{ such that }C_1(x) = y]\leq \frac{1}{2^n}$, 
\end{enumerate}
where $n = \poly(m)$ and $\class{Im}(C) := \{C(x): x\in \{0,1\}^m\}$. The problem is to decide which case is true. 
\end{definition}

\begin{restatable}{theorem}{mcspszkhard}\label{thm:szk_mcsp}
$\class{SZK} \subseteq \class{BPP}^{\class{MQCSP}}$
\end{restatable}

To prove Theorem~\ref{thm:szk_mcsp}, we first observe
that the existence of small classical circuit implies the existence small quantum circuits and an $\MQCSP$ oracle can invert one-way functions (which we will prove in Section~\ref{sec:crypto}). Then, we can show that $\class{PIID}$ is in $\class{BPP}^{\class{MQCSP}}$ following the framework of~\cite{AB14}.  We leave the proof to Appendix~\ref{sec:proof_szk} for completeness.

Next, we quantize the recent breakthrough of Ilango et al.~\cite{ILO20} on the $\class{NP}$-hardness of classical $\MCSP$. There are two main differences between the classical and quantum settings: (i) the circuit model is different and hence makes the combinatorics different, and (ii) the quantum setting allows the output to have some errors.
We partially overcome these two difficulties and prove the following theorem. 

\begin{restatable}{theorem}{mbmcspnphard}\label{thm:mqcsp_multi hard}
Suppose $\fCNOT\circ(\fI\otimes \fX), \fToffoli\in \g$. Every multi-bit gate in $\g$ behaves classically on classical inputs and has at most 1 target wire and at most 2 control wire. (That is, except 1 wire, the outputs of the other wires, at most 2, are the same as their corresponding classical inputs. For example, $\fCNOT$ gate.) Then $\g$-$\mBMCSP$ is $\class{NP}$-hard under randomized reduction.
\end{restatable}
$\fCNOT\circ(\fI\otimes \fX)$ is the following operation on two input wires, denoted as control wire and target wire: first do a $\fX$ on the target wire, and do a $\fCNOT$ from the control wire to the target wire. We consider it as a single gate, as the analog of the classical NOT gate.\par
Here the choice of gate set matter: we need the quantum gate set to contain the analog of the usual classical gate set. $\fCNOT\circ(\fI\otimes \fX)$ is the analog of classical single-bit NOT operation, and $\fToffoli$ is the analog of classical AND operation. Here the correspondence has two properties: (1) if the target wire is in the zero state and the control wire is classical, the output of the target wire will be the corresponding classical logical computation result; (2) if the input of the control wire is classical, the output of the control wire will remain the same. Since in the quantum world data copy is not for free, the second property is important for deriving our result.\par
The proof follows the outline of the proof in \cite{ILO20}. We note there are two differences during the proof in the quantum case compared to the classical case:
\begin{itemize}
    \item The circuit model is different. In the classical world the gates are single-output and we assume free-copy. And the basic gate set contains AND, OR, NOT gates. In quantum world, data copy is not for free and we need to use the $\fToffoli$ gate to implement the AND/OR gates.
    \begin{remark}
    One idea might be to use the Solovay-Kitaev theorem to switch the gate set and make the theorem general. But this does not work here in an immediate way. Our proof does not imply the problem is also $\class{NP}$-hard to approximate multiplicatively. On the other hand, the classical result \cite{ILO20} is not known to be general on different gate set either.
    \end{remark}
    \item In the definition of multi-output minimum quantum circuit size problem, we allow the output to have some errors, which is not considered in the classical world.
\end{itemize}
\begin{proof}[Proof of Theorem~\ref{thm:mqcsp_multi hard}]
We consider the same construction as \cite{ILO20}. Let's restate it here for completeness.\par
\begin{enumerate}
    \item Choose a large enough constant $r$ so that 20-approximating $r$-bounded set cover problem is $\class{NP}$-hard. Consider an instance $(1^n, \cS)$ of this problem.\par
    \item $m$ is the least power of $2$ that is greater than $n^3$. Sample the truth table $T$ representing a function on $\{0,1\}^{\log m}\rightarrow \{0,1\}$ uniformly at random.
     Construct $g:=\bullet_{S\in \cS} \text{Eval-DNF}_{T_{\langle S^m\rangle}}$ where:
    \begin{itemize}
        \item[--] To define $\text{DNF}_{f}$ that encode the truth table $f$, we first repeat the construction in \cite{ILO20} for completeness:
        $$\text{DNF}_f:=((x_1=y_1^1)\land\cdots\land (x_n=y_n^1))\lor \cdots \lor ((x_1=y_1^t)\land\cdots\land (x_n=y_n^t))$$
        where $y^1,\cdots y^t$ are YES inputs of $f$ in lexicographical order, $x_1,\cdots x_n$ index the bits of the input string $x$, $y_1^j,\cdots y_n^j$ index the bits of $y^j$, and $(x_i=y_i^t)$ denotes $(x_i\oplus (1\oplus y_i^j))$.\par
        We use the same construction with one difference: here $\lor$ is further decomposed to $\lnot$ and $\land$.
        \item[--] $T_{\langle S\rangle}$ is the truth table that is equal to $T$ for input in $S$ and $0$ everywhere else.
        \item[--] $S^m:=\cup_{i\in S} P_i^{m,n}$ where $P_i^{m,n}:=\{j\in [m]:j\equiv i \mod n\}$. This step closes the gap between $[m]$ (the $\MCSP$ size) and $[n]$ (the set cover size).
        \item[--] ``$\bullet$'' is used on two functions that have the same input domain, and it concatenates the outputs of these functions to get a new function.
        \item[--] To define Eval-$C$, we first consider $x_1\bullet x_2\bullet \cdots x_n\bullet g_1\bullet g_2\bullet\cdots g_s$ where $g_1,\cdots g_s$ are the output of each gate in circuit $C$. Then we remove the gate output that are the same on all the inputs.
    \end{itemize}
    \item As in \cite{ILO20}, define $k$ as the number of distinct components of $g$ that are not directly a function identical to an input. Note that this can be efficiently computed.\par
    Take $\alpha=1, \beta=0.99, t=10s$ ($s$ is the output number of our construction).\par Define $CC_{\alpha,\beta}({t,\T(f)})$ as the subroutine that uses binary search to find the minimum $s$ such that $\g-\mBMCSP_{\alpha,\beta}{(s,t)}(\T(f))=\text{true}$.\footnote{Since multiMQCSP is a promise problem this routine does not necesarrily find the minimum $s$ but should return a value that there exists a circuit of this size that approximate the function everywhere with correct probability $\beta$. This is sufficient for later proof.} Use the $\mBMCSP$ oracle and compute $$\Delta:=CC_{\alpha,\beta}({t,\T(T\bullet g)})-k$$ as the approximation of the set cover instance $(1^n, \cS)$.
\end{enumerate}
To analyze this reduction, we need to prove the followings steps:
\begin{enumerate}
    \item $CC_{\alpha,\beta}({t,\T(g)})=k$
    \item $\Delta\leq 3\cdot \mathrm{cover}([n],\cS)+1$ where $\mathrm{cover}([n],\cS)$ is the size of the minimum set cover solution for $\cS$.
    \item $\Delta\geq \mathrm{cover}([n],\cS)/6-6$ with probability $1-2^{-\Omega(m)}$.
\end{enumerate}
Then we get an approximation to the set cover problem.\par
Let us prove the three statements step-by-step.
    \paragraph{Step 1:}
    The $\leq $ part is proved by the function construction itself. We implement $\lnot$ with the $\fCNOT\circ(\fI\otimes \fX)$ gate (and write the output on an empty ancilla system) and implement $\land$ with the $\fToffoli$ gate.\par
    The $\geq$ part is slightly different since in quantum case the gate model is different.
    In classical world all the gates are single-output, while in quantum world there are multi-output gates. However, for the multi-output gates like $\fCNOT$ and $\fToffoli$, there is only one target wire, and the other wires are control wire. Thus for each output component, we can always find the nearest gate that does not use it as a control wire (if there is such a gate along the way, ignore it). In this way each different output component corresponds to a different gate in the circuit, which completes the proof.
    \paragraph{Step 2:} As \cite{ILO20}, when $\mathrm{cover}([n],\cS)=\ell$, without loss of generality assume $S_1,\cdots S_\ell$ are a set cover. Then $T=T_{\langle S^m_1\rangle}\lor\cdots \lor T_{\langle S^m_\ell\rangle}$. This can be computed using $3\ell+1$ extra gates on the minimum circuit of Eval-$g$. (Note that in the quantum world we need slightly more gates than the classical world. And we need to evaluate the OR gate by NOT-AND-NOT gates to get $T$.)
    \paragraph{Step 3:} Denote $\ell=\lfloor \mathrm{cover}([n],\cS)/6\rfloor$. The goal is to show that the probability that $\Delta\leq \ell$ is small by showing that $T$ satisfying $\Delta\leq \ell$ must have a short description. Suppose $T$ is a truth table such that the condition $\Delta>\ell$ does not hold. We need to find a circuit of gate number $\leq 2\ell$ where:\begin{itemize}
        \item The inputs are: the bits of $x$; and the output of $g$.
        \item It encodes the output of $T$.
    \end{itemize} We use the similar idea to \cite{ILO20} but we need to address the two problems discussed before this proof.\par
    As what we did in Step 1, we can associate each output component ($g_i(x)$, for example) to a unique gate in the circuit. 
    As \cite{ILO20}, we remove these gates from the circuit. There might be some gates between this gate and the output $g_i(x)$ that use the wire as control wires. For these gates, simply use $g_i(x)$ as the control value.\par
    As in \cite{ILO20} we have $CC_{\alpha,\beta}(t,\T (T\bullet g))\leq \ell+k$. And since for each $g_i$ at least one gate is removed, the remaining circuit is a circuit $D$ that takes $\log (m)+k$ inputs and has at most $\ell$ gates such that
    $$D(x,g_1(x),\cdots g_k(x))\text{   encodes   }T(x)$$
    Then since each gate has fan-in at most $3$ the circuit uses at most $3\ell$ components of $g$. Then after a possible relabling of $g_1\cdots g_k$ we can assume $D$ takes $\log (m)+3\ell$ inputs such that
    $$D(x,g_1(x),\cdots g_{3\ell}(x))\text{   encodes   }T(x)$$
    The new circuit does not necessarily behave the same as the original circuit, but they do behave the same (up to a global phase) on the subspace that all the outputs are computed correctly. By the definition of $\mBMCSP$ and the choices of parameters this is true with norm $\geq 0.99$. Thus we can view the shrinked circuit as an encoding of $T$ by focusing on the most-possible outputs of this circuit. Then by the same argument as \cite{ILO20} such a shrinked circuit has a  description of $(1-\Omega(1))m$ bits, which implies such $T$ has at most $2^{(1-\Omega(1))m}$ choices thus a random $T$ falls into this case with exponentially small probability.
\end{proof}


However, we don't know whether this problem is NP-complete, since it's not known to be in NP. With a proof similar to that of Theorem~\ref{thm:QCMA_MCSP}, we only know $\mBMCSP\in\class{QCMA}$. Namely, there remains a gap between our understandings of the upper bound and hardness of $\mBMCSP$. We pose it as an open problem to settle the complexity of $\mBMCSP$.

%% file: main-connections.tex
\section{Connections Between MQCSP and Other Problems}\label{sec:main-connections}







\input{crypto}

\subsection{Learning theory}\label{sec:main learning}

In this section, we discuss connections between $\MQCSP$ and learning theory. We consider two standard settings: probably approximately correct (PAC) learning and quantum learning. We postpone the details to Appendix~\ref{app:learning}.

\paragraph{PAC learning.}
Let $\class{C}$ be a circuit class. We are interested in how to efficiently learn a function in $\class{C}$. PAC learning is a theoretical framework to evaluate how well a learning algorithm is. Here we focus on a special setting of PAC learning where the algorithm is able to query any input to the unknown function. In the following, we denote $\class{C}$-$\MCSP$ as the classical $\MCSP$ problem with respect to the circuit class $\class{C}$.

\begin{restatable}[PAC learning over the uniform distribution with membership queries]{definition}{paclearning}\label{def:PAC learning}
Let $\class{C}$ be a circuit class and let $\epsilon,\delta>0$. We say an algorithm $(\epsilon,\delta)$-PAC-learns $\class{C}$ over the uniform distribution with membership queries if the following hold. For every $n\in\mathbb{N}$ and $n$-variate $f\in\class{C}$, given membership query access to $f$, the algorithm outputs a circuits $C$ such that with probability at least $1-\delta$ over its internal randomness, we have $\Pr_{x\in\{0,1\}^n}[f(x)\neq C(x)]<\epsilon$. The running time of the learning algorithm is measured as a function of $n,1/\epsilon,1/\delta$ and, $\textsf{size}(f)$.
\end{restatable}

The seminal paper of Carmosino et al.~\cite{CIKK16} showed that efficient PAC learning for a (classical) circuit class $\class{C}$ is \textit{equivalent} to the corresponding $\MCSP$ being easy. Here, we quantize this connection and show in the following theorem that efficient PAC-learning for $\class{BQP/poly}$ is equivalent to efficient algorithm for $\MQCSP$. Here, $\class{BQP/poly}$ is defined as $\bigcup_{s\leq \poly(n)} \class{BQC}(s)$.

For technical reason, we need to work on a gap version of $\MQCSP$ in one direction of the equivalence. Let $\tau:\N\to(0,1/2)$, $\MQCSP[s,s',t,\tau]$ is defined as the gap problem where the No instances in Definition~\ref{def:app_gap MQCSP} becomes ``for every quantum circuit $\Circ$ using at most $s'$ gates and operating on at most $n+t$ qubits, there are at least $\tau$ fraction of $x\in\{0,1\}^n$ such that $\|(\bra{f(x)}\otimes I_{n+t-1})\Circ\ket{x,0^t}\|^2 \leq \frac{1}{2}$''.

\begin{restatable}[Equivalence of efficient PAC learning for $\class{BQP/poly}$ and efficient randomized algorithm for $\MQCSP$]{theorem}{learningpac}\label{thm:PAC learning and MQCSP}
\mbox{}
\begin{itemize}
\setlength\itemsep{-1mm}
    \item If $\MQCSP\in\class{BPP}$, then there is a randomized algorithm that $(1/\poly(n),\delta)$-PAC learns $f\in\class{BQP/poly}$ under the uniform distribution with membership queries for every $\delta>0$. Specifically, the algorithm runs in quasi-polynomial time.
    \item If there is a randomized algorithm that $(1/\poly(n),\delta)$-PAC learns $f\in\class{BQP/poly}$ under the uniform distribution with membership queries for some $\delta>0$ in $2^{O(n)}$ time, then we have $\MQCSP[\poly(n),\omega(\poly(n)),\poly(n),\tau]\in\class{BQP}$ and $\MQCSP[\poly(n),\omega(\poly(n)),O(n),\tau]\in\class{BPP}$ for every $\tau>0$.
\end{itemize}
\end{restatable}

Similarly, the positive resolution of Open Problem~\ref{conj:linear ancilla} would strengthen the conclusion of the second item in Theorem~\ref{thm:PAC learning and MQCSP} to $\MQCSP[\poly(n),\omega(\poly(n)),\poly(n),\tau]\in\class{BPP}$.

\paragraph{Quantum learning.}
As it could be the case that $\MQCSP$ might have non-trivial quantum algorithm, it is also of interest to study the connection to quantum learning~\cite{agg20}.

\begin{restatable}[Quantum learning]{definition}{quantumlearning}\label{def:quantum learning}
Let $\class{C}$ be a circuit class of boolean functions and let $\epsilon,\delta>0$. We say a quantum algorithm $(\epsilon,\delta)$-learns $\class{C}$ if the following hold. For every $n\in\mathbb{N}$ and $n$-variate $f\in\class{C}$, given quantum oracle access to $f$, the algorithm outputs a polynomial-size quantum circuit $U$ such that with probability at least $1-\delta$, we have $\mathbb{E}_{x\in\{0,1\}^n}[|(\bra{f(x)}\otimes I)U\ket{x,0^m}|^2]>1-\epsilon$. The running time of the learning algorithm is measured as a function of $n,1/\epsilon,1/\delta$ and, $\textsf{size}(f)$.
\end{restatable}

It turns out that efficient quantum learning for a circuit class $\class{C}$ (could be either a classical circuit class or a quantum circuit class) is equivalent to efficient quantum algorithm for $\class{C}$-$\class{MCSP}$. Similarly, $\class{C}$-$\class{MCSP}[s,s',\tau]$ is defined as the gap problem with the No instances being the truth tables where every circuit $\mathcal{C}$ of size $s'$ errs on $\tau$ fraction of the inputs.

\begin{restatable}[Equivalence of efficient quantum learning and efficient quantum algorithm for $\class{C}$-$\class{MCSP}$]{theorem}{learningquantum}\label{thm:quantum learning MCSP}
Let $\class{C}$ be a circuit class.
\begin{itemize}
\setlength\itemsep{-1mm}
    \item If $\class{C}$-$\class{MCSP}\in\class{BQP}$, then there exists a quantum algorithm that $(1/\poly(n),\delta)$-learns $\class{C}$ for every $\delta>0$. Specifically, the algorithm runs in polynomial time.
    
    \item If there exists a quantum algorithm that $(\epsilon,\delta)$-learns $\class{C}$ in time $2^{O(n)}$ for some constants $\epsilon,\delta\in(0,1/2)$, then we have $\class{C}$-$\class{MCSP}[\poly(n),\omega(\poly(n)),\tau]\in\class{BQP}$ for every $\tau>0$.
\end{itemize}
\end{restatable}

\subsection{Circuit lower bounds}\label{sec:main circuit lbs}
The classical $\MCSP$ is tightly connected to circuit lower bounds. Many results show that a fast algorithm for $\MCSP$ will lead to breakthrough in circuit lower bounds, which on the other hand indicates that $\MCSP$ might be very difficult to solve. In this section, we ``quantize'' four results relating $\MQCSP$ and quantum circuit lower bounds. 

\paragraph{Quantum circuit lower bound via quantum natural proof}
By Observation~\ref{obs:mqcsp_nat_prop}, we know that $\MQCSP$ gives a $\class{BQP}$-quantum natural property. Then, we follow a recent work by Arunachalam et al.~\cite{agg20} and prove the following theorem:

\begin{restatable}{theorem}{cktlbfixed}\label{thm:ckt lb from MQCSP in BQP fixed}
If $\class{MQCSP}\in \class{BQP}$, then $\class{BQE}\not\subset\class{BQC}[n^k]$ for any constant $k\in\mathbb{N}_+$, where $\class{BQE}=\class{BQTIME}[2^{O(n)}]$.
\end{restatable}

\begin{remark}
A key difference between Theorem~\ref{thm:ckt lb from MQCSP in BQP fixed} and \cite{agg20} is that their circuit lower bound for $\class{BQE}$ is against \emph{classical} circuits, while ours is against \emph{quantum} circuits by proving a diagonalization lemma for quantum circuits.
\end{remark}

An ingredient of our proof is a conditional pesudorandom generator against uniform quantum computation. We first recall the definition of PRG against uniform quantum circuits given by \cite{agg20}.

\begin{definition}[Pesudorandom generator against uniform quantum circuit, \cite{agg20}]\label{def:prg_agg20}
A family of functions $\{G_n\}_{n\geq 1}$ is an infinitely often $(\ell, m, s, \epsilon)$-generator against uniform quantum circuits if the following properties holds:
\begin{enumerate}
    \item \textsf{Stretch: }$G_n: \{0,1\}^{\ell(n)}\rightarrow \{0,1\}^{m(n)}$.
    \item \textsf{Uniformity and efficiency: } There is a deterministic algorithm $A$ that when given $1^n$ and $x\in \{0,1\}^{\ell(n)}$ runs in time $O(2^{\ell(n)})$ and outputs $G_n(x)$.
    \item \textsf{Pseudorandomness: } For every deterministic algorithm $A$ such that when given $1^{m(n)}$ runs in time $s(m)$ and outputs a quantum circuit $C_m$ of size at most $s(m)$ computing a $m$-input Boolean function, for infinitely many $n\geq 1$, 
    \begin{align*}
        \left| \Pr_{x\sim \{0,1\}^{\ell(n)}, C_m}[C_m(G_n(x))=1]-\Pr_{y\sim \{0,1\}^{m(n)}, C_m}[C_m(y)=1] \right|\leq \epsilon(m).
    \end{align*}
\end{enumerate}
\end{definition}

\cite{agg20} constructed the following infinitely often PRG based on the assumption $\class{PSPACE} \nsubseteq \class{BQSUBEXP}$.
\begin{theorem}[Conditional PRG against uniform quantum computations, \cite{agg20}]\label{thm:bqp_PRG}
Suppose that $\class{PSPACE} \nsubseteq \class{BQSUBEXP}$. 
Then, for some choice of constants $\alpha \geq 1$ and $\lambda \in (0,1/5)$, there is an infinitely often $(\ell, m, s, \varepsilon)$-generator $G = \{G_n\}_{n \geq 1}$, where 
$\ell(n) \leq n^\alpha$, $m(n) = \lfloor 2^{n^\lambda} \rfloor$,  $s(m) = 2^{n^{2\lambda}} \geq \mathsf{poly}(m)$ \emph{(}for any polynomial\emph{)}, and $\varepsilon(m) = 1/m.$ 
\end{theorem}

Now, we are ready to prove the lower bound for $\class{BQE}$ based on the conditional PRG and a diagonalization theorem for quantum circuits.

\begin{proof}[Proof of Theorem~\ref{thm:ckt lb from MQCSP in BQP fixed}]
We use a win-win argument to prove the circuit lower bound. 

\paragraph{Case 1:} Suppose $\class{PSPACE}\subseteq \class{BQSUBEXP}$, i.e., for every $\gamma\in (0, 1] $, $\class{PSPACE}\subseteq \class{BQTIME}[2^{n^\gamma}]$. Then, for a fixed $k\in \N$, by a diagonalization lemma for quantum circuits (Claim~\ref{clm:qc_diag}), we know that there exists a language $L\in \class{PSPACE}$ such that $L\notin \class{BQC}[n^k]$. However, by the assumption, $L\in \class{BQE}$, which implies that $\class{BQE}\not\subset \class{BQC}[n^k]$.

\paragraph{Case 2:} $\class{PSPACE}\not \subseteq \class{BQSUBEXP}$, that is, there exists a language $L\in \class{PSPACE}$ and $\gamma >0$ such that $L\notin \class{BQTIME}[2^{n^\gamma}]$. By Theorem~\ref{thm:bqp_PRG}, for some $\alpha\geq 1, \lambda\in (0,1/5)$, there exists an infinitely often $(\ell,m,s,\epsilon)$-PRG ${\cal G}=\{G_n\}_{n\geq 1}$, where $\ell(n) = n^{\alpha}$, $m(n) = \lfloor 2^{n^\lambda}\rfloor$, $s(m)=\lfloor 2^{n^{2\lambda}}\rfloor$, $\epsilon(m)=1/m$.

For each $w\in \{0,1\}^{n^\alpha}$, we consider $G_n(w)$ as the truth table of Boolean function $\textsf{fnc}(G_n(w)):\{0,1\}^{d}\rightarrow \{0,1\}$, where $d:=\log (m(n))$ is the input length of the function. We will show that $\textsf{fnc}(G_n(w))$ is a hard function for $\class{BQC}[d^{O(1)}]$ for most $w\in \{0,1\}^{\ell(n)}$. 

Suppose that this is not true, i.e., there exists a $k>0$ such that for almost every $n>0$,  $\textsf{fnc}(G_n(w))\in \class{BQC}[d(n)^k]$ for a constant fraction of seeds $w\in \{0,1\}^{\ell(n)}$. Then, consider a quantum circuit $C^{\MQCSP}_m$ which takes a $m$-bit string $s$ and accepts it if and only if $\MQCSP(s,1^{d^k})=1$, where $s$ is the truth table and $d^k$ is the size parameter. Since we assume $\MQCSP\in \class{BQP}$, the quantum circuit $C^{\MQCSP}_m$ can be generated by a deterministic algorithm
in time $\poly(m)\leq s(m)$\footnote{For all problems in $\class{BQP}$, there exists a classical Turing machine that can efficiently uniformly generate the quantum circuits.}. This implies that 
\begin{align*}
    \Pr_{w\sim \{0,1\}^{\ell(n)}, C^{\MQCSP}_m}\left[C^{\MQCSP}_m(G_n(w))=1\right] \geq \delta
\end{align*}
for some constant $\delta\in (0, 1)$. On the other hand, by the pseudorandomness property of $G_n$ (part 3 in Definition~\ref{def:prg_agg20}), for infinitely many $n$, we have
\begin{align}\label{eq:prg_contradict}
    \left| \Pr_{w\sim \{0,1\}^{\ell(n)}, C^{\MQCSP}_m}\left[C^{\MQCSP}_m(G_n(w))=1\right]-\Pr_{y\sim \{0,1\}^{m(n)}, C^{\MQCSP}_m}\left[C^{\MQCSP}_m(y)=1\right] \right| \leq \frac{1}{m}.
\end{align}
However, only $o(1)$-fraction of random functions have polynomial-size quantum circuits, i.e., 
\begin{align*}
    \Pr_{y\sim \{0,1\}^{m(n)}, C^{\MQCSP}_m}\left[C^{\MQCSP}_m(y)=1\right]\leq o(1),
\end{align*}
which means Eq.~\eqref{eq:prg_contradict} cannot hold. Therefore, for infinitely many $n$, and almost all $w$, the function $\textsf{fnc}(G_n(w))\notin \class{BQC}[n^k]$ for every $k\in \N_+$.

Therefore, we can construct a hard language $L^{\cal G}$ as follows:
\begin{itemize}
    \item For any $n>0$ and every $x\in \{0,1\}^n$, check if $x$ can be written as $(w, y)$, where $|w|=\ell(t)$ and $|y|=\lceil\log m(t)\rceil$ for some $t\in \N$. 
    \item If not, then $L^{\cal G}(x): = 0$.
    \item Otherwise, $L^{\cal G}(x):=\textsf{fnc}(G_t(w))(y)$.
\end{itemize}

We first show that $L^{\cal G}\in \class{BQE}$. By the running time property of $G_n$ (part 2 in Definition~\ref{def:prg_agg20}), $G_n(w)$ can be computed in deterministic time $O(2^{\ell(t)})\leq O(2^n)$. Hence, $L^{\cal G}\in \class{E}\subset \class{BQE}$.

Then, we show that $L^{\cal G}\notin \class{BQC}[n^{k}]$ for every $k\in \N_+$. Fix $k> 0$. Suppose there exists a quantum circuit family $\{C_n\}_{n\geq 1}$ that computes $L^{\cal G}$ and $C_n$ has size $n^k$ for every $n\geq 1$. However, we already proved that there exists an infinite-size subset $\{\mathcal{S}\subset \N\}$ such that for $n\in \mathcal{S}$, there exists many ``hard seed'' $w_n$ such that
\begin{align}\label{eq:L_g_hard_cond}
    \textsf{fnc}(G_t(w_n))\notin \class{BQC}[t^{2\alpha k}].
\end{align}
Then, for any $n\in S$ and any $w_n$ that makes Eq.~\eqref{eq:L_g_hard_cond} hold, define a new quantum circuit family $\{C\downharpoonright_{w_n}\}_{n\geq 1}$ such that $C\downharpoonright_{w_n}(y):=C(w_n, y)$, i.e., $C\downharpoonright_{w_n}$ computes the hard function $\textsf{fnc}(G_t(w_n))$. Hence, $C\downharpoonright_{w_n}$ must have size larger than $t^{2\alpha k}$.  Since $n=\ell(t) + \log m(t) = t^{\alpha} + t^{\lambda}\leq t^{2\alpha}$, and the size of $C_n$ should be least the size of its restriction $C\downharpoonright_{w_n}$, we conclude that $C_n$ has size larger than $n^k$ for these infinitely many $n\in {\cal S}$.  Therefore, the $\class{BQE}$ language $L^{\cal G}\notin \class{BQC}[n^k]$, which implies $\class{BQE}\not\subset \class{BQC}[n^k]$.

Combining Case 1 and 2 completes the proof of the theorem.
\end{proof}

\paragraph{Circuit lower bound for $\class{BQP}^{\class{QCMA}}$}
Our second result shows that if $\MQCSP\in \class{BQP}$, then $\class{BQP}^{\class{QCMA}}$ cannot be computed by polynomial-size quantum circuits. Our result follows the seminal work of Kabanets and Cai~\cite{KC00}, which showed a circuit lower bound for $\class{P}^{\class{NP}}$ based on $\MCSP$ is easy.  More specifically, we consider the following ``hard problem'':

\begin{definition}[Maximum quantum circuit complexity problem]
The input of this problem is $1^n$ for $n\in \N_+$. The output is the truth table of a function $f:\{0,1\}^n\rightarrow \{0,1\}$ such that for any $f':\{0,1\}^n\rightarrow \{0,1\}$, the quantum circuit complexity $\mathrm{qCC}(f)\geq \mathrm{qCC}(f')$. 
\end{definition}

We first prove that $\class{BPE^{QCMA}}$ can solve the maximum quantum circuit complexity problem, which implies that $\class{BPE^{QCMA}}$ contains the hardest Boolean function. Then, by the standard padding  argument, we can show quantum circuit lower bound for $\class{BQP^{QCMA}}$.
\begin{restatable}{theorem}{cktlbmax}\label{thm:ckt lb from MQCSP in BQP max}
If $\class{MQCSP}\in \class{BQP}$, then $\class{BPE^{QCMA}}$ contains a function with maximum quantum circuit complexity. Furthermore, $\class{BQP^{QCMA}}\not\subset \class{BQC}[n^k]$ for any constant $k>0$.
\end{restatable}

We note that there are two subtle differences between Theorem~\ref{thm:ckt lb from MQCSP in BQP max} and \cite{KC00}'s result:
\begin{itemize}
    \item We need a $\class{QCMA}$ oracle while \cite{KC00} used an $\class{NP}$ oracle. This is because we assume that $\MQCSP\in \class{BQP}$. In order to decide the maximum quantum circuit complexity, we can non-deterministically guess a truth table and use the $\class{BQP}$ algorithm to verify its quantum circuit complexity. This process can be achieved by an $\class{QCMA}$ oracle.
    \item Another difference is that we consider the $\class{BPE}$ class while \cite{KC00} considered the $\class{E}$ class. This is because our $\class{QCMA}$ oracle can only output correct answers with high probability. Thus, the whole algorithm will be a randomized algorithm. 
\end{itemize} 

The formal proof is deferred to Section~\ref{sec:mqcsp_bpe_qcma}.

\paragraph{Hardness amplification using $\class{MQCSP}$}
\cite{KC00} showed that the classical $\class{MCSP}$ can be used for hardness amplification, i.e., given one very hard Boolean function, there exists an efficient algorithm to find many hard functions via an $\class{MCSP}$ oracle. We show that it also holds for quantum circuits:
\begin{restatable}{theorem}{ampckt}\label{thm:mqcsp_hardness_amp}
Assume $\class{MQCSP}\in \class{BQP}$. Then, there exists a $\class{BQP}$ algorithm that, given the truth table of an $n$-variable Boolean function of quantum circuit complexity $2^{\Omega(n)}$, outputs $2^{\Omega(n)}$ Boolean functions on $m=\Omega(n)$ variables each, such that all of the output functions have quantum circuit complexity greater than $\frac{2^m}{(c+1)m}$ for any $c>0$.     
\end{restatable}

In order to prove Theorem~\ref{thm:mqcsp_hardness_amp}, we first construct a ``quantum version'' of the Impagliazzo-Wigderson generator \cite{iw97}. We note that the construction in the following lemma is stronger than the Definition~\ref{def:prg_agg20}, based on the truth table of a very hard function.
\begin{restatable}[Quantum Impagliazzo-Wigderson generator]{lemma}{quantumiwprg}\label{lem:quantum_IW_prg}
For every $\epsilon>0$, there exist $c,d\in \N$ such that the truth table of a Boolean function $f:\{0,1\}^{cn}\rightarrow \{0,1\}$ of quantum circuit complexity $2^{\epsilon cn}$ can be transformed in time $O(2^n)$ into a pseudorandom generator $G:\{0,1\}^{dn}\rightarrow \{0,1\}^{2^n}$ running in time $O(2^n)$ that can fool quantum circuits of size $2^{O(n)}$,
i.e., for any $p>0$, any quantum circuit $\cal C$ of size at most $2^{pn}$,
\begin{align*}
    \left|\Pr_{x\sim \{0,1\}^{dn}, {\cal C}}[{\cal C}(G(x))=1] - \Pr_{y\sim \{0,1\}^{2^n}, {\cal C}}[{\cal C}(y)=1]\right|\leq 2^{-n}.
\end{align*}
\end{restatable}

\begin{proof}[Proof of Theorem~\ref{thm:mqcsp_hardness_amp}]
Let $c>0$ and $s(n) = \frac{2^n}{(c+1)n}$. Assuming that $\class{MQCSP}\in \class{BQP}$, we get a polynomial-size quantum circuit family $\{{\cal D}_n\}$ that only accept $n$-variable Boolean functions of quantum circuit complexity greater than $s(n)$. By Claim~\ref{claim:gatecount}, the acceptance probability is close to one.

However, the size of ${\cal D}_n$ is bounded by a fixed polynomial in the input size, by Lemma~\ref{lem:quantum_IW_prg}, the quantum Impagliazzo-Wigderson generator $G$ will fool ${\cal D}_n$. That is, almost all $2^n$-bit strings output by $G$ will have quantum circuit complexity greater than $s(n)$. We can then use the $\class{MQCSP}$ circuit to decide the quantum circuit complexity of these strings and only output hard functions.    
\end{proof}

The proof of Lemma~\ref{lem:quantum_IW_prg} relies on a quantum-secure direct product generator and several hardness amplification steps. It is deferred to Section~\ref{sec:mqcsp_amp}.

\paragraph{Hardness magnification for $\class{MQCSP}$.}
Hardness magnification refers to a transformation of a weak circuit lower bound (e.g., linear size lower bound) to a stronger circuit lower bound (e.g., polynomial size lower bound). Note that a magnification theorem for a circuit class is highly dependent on the structure of the circuits. Specifically, it is not immediately clear that every circuit class is \textit{magnifiable}. Here, we show that there exists hardness magnification for quantum circuits when it comes to $\MQCSP$.

\begin{restatable}{theorem}{magnification}\label{thm:magnification}
    If $\class{MQCSP}\left[2^{n^{1/2}}/2n, 2^{n^{1/2}}\right]$ is hard for $\class{BQC}\left[2^{n+O(n^{1/2})}\right]$, then $\class{QCMA}\not\subseteq \class{BQC}[\poly(n)]$.
\end{restatable}

The proof of Theorem~\ref{thm:magnification} is via antichecker lemma, which was first given by \cite{ops19,chopr20} for proving hardness magnification for $\class{MCSP}$.

\begin{restatable}[Antichecker lemma for quantum circuits]{lemma}{antichecker}\label{lem:antichecker}
Assume $\class{QCMA} \subseteq \class{BQC}[\poly]$. Then for any $\lambda \in (0, 1)$ there are circuits $\{C_{2^n}\}_{n=1}^\infty$ of size $2^{n+O(n^{\lambda})}$ which given the truth table $\textsf{tt}(f) \in \{0, 1\}^{2^n}$ , outputs $2^{O(n^\lambda)}$ $n$-bit strings $y_1,\dots,y_{2^{O(n^\lambda)}}$ together with bits $f(y_1),\dots, f(y_{2^{O(n^\lambda)}})$ forming a set of anticheckers for $f$, i.e. if $f$ is hard for quantum circuits of size $2^{n^\lambda}$ then every quantum circuit of size $2^{n^\lambda}/2n$ fails to compute $f$ on one of the inputs $y_1,\dots,y_{2^{O(n^\lambda)}}$.
\end{restatable}
With Lemma~\ref{lem:antichecker}, we can prove Theorem~\ref{thm:magnification} by using a small quantum circuit to verify the given circuits only on the anticheckers.
\begin{proof}[Proof of Theorem~\ref{thm:magnification}]
Suppose $\class{QCMA}\subseteq \class{BQC}[\poly]$. Let $\textsf{tt}(f)$ be the input of $\class{MQCSP}[2^{n^{1/2}}/2n, 2^{n^{1/2}}]$. By Lemma~\ref{lem:antichecker}, we can find a set of anticheckers $y_1,\dots,y_{2^{O(n^{1/2})}}$ by a quantum circuit of size $2^{n+O(n^{1/2})}$. Then, we use a $\class{QCMA}$ algorithm to decide if there exists a quantum circuit of size $2^{n^\lambda}/2n$ that computes $f$ correctly on $\{(y_1,f(y_1)), \dots,(y_{2^{O(n^\lambda)}}, f(y_{2^{O(n^\lambda)}}))\}$. By the assumption, it can be done by a $2^{O(n^\lambda)}$ size quantum circuit. Then, there are two cases:
\begin{itemize}
    \item If the $\class{QCMA}$ algorithm returns ``Yes'', it means that $y_1,\dots,y_{2^{O(n^{1/2})}}$ are not anticheckers. By Lemma~\ref{lem:antichecker}, $f$ is \emph{not} hard for $2^{n^{1/2}}$ size quantum circuit.
    \item If the $\class{QCMA}$ algorithm returns ``No'', then no $2^{n^{1/2}}/2n$ size quantum circuit can compute $f$ on $y_1,\dots,y_{2^{O(n^{1/2})}}$. So, $f$ is hard for $2^{n^{1/2}}/2n$ size quantum circuit. 
\end{itemize}
Hence, $\class{MQCSP}[2^{n^{1/2}}/2n, 2^{n^{1/2}}]\in \class{BQC}[2^{n+O(n^{1/2})}]$.
\end{proof}

The proof of Lemma~\ref{lem:antichecker} is deferred to Section~\ref{sec:magnify}.

\subsection{Fine-grained complexity}\label{sec:main fine-grained}

It is a long-standing open problem to show the hardness of $\class{MCSP}$ based on some fine-grained complexity hypotheses, like the Exponential-Time Hypothesis (\textsf{ETH}), which was conjectured by Impagliazzo,  Paturi, and Zane \cite{ipz01} and becomes a widely used assumption in fine-grained complexity area. 
\begin{definition}[Exponential Time Hypothesis (\textsf{ETH})]
There exists $\delta>0$ such that 3-SAT with $n$ variables cannot be solved in time $2^{\delta n}$.
\end{definition}
Very recently, a breakthrough result by Ilango \cite{ila20} proved the \textsf{ETH}-hardness of $\class{MCSP}$ for partial Boolean functions. On the other hand, Quantum fine-grained complexity was studied very recently by \cite{aclwz20,bps21,al20,gs20}. Motivated by the fact that currently there is no quantum algorithm for 3-\textsf{SAT} that is significantly faster than Grover's search, we conjecture that 3-\textsf{SAT} with $n$ variables cannot be solved in $2^{o(n)}$ quantum time (\textsf{QETH}). And based on \textsf{QETH}, we want show that $\class{MQCSP}$ for partial Boolean function is also hard. 

We first formally define \textsf{QETH} and $\MQCSP$ for partial functions ($\mqcsps$).
\begin{definition}[Quantum Exponential Time Hypothesis (\textsf{QETH})]
There exists $\delta'>0$ such that 3-SAT with $n$ variables cannot be solved in time $2^{\delta' n}$ in quantum.
\end{definition}

\begin{definition}[\textsf{MQCSP} for partial functions ($\mqcsps$)]\label{def:partial_mqcsp}
The input is the truth table $\{0,1,\star\}^{2^n}$ of a partial function $f: \{0,1\}^n \rightarrow \{0,1,\star\}$ and an integer parameter $s$. The goal is to decide whether there exists a quantum circuit $C$ of size at most $s$ (using single-qubit and 2-qubit gates) that computes $f$. That is, for all $x\in \{0,1\}^n$ such that $f(x)\ne \star$, we have
\begin{align*}
    \Pr[C(x)=f(x)]\geq \frac{2}{3}.
\end{align*}
\end{definition}

Our main result of this section is as follows:

\begin{restatable}[\textsf{QETH}-hardness of $\mqcsps$]{theorem}{qethhardness}\label{thm:qeth_hard_mqcsps}
$\mqcsps$ cannot be solved in $N^{o(\log \log N)}$-time quantumly on truth tables of length $N$ assuming \textsf{QETH}. 
\end{restatable}

Our reduction reveals the connections between $\class{MQCSP}^\star$, quantum read-once formula and classical read-once formula. The proof is given in Section~\ref{sec:fine grained}.

\paragraph{Classical reduction for $\MCSP^\star$.}

We first give a brief overview of the classical reduction for $\MCSP^\star$ in \cite{ila20}. They reduced $\mcsps$ to a fine-grained problem: $2n\times 2n$ \textit{Bipartite Permutation Independent Set problem}, which is defined as follows: 
\begin{definition}[Bipartite Permutation Independent Set problem]
A $2n\times 2n$ bipartite permutation independent set problem is defined on a directed graph $G$ with vertex set $[n]\times [n]$ and edge set $E$. The goal is to decide whether there exists a permutation $\pi\in \mathcal{S}_{2n}$ such that
\begin{itemize}
    \item $\pi([n])= [n]$,
    \item $\pi(\{n+i:i\in [n]\})=\{n+i:i\in [n]\}$,
    \item if $((j,k), (j',k'))\in E$, then either $\pi(j)\ne k$ or $\pi(n+j')\ne \pi(n+k')$.
\end{itemize}
\end{definition}

Lokshtanov, Marx, and Saurabh \cite{lms11} proved that this problem is $2^{o(n\log n)}$-hard under \textsf{ETH}, which implies the \textsf{ETH}-hardness of $\mcsps$.

The reduction from $2n\times 2n$ bipartite permutation independent set problem to $\mcsps$ is via the following partial function $\gamma$. Consider an instance $G=([n]\times [n], E)$ of $2n\times 2n$ bipartite permutation independent set problem. The reduction outputs the truth table of a partial Boolean function $\gamma: \{0,1\}^{2n}\times \{0,1\}^{2n}\times \{0,1\}^{2n}\rightarrow \{0,1,\star\}$ such that
\begin{align}\label{eq:gamma_def}
    \gamma(x,y,z):=\begin{cases}
        \bigvee_{i\in [2n]} (y_i \wedge z_i) & \text{if}~x=0^{2n},\\
        \bigvee_{i\in [2n]} z_i & \text{if }x=1^{2n},\\
        \bigvee_{i\in [2n]} (x_i\vee y_i) & \text{if }z=1^{2n},\\
        0 & \text{if }z=0^{2n},\\
        \bigvee_{i\in [n]} x_i & \text{if }z=1^n0^n\text{ and }y=0^{2n},\\
        \bigvee_{i\in \{n+1,\cdots,2n\}} x_i & \text{if }z=0^n1^n\text{ and }y=0^{2n},\\
        1 & \text{if }\exists ((j,k),(j',k'))\in E ~\text{s.t.}~(x,y,z)=(\overline{e_ke_{k'}},0^{2n},e_je_{j'}),\\
        \star & \text{otherwise}.
    \end{cases}
\end{align}
In particular, the small circuit size of $\gamma$ implies that $G$ is a ``Yes'' instance of the bipartite permutation independent set problem:
\begin{lemma}[\cite{ila20}]\label{lem:eth_reduction}
Each of the following are equivalent:
\begin{enumerate}
    \item $\mcsps(\gamma, 6n-1)=1$;
    \item $\gamma$ can be computed by a read-once formula;
    \item there exists a $\pi\in {\cal S}_{2n}$ such that $\bigvee_{i\in [2n]} ((x_{\pi(i)}\vee y_i) \wedge z_i)$ computes $\gamma$;
    \item there exists a $\pi\in {\cal S}_{2n}$ that satisfies the instance of bipartite permutation independent set problem given by $G$.
\end{enumerate}
\end{lemma}

\paragraph{Quantum reduction for $\MQCSP^\star$}
We follow the proof in \cite{ila20} but adapt it to  quantum circuits. More specifically, we want to show that for the partial function $\gamma$ defined by Eq.~\eqref{eq:gamma_def}, $\mqcsps(\gamma, 6n-1)=1$ is equivalent to the case that $\gamma$ can be computed by a read-once formula.

The reverse direction is easy:
\begin{claim}\label{clm:formula_to_mqcsp}
If $\gamma$ can be computed by a read-once formula, then $\mqcsps(\gamma, 6n-1)=1$.
\end{claim}
\begin{proof}
It is easy to see that a read-once formula on $6n$ input variables has at most $6n-1$ Boolean gates. Hence, it implies that $\mcsps(\gamma, 6n-1)=1$. Then, we have $\mqcsps(\gamma, 6n-1)=1$ because we can use a quantum circuit with all 2-qubit gates to simulate a Boolean circuit without increasing the circuit size. 
\end{proof}

For the forward direction, we consider an intermediate  model: \textit{read-once quantum formula}. The quantum formula was defined by Yao \cite{yao93} as follows:
\begin{definition}
A quantum formula is a single-output quantum circuit such that every gate has at most one output that is used as an input to a subsequent one. 

If a quantum formula only uses every input qubit at most once, then we say it is a read-once quantum formula.
\end{definition}

We first prove the forward direction for the quantum read-once formula:
\begin{claim}\label{clm:mqcsp_to_qformula}
If $\mqcsps(\gamma, 6n-1)=1$, then $\gamma$ can be computed by a read-once quantum formula. Here, we assume that the quantum circuits only use single-qubit and 2-qubit gates.
\end{claim}

\begin{proof}
It is easy to verify that $\gamma$ depends on all of the $6n$ input variables. Hence, by a light-cone argument, the topology of the quantum circuit that computes $\gamma$ using $6n-1$ 2-qubit gates must be a full binary tree with $6n$ leaves. Hence, that circuit is a read-once quantum formula.
\end{proof}

Cosentino, Kothari, and Paetznick \cite{ckp13} proved that any read-once quantum formula can be ``dequantized'' to the classical read-once quantum formula:
\begin{theorem}[\cite{ckp13}]\label{thm:dequantize_readonce}
If a language is accepted by a bounded-error read-once quantum formula over single-qubit and 2-qubit gates, then it is also accepted by an exact read-once classical formula with the same size, using \textsf{NOT} and all 2-bit Boolean gates. 
\end{theorem}

Hence, we can apply Theorem~\ref{thm:dequantize_readonce} to dequantize Claim~\ref{clm:mqcsp_to_qformula}:

\begin{claim}\label{clm:mqcsp_to_xorformula}
If $\mqcsps(\gamma, 6n-1)=1$, then $\gamma$ can be computed by a classical read-once formula with $6n-1$ 2-bit gates. In particular, all the \textsf{NOT} gates can be pushed to the leaf level and the high level gates are $\{\textsf{AND}, \textsf{OR}, \textsf{XOR}\}$.
\end{claim}
\begin{proof}
By Theorem~\ref{thm:dequantize_readonce}, there is a read-once classical formula that computes $\gamma$ using $6n-1$ 2-bit logical gates. We can enumerate all of the 2-bit Boolean function and check that they can be expressed by one of $\textsf{AND}, \textsf{OR}, \textsf{XOR}$ gate with some $\textsf{NOT}$ gates on the input wire. Then, by De Morgan's laws, we can push the $\textsf{NOT}$ gate to the bottom level. Note that these transformations will preserve the read-once property. 
\end{proof}


The next claim shows that $\textsf{NOT}$ and $\textsf{XOR}$ gates do not help computing $\gamma$:

\begin{claim}\label{clm:xorformula_to_formula}
The classical read-once formula computing $\gamma$ only uses $\textsf{AND}$ and $\textsf{OR}$ gates.
\end{claim}
\begin{proof}
The proof is similar to the proof of Claim 13 in \cite{ila20}.

We first note that the $\textsf{XOR}$ gate is not monotone. Then, by setting $x=0^{2n}$, we have $\gamma(0^{2n}, y, z)=\bigvee_{i\in [2n]}(y_i \wedge z_i)$, which is a monotone function in $y$ and $z$. Hence, the $\textsf{XOR}$ gates in the formula cannot depend on the all the $y$ and $z$ variables. Similarly, by setting $z=1^{2n}$, we have $\gamma(x, y, 1^{2n}) = \bigvee_{i\in [2n]} (x_i \vee y_i)$, which is monotone in $x$ and $y$. It implies that the $\textsf{XOR}$ gates cannot depend on all the $x$ variables. Hence, the formula will not use the $\textsf{XOR}$ gate.

For the \textsf{NOT} gate, since the function is monotone in the positive input variables after some restrictions, and the formula is read-once, the $\textsf{NOT}$ gate will also not be used. 
\end{proof}

By Claim~\ref{clm:formula_to_mqcsp}, \ref{clm:mqcsp_to_xorformula} and \ref{clm:xorformula_to_formula}, we get that $\mqcsps(\gamma, 6n-1)=1$ is equivalent to the case that $\gamma$ can be computed by a read-once formula using $\textsf{AND}$ and $\textsf{OR}$ gates. This statement corresponds to showing that $\text{(1)}~\Leftrightarrow~\text{(2)}$ in Lemma~\ref{lem:eth_reduction} for $\mcsps$. Then, by $\text{(2)}~\Leftrightarrow~\text{(4)}$ in Lemma~\ref{lem:eth_reduction}, we prove the following reduction for $\mqcsps$:
\begin{lemma}\label{lem:mqcsp_reduction}
$\mqcsps(\gamma, 6n-1)=1$ is equivalent to the existence of  $\pi\in {\cal S}_{2n}$ that satisfies the instance of bipartite permutation independent set problem given by $G$.
\end{lemma}

The remaining thing is to prove the quantum hardness of the $2n\times 2n$ ~\textsf{Bipartite Permutation Independent Set} problem. We follow the quantum fine-grained reduction framework by \cite{aclwz20} and show the following $\class{QETH}$-hardness result. The proof is given in Section~\ref{sec:fine grained}.

\begin{restatable}{lemma}{qethperm}\label{lem:qeth_perm_ind_set}
Assuming \textsf{QETH}, there is no $2^{o(n\log n)}$-time quantum algorithm that solves $2n\times 2n$ \textsf{Bipartite Permutation Independent Set} problem. 
\end{restatable}

Now, we can prove the $\textsf{QETH}$-hardness of $\mqcsps$:
\begin{proof}[Proof of Theorem~\ref{thm:qeth_hard_mqcsps}]
By Lemma~\ref{lem:mqcsp_reduction}, $\mqcsps$ can be reduced to $2n\times 2n$ \textsf{Bipartite Permutation Independent Set} problem and the hardness follows from Lemma~\ref{lem:qeth_perm_ind_set}.
\end{proof}

%% file: crypto.tex
\subsection{Cryptography and MQCSP}\label{app:crypto}
Classically, we have already known connections between $\class{MCSP}$ and one-way functions~\cite{KC00,RR97} and indistinguishable obfuscation~\cite{impagliazzo2018power}. In this section, we show the quantum analogies of these results.  

\subsubsection{Quantum cryptographic primitives}
\label{sec:crypto}
We first introduce relevant primitives in cryptography. 

\begin{definition}[Pseudorandom Generator ($\class{PRG}$)]
Let $G: \{0,1\}^* \rightarrow \{0,1\}^*$ be a polynomial-time computable function. Let $\ell: \N\rightarrow \N$ be a polynomial-time computable function such that $\ell(n)> n$ for all $n$. $G$ is a pseudorandom generator of stretch $\ell(n)$ if it satisfies: 
\begin{enumerate}
    \item $|G(x)| = \ell(|x|)$ for all $x\in \{0,1\}^*$, and 
    \item for all Probabilistic polynomial-time (PPT) algorithm $\A$, there exists a negligible function $\epsilon: \N\rightarrow [0,1]$ such that for all $n\in \N$ 
    \begin{align*}
        \left|\Pr_{x\sim\{0,1\}^n}[\A(G(x)) = 1] - \Pr_{y\sim \{0,1\}^{\ell(n)}}[\A(y)=1]\right| \leq \epsilon(n).
    \end{align*}
\end{enumerate}
\end{definition}

We say that a $\class{PRG}$ is \emph{local} if every output bit of the $\class{PRG}$ can be computed in time $\poly(n)$. In the following, we define $\class{PRG}$ secure against any quantum polynomial-time adversary.

\begin{definition}[Quantum-Secure Pseudorandom Generator ($\class{qPRG}$)]\label{def:qprg}
Let $G: \{0,1\}^* \rightarrow \{0,1\}^*$ be a polynomial-time computable function\footnote{It is worth noting that $G$ can be any function that is efficiently computable in either quantum or classical polynomial time.}. Let $\ell: \N\rightarrow \N$ be a polynomial-time computable function such that $\ell(n)> n$ for all $n$. $G$ is a pseudorandom generator secure against quantum adversaries of stretch $\ell(n)$ if it satisfies: 
\begin{enumerate}
    \item $|G(x)| = \ell(|x|)$ for all $x\in \{0,1\}^*$, and 
    \item for all quantum polynomial-time (QPT) algorithm $\A$, there exists a negligible function $\epsilon: \N\rightarrow [0,1]$ such that for all $n\in \N$ 
    \begin{align*}
        \left|\Pr_{x\sim\{0,1\}^n}[\A(G(x)) = 1] - \Pr_{y\sim \{0,1\}^{\ell(n)}}[\A(y)=1]\right| \leq \epsilon(n).
    \end{align*}
\end{enumerate}
\end{definition}  

In this work, we consider two ways of constructing quantum-secure $\class{PRG}$s based on different cryptographic primitives. One is based on the quantum-secure one-way functions and the other one is based on the hard function. 

\begin{restatable}[Quantum-Secure One-Way function ($\class{qOWF}$)]{definition}{qOWF}\label{def:qowf}
A function $f:\{0,1\}^* \rightarrow \{0,1\}^*$ is a quantum-secure one-way function, if the following conditions hold: For every $n\in \N$, for any $x\in \{0,1\}^n$ picked uniformly at random, 
\begin{enumerate}
    \item There exists a $\poly(n)$-time deterministic algorithm for computing $f$. 
    \item For any $\poly(n)$-time quantum algorithm $\A'$, $\Pr_x[\A'(f(x))\in f^{-1}(f(x))] = \negl(n)$. 
\end{enumerate}
\end{restatable}

\begin{definition}[GGM Construction~\cite{GGM86}]\label{def:ggm}
Let $G:\{0,1\}^n\rightarrow \{0,1\}^{2n}$ be a $\class{(q)PRG}$. For every $z\in \{0,1\}^{m}$, the GGM construction of a pseudorandom function family $\{h_z: \{0,1\}^n\rightarrow \{0,1\}^n\}_{z\in \{0,1\}^m}$ is defined as follows:
\begin{align*}
    f_z(x) = G_{z_{m}}\circ G_{z_{m-1}}\circ\cdots\circ G_{z_{1}}(x),  
\end{align*}
where we denote by $G_0(x)$ the first $n$ bits of $G$, and by $G_1(x)$ the last $n$ qubits. 
\end{definition}

\begin{lemma}[\cite{HILL99}]\label{owf2prg}
If $\class{OWF}$s exist, then for every $c\in \N$, there exists a secure $\class{PRG}$ with stretch $\ell(n)=n^c$.  
\end{lemma}

Since the security proof of Lemma~\ref{owf2prg} is black-box, the analysis carries over to the quantum setting directly if the one-way function is secure against quantum adversaries.  Therefore, we can obtain Lemma~\ref{qowf2qprg}.  

\begin{lemma}[Folklore]\label{qowf2qprg}
If $\class{qOWF}$s exist, then for every $c\in \N$, there exist $\class{qPRG}$s with stretch $\ell(n)=n^c$. 
\end{lemma}

\begin{lemma}\label{lem:qprg2lqprg}
Suppose that there exists a $\class{qPRG}$ $G:\{0,1\}^n\rightarrow \{0,1\}^{2n}$. Then, for $m=O(\log n)$, there exists a local $\class{qPRG}$ $\hat{G}: \{0,1\}^n \rightarrow \{0,1\}^{2^m}$.
\end{lemma}

\begin{proof}

We first give the construction of $\hat{G}$. Follow the GGM construction in Definition~\ref{def:ggm}, we let 
\begin{align*}
    h'_x(z) = G_{z_{m}}\circ G_{z_{m-1}}\circ\cdots\circ G_{z_{1}}(x)
\end{align*}
where $z\in \{0,1\}^m$, $x\in \{0,1\}^n$. We let $h_x(z)$ be the first output bit of $h'_x(z)$ and define the $\class{qPRG}$ as 
\begin{align*}
    \hat{G}(x) = h_x(0)~|~h_x(1)~|~\cdots~|~h_x(2^m-1). 
\end{align*}
It is obvious that each bit of $\hat{G}(x)$ can be computed in time $m$ times the runtime of $G$.

We then prove that $\hat{G}(x)$ is indistinguishable from a truly random string by the standard hybrid approach. For $i\in [m]$, we define 
\begin{align*}
    H^{i}(z) = (G_{z_{m}}\circ G_{z_{m-1}}\circ\cdots \circ G_{z_{i}}(y_{z,i}))_1, 
\end{align*}
where $y_{z,i}$ is drawn independently and uniformly randomly from $\{0,1\}^n$. Note that $H^1(z) = h_z(x)$ and $H^m(z)$ is a random bit. Let 
\begin{align*}
     \hat{G}^{i} = H^i(0)~|~H^i(1)~|~\cdots~|~H^i(2^m-1)\quad \forall i\in[m]. 
\end{align*}
Suppose that there exists a QPT algorithm $\A$ such that
\begin{align*}
    \left|\Pr_{x\sim \{0,1\}^n}[\A(\hat{G}(x))=1] - \Pr_{u\sim \{0,1\}^{2^m}}(\A(u))\right|\geq 1/\poly(n). 
\end{align*}
Then, by the triangular inequality, 
\begin{align*}
    \sum_{i=1}^{m-1} \left|\Pr[\A(\hat{G}^i)=1] - \Pr[\A(\hat{G}^{i+1})=1]\right|\geq 1/\poly(n)
\end{align*}
which implies that there exists $i^*$ such that $|\Pr[\A(\hat{G}^{i^*})=1] - \Pr[\A(\hat{G}^{i^*+1})=1]|\geq 1/\poly(n)$. Since distinguishing $\hat{G}^{i^*}$ and $\hat{G}^{i^*+1}$ implies that one can distinguish $G(x)$ from a random string, $G$ is not a $\class{qPRG}$. This completes the proof.   

\end{proof}

\subsubsection{Implications for quantum-secure one-way functions (qOWF)}\label{subsec:mcsp_owf}
Here, we show a quantum analogous result for~\cite{KC00,RR97} by considering the implication of the existence of efficient quantum algorithms for either classical or quantum $\class{MCSP}$.

\begin{restatable}{theorem}{mqcspqowf}\label{thm:mqcsp_qowf}
If $\class{MQCSP}\in \class{BQP}$, then there is no quantum-secure one-way function ($\class{qOWF}$). 
\end{restatable}

\begin{proof}
Let $f:\{0,1\}^* \rightarrow \{0,1\}^*$ be any function. By Lemma~\ref{qowf2qprg}, we construct $G_f: \{0,1\}^n\rightarrow \{0,1\}^{n^a}$ that is a $\class{qPRG}$ if $f$ is a $\class{qOWF}$. We denote the runtime for $G_f$ as $O(n^b)$ for some constant $b$. 

Given $G_f$, we construct a $\class{qPRG}$ $\hat{G}:\{0,1\}^n\rightarrow \{0,1\}^{2^m}$ where $m=O(\log n)$ by Lemma~\ref{lem:qprg2lqprg}. Then, we view the outputs of $\hat{G}(x)$ as a truth table of some Boolean function $g_x:\{0,1\}^{m}\rightarrow \{0,1\}$. Note that according to the construction in Lemma~\ref{lem:qprg2lqprg}, the time for evaluating $g_x$ on $z\in \{0,1\}^m$ is $O(m\cdot n^b) = \tilde{O}(n^b)$. On the other hand, for a random Boolean function from $\{0,1\}^m$ to $\{0,1\}$, we know from Claim~\ref{claim:gatecount} that its circuit complexity is greater than $\frac{2^m}{(c+1)m}$ with high probability. Therefore, by setting $m = d\log n$ for some constant $d\gg b$, the circuit complexity of the random function is $\tilde{O}(n^d) \gg \tilde{O}(n^b)$ with high probability. 

\begin{algorithm}[H]
\caption{A quantum algorithm for breaking $\class{qPRG}$}
	\label{alg:break_prg}
    \begin{algorithmic}[1]
	    \Input Given $\T(h)$ for $h:\{0,1\}^m\rightarrow \{0,1\}$ constructed from $\hat{G}$ in Lemma~\ref{lem:qprg2lqprg}.
        \State Runs the quantum algorithm for $\MQCSP$ with $s = \frac{2^m}{(c+1)m}$
        \State \Return ``Yes'' if the algorithm in previous step outputs yes.
        \State \Return ``No'', otherwise 
    \end{algorithmic}
\end{algorithm}

Since we assume $\MQCSP\in \class{BQP}$, we obtain a quantum polynomial-time algorithm $\A$ for distinguishing $\{g_x\}_{x\in \{0,1\}^n}$ and the random function family $\mathcal{F}_m$ as in Algorithm~\ref{alg:break_prg}. The circuit complexity for $g_x$ is at most $\tilde{O}(n^b)$ and the for a random function $h$ is greater than $\frac{2^m}{(c+1)m} = \tilde{O}(n^d)$ for $d\gg b$. thus, we obtain
\begin{align*}
    \left|\Pr_{x\sim \{0,1\}^n}[\A(\T(g_x))=1] - \Pr_{h\sim \mathcal{F}_m}[\A(\T(h))=1]\right| \geq 1/\poly(n). 
\end{align*}
This implies that we can use $\A$ to break $G$ in quantum polynomial time by Lemma~\ref{lem:qprg2lqprg}. Finally, by Lemma~\ref{qowf2qprg}, we obtain a quantum polynomial-time algorithm $\A_{inv}$ for inverting any $f$.  
\end{proof}

\subsubsection{Implication for quantum-secure \texorpdfstring{$\iO$}{iO}}
In this section, we use Theorem~\ref{thm:mqcsp_qowf} and quantum-secure $\iO$ to show that if $\class{MQCSP}$ can be solved by a $\class{BQP}$ algorithm, then $\class{NP}\subset \class{coRQP}$, which is the class of one-sided error quantum polynomial-time algorithms such that a ``Yes'' instance will always be accepted while a ``NO'' instance will be rejected with high probability.

We define the quantum-secure $\iO$ as follows: 
\begin{restatable}[Quantum-secure indistinguishability obfuscation, $\iO$]{definition}{qiO}
A probabilistic polynomial-time machine iO is an indistinguishability obfuscator for a circuit class $\{{\cal C}_\lambda\}_{\lambda\in \N}$ if the following conditions are satisfied for all $\lambda \in \N$:
\begin{itemize}
    \item \textbf{Functionality: }For any $C\in {\cal C}_\lambda$, for all inputs $x$, $\iO(C)(x)=C(x)$.
    \item \textbf{Indistinguishability: }For any $C_1,C_2\in {\cal C}_\lambda$ such that $|C_1|=|C_2|$ and $C_1(x)=C_2(x)$ for all inputs $x$, any quantum polynomial-time distinguisher ${\cal A}$ cannot distinguish the distributions $\iO(C_1)$ and $\iO(C_2)$ with noticeable probability, i.e., $\big|\Pr[{\cal A}(\iO(C_1))=1] - \Pr[{\cal A}(\iO(C_2))=1]\big| \leq \negl(\lambda).$
\end{itemize}
\end{restatable}

\begin{remark}
We note that there are some (candidate) constructions of post-quantum $\iO$, based on different assumptions. For example, \cite{bdgm20} constructed $\iO$ based on the circular security of LWE-based encryption schemes, which is conjectured to be quantum-secure. \cite{ww20} showed a construction of $\iO$ based on the indistinguishability of two distributions which is also arguably quantum-secure.
\end{remark}

Theorem~\ref{thm:mqcsp_qowf} implies the following result for quantum-secure $\iO$:

\begin{restatable}{theorem}{qioimplication}\label{thm:qioimplication}
Suppose that quantum-secure $\iO$ for polynomial-size circuits exists. Then, $\class{MQCSP}\in \class{BQP}$ implies $\class{NP}\subseteq \class{coRQP}$. 
\end{restatable}

\begin{proof}
Let $f_C(r) := \iO(C,r)$, where $r$ is the random string.
Then, by Theorem~\ref{thm:mqcsp_qowf}, we know that there exists a quantum polynomial-time algorithm ${\cal A}_{inv}$ with access to an $\class{MQCSP}$ oracle and a non-negligible function $p$ such that for any circuit $C$,
\begin{align}\label{eq:invert_io}
    \Pr_r\left[f_C({\cal A}_{inv}^{\class{MQCSP}}(C, \iO(C, r)))=f_C(r)\right]\geq p(|r|). 
\end{align}

Then, we can use ${\cal A}_{inv}$ to solve the $\textsf{Circuit-SAT}$ problem.  The algorithm is as follows:

\begin{algorithm}[H]
\caption{A quantum algorithm for \textsf{Circuit-SAT}}
	\label{alg:circuit_sat}
    \begin{algorithmic}[1]
	    \Input The description of a circuit $C:\{0,1\}^n\rightarrow \{0,1\}$.
        \State $s\gets |C|$.
        \State Compute $\bot_s$. \Comment{A canonical unsatisfiable circuit}
        \State $\hat{C}\gets \iO(C,r)$.
        \State $r'\gets {\cal A}_{inv}^{\class{MQCSP}}(\bot_s, \hat{C})$.
        \State \Return ``No'' if $\hat{C}=\iO(\bot_s, r')$.
    \end{algorithmic}
\end{algorithm}
We assume that for any $s\geq 0$, we can compute a canonical unsatisfiable circuit of size $s$ in $\poly(s)$ time.    

If $C\in \textsf{UNSAT}$, then $C\equiv \bot_s$. If $C=\bot_s$, by Eq.~\eqref{eq:invert_io}, ${\cal A}_{inv}^{\class{MQCSP}}$ finds $r$ with probability at least $p(|r|)$. Otherwise, by the indistinguishability of $\iO$ and $\class{MQCSP}\in \class{BQP}$, ${\cal A}_{inv}^{\class{MQCSP}}$ is a quantum polynomial-time algorithm and hence cannot distinguish $C\in\class{UNSAT}\setminus\{\bot_s\}$ and $\bot_s$ with more than $\negl(|r|)$ probability. Therefore, Algorithm~\ref{alg:circuit_sat} will reject $C$ with probability $O(p(|r|))$.

If $C\in \textsf{SAT}$, then $C\not\equiv \bot_s$. By the functionality of $\iO$, for any $r,r'$, $\iO(C,r)\ne \iO(\bot_s, r')$. Hence, Algorithm~\ref{alg:circuit_sat} will always accept $C$.

Hence, by repeatedly running Algorithm~\ref{alg:circuit_sat} many times, we get that $\class{NP} \subset \class{coRQP}$, the one-sided error analog of $\class{BQP}$
\end{proof}

\begin{remark}
It is worth noting that in the classical setting, the existence of $\iO$ implies that $\class{NP}$ and $\MCSP$ are equivalent under randomized reductions; the other direction directly follows from the fact that $\MCSP \in \class{NP}$. However, since it is unclear if $\MQCSP\in \class{NP}$, we can only conclude that $\class{NP} \subseteq \class{RQP}^{\MQCSP}$ assuming the existence of quantum-secure $\iO$. 
\end{remark}

%% file: UandS.tex
\section{MCSP for Quantum Objects}\label{sec:quantum obj}

In this section, we generalize the problem to considering circuit complexities of quantum objects, including unitaries and quantum states. In particular, we study their hardness, related reductions, and their implications to other subjects in quantum computer science. We start by defining the two problems.   


\begin{definition}[$\UMCSP_{\alpha,\beta}$]\label{def:umcsp}
Let $n,s,t\in \N$ and $t\leq s$. Let $\alpha,\beta\in (0,1]$. Let $U\in \mathbb{C}^{2^n\times 2^n}$ be a unitary. $\UMCSP$ is a promise problem defined as follows.
\begin{itemize}
\setlength\itemsep{-1mm}
    \item \textbf{Inputs:} the unitary matrix $U$, the size parameter $s$ in unary representation, and the ancilla parameter $t$.
    \item \textbf{Yes instance:} there exists a quantum circuit $\Circ$ using at most $s$ gates and operating on at most $n+t$ qubits such that for all $\ket{\psi}\in \C^{2^n}$, 
    \begin{align}
        \|(\bra{\psi}\otimes I_t)(U^{\dag}\otimes I_t) \Circ\ket{\psi,0^t}\|^2\geq \alpha, \label{eq:umcsp1}
    \end{align}  
    \item \textbf{No instance:} for every quantum circuit $\Circ$ using at most $s$ gates and operating on at most $n+t$ qubits, there exists some $\ket{\psi}\in \C^{2^n}$ such that 
    \begin{align}
        \|(\bra{\psi}\otimes I_t)(U^{\dag}\otimes I_t) \Circ\ket{\psi,0^t}\|^2\leq \beta. \label{eq:umcsp2}
    \end{align} 
\end{itemize}
With the promise that the input must be either a yes instance or a no instance, the problem is to decide whether the input is a yes instance or not. 
\end{definition}

\begin{remark}\label{rmk:precision}
Since the input to $\UMCSP$ is a unitary matrix $U$ and each entry is a complex number, we cannot fully describe $U$ and hence need to specify a precision parameter. Moreover, the precision issue is subtle in the search-to-decision reduction.
For a gate set $\mathcal{G}$, we denote $\ell_{\mathcal{G}}\in\N$ as the maximum number of bits used to encode an entry of a gate. Note that if a circuit uses $s$ gates from $\mathcal{G}$, then each entry in the resulting unitary can be written down with at most $s\cdot\ell_\mathcal{G}$ bits. Thus, by the triangle inequality for the distance between unitaries, it suffices to use $s\cdot\ell_{\mathcal{G}}$ bits to encode each entry of the input unitary.
Also, note that when $\alpha-\beta<2^{-s\cdot\ell_{G}}$, $\UMCSP_{\alpha,\beta}$ becomes a non-promise problem since effectively the gap between Yes and No instances does not matter.
In the definition of $\UMCSP$, we hide the introduction of precision parameter for simplicity. Note that from the above reasoning and the fact that the input unitary is $2^n\times 2^n$, it would not affect the complexity of the problem even one chooses the bit complexity to be $2^{O(n)}$, which is more than enough for most interesting situations. 
\end{remark}



\begin{definition}[$\SMCSP_{\alpha,\beta}$]~\label{def:smcsp}
Let $n,s,t\in \N$, where $t\leq s$. Let $\alpha,\beta\in (0,1]$. Let $\ket{\phi}\in \mathbb{C}^{2^n}$ be a quantum state. $\SMCSP$ is a promise problem defined as follows.
\begin{itemize}
\setlength\itemsep{-1mm}
     \item \textbf{Inputs:} size parameters $s$ and $n$ in unary, access to arbitrary many copies of $\ket{\psi}$, and the ancilla parameter $t$.
   \item \textbf{Yes instance:} there exists a quantum circuit $\Circ$ using at most $s$ gates and operating on at most $n+t$ qubits such that
    \begin{align*}
        \|(\bra{\phi}\otimes I_{n+t-1}) \Circ\ket{0^{n+t}}\|^2\geq \alpha, 
    \end{align*}
    \item \textbf{No instance:} for every quantum circuit $\Circ$ using at most $s$ gates and operating on at most $n+t$ qubits,
    \begin{align*}
        \|(\bra{\phi}\otimes I_{n+t-1}) \Circ\ket{0^{n+t}}\|^2\leq \beta.  
    \end{align*}
\end{itemize}
With the promise that the input must be either a yes instance or a no instance, the problem is to decide whether the input is a yes instance or not. 
\end{definition}

\begin{remark}
Similarly, the precision of the input parameters $\alpha,\beta$ of $\SMCSP$ has to depend on the bit complexity of the gate set. See Remark~\ref{rmk:precision} for more discussion.
\end{remark}

\begin{remark}
For the thresholds $\alpha,\beta$, it is worth noting that a quantum circuit that outputs a mixed state can always have nonzero inner product with an arbitrary state. Therefore, we cannot set $\beta$ to be arbitrarily small; otherwise, there will not be any $U$ or $\ket{\phi}$ satisfying the no instance.
\end{remark}

For $\SMCSP$, we focus on the version where the inputs are multiple quantum states. The input format is quite different from $\UMCSP$ and $\MQCSP$; instead of having the full classical description, $\SMCSP$ is given access to many copies of the quantum state. Hence, we say an algorithm for $\SMCSP$ is efficient if it runs in time $\poly(n,t,s)$, i.e., an efficient algorithm can use at most $\poly(n,t,s)$ copies of $\ket{\psi}$.  We choose this input format because that in the quantum setting, we generally cannot have the classical description of the quantum state. For instance, in shadow tomography\cite{aaronson18}, quantum gravity\cite{BFV19}, and quantum pseudorandom states\cite{JLF18}, the problem is given many copies of a quantum state, identify some properties of the state. Furthermore, although this problem seems to be much harder than having the full description or a succinct description (e.g, a circuit that generates the state) of the state, we will see that this problem has a $\class{QCMA}$ protocol. \footnote{Since $\SMCSP$ takes quantum inputs, the problem is not in $\class{QCMA}$ under the standard definition. However, problems with quantum inputs in quantum computing is natural, so, it is also reasonable to study the complexity classes that allow quantum inputs.}

\begin{remark}
On the other hand, the hardness results including the problem is in $\class{QCMA}$ (Theorem~\ref{thm:SQCMA_UMCSP}), the search-to-decision reduction (Theorem~\ref{thm:s2d_smcsp}), and the approximate self-reduction (Theorem~\ref{thm:smcsp_self_reduction}) all hold for the version where the input is a classical description for the state. 
\end{remark}  


Before proving the main theorems in this section, we introduce some notations and the swap test. Swap test \cite{BCWW01} is a quantum subroutine for testing whether two pure quantum states are close to each other.

\begin{notation}
We write $a\approx_\epsilon b$ for $a,b\in \mathbb{R}$ to mean $\|a-b\|\leq \epsilon$.
\end{notation}
\begin{notation}
We write $\ket{\varphi}\approx_\epsilon \ket{\phi}$ to mean $\|\ket{\varphi}-\ket{\phi}\|\leq \epsilon$.
\end{notation}
\begin{lemma}[Correctness of Swap Test]
For any two states $\ket{\phi},\ket{\psi}$, consider the following state
$$(\fH\otimes \fI)(\fcSWAP)(\fH\otimes \fI)\ket{0}\ket{\phi}\ket{\psi}$$
Measuring the first qubit gives outcome $1$ with probability $\frac{1}{2}-\frac{1}{2}|\bra{\phi}{\psi}\rangle|^2$.
\end{lemma}

\begin{claim}
\label{claim:simplebound}
Let $\ket{\phi}, \ket{\psi}\in \C^{2^{n}}$ be two quantum states such that $\ket{\phi} \approx_{\epsilon} \ket{\phi}$. Then, for any $\ket{\psi'}$ which is a state on at most $n$ qubits, 
\begin{align*}
    \|(\bra{\psi'}\otimes I)\ket{\phi}\| -\epsilon\leq \|(\bra{\psi'}\otimes I)\ket{\psi}\| \leq \|(\bra{\psi'}\otimes I)\ket{\phi}\| +\epsilon.  
\end{align*}
\end{claim}
\begin{proof}
Without loss of generality, we can write $\ket{\psi} = \ket{\phi} + \ket{\epsilon}$, where $\|\ket{\epsilon}\| \leq \epsilon$. Then, $\|(\bra{\psi'}\otimes I)\ket{\psi}\| = \|(\bra{\psi'}\otimes I)\ket{\phi}+ (\bra{\psi'}\otimes I)\ket{\epsilon}\|$. By using triangular inequality, we obtain the following two inequalities: 
\begin{align*}
    &\|(\bra{\psi'}\otimes I)\ket{\psi}\| \leq \|(\bra{\psi'}\otimes I)\ket{\phi}\|+ \|(\bra{\psi'}\otimes I)\ket{\epsilon}\|, \mbox{ and}\\
    &\|(\bra{\psi'}\otimes I)\ket{\psi}\| \geq \|(\bra{\psi'}\otimes I)\ket{\phi}\|- \|(\bra{\psi'}\otimes I)\ket{\epsilon}\|.
\end{align*}
Since $\|\ket{\epsilon}\| \leq \epsilon$, $\|(\bra{\psi'}\otimes I)\ket{\epsilon}\|\leq \epsilon$. This completes the proof. 
\end{proof}



\begin{theorem}\label{thm:umcspqcma}
$\UMCSP_{\alpha,\beta}$ where $\beta\leq 1-\poly(1/2^n)$ and $\alpha> 1- 2^{-2n-20}(1-\beta)^4$ (for example, $\alpha=1-\mathsf{exp}(-2^n),\beta=1-\poly(1/2^n)$) is in $\class{QCMA}$.
\end{theorem}
To design the verifier (that verifies a quantum circuit $\Circ$ really implements $U$ as we want), what we will do is the following checking:
\begin{enumerate}
    \item Standard basis check: check whether Eq.~\eqref{eq:umcsp1} is satisfied on standard basis states.
    \item Coherency check: Check Eq.~\eqref{eq:umcsp1} on superposition states in the form of $\ket{a}+\ket{b}$. This step has two goals: (1) checking whether the operation does behave similar to a unitary (instead of, for example, a collapsing measurement). (2) the unitary does not introduce different phases on different basis states.
\end{enumerate}
\begin{proof}
Our checking algorithm follows the two steps above. The certificate is the circuit that implements the unitary such that Eq.~\eqref{eq:umcsp1} is satisfied. The following algorithm verifies it (assuming the promise):
\begin{enumerate}
    \item (Standard basis check) For each $i\in [2^n]$, evaluate $(U^\dagger\otimes I_t)\Circ(\ket{i}\otimes \ket{0^t})$ for $\poly_1(2^n)$ times. Store the output state (which requires only polynomial memory); denote the $j$-th sample on input $i$ as $\ket{\varphi_i^j}$. Measure each of the states and check whether the output for $\ket{\varphi_i^j}$ is $i$. If not, mark it as a negative sample.\par
    If for any $i$, the ratio of negative samples is $\geq 2^{-2n-18}(1-\beta)^4$, reject. 
    \item (Coherency check) Do the following for each $i,j\in [2^n],i\neq j$ for $\poly_2(2^n)$ times:\par
    Apply $(U^\dagger\otimes I_t)\Circ$ on $\frac{1}{\sqrt{2}}(\ket{i}+\ket{j})\otimes \ket{0^t}$. Project the output system on $\frac{1}{\sqrt{2}}(\ket{i}+\ket{j})$. If the projection does not succeed, consider it as a negative sample.\par
        If for any of $i$, the ratio of negative samples is $\geq 2^{-2n-18}(1-\beta)^4$, reject.
\end{enumerate}
We will show, when $\poly_1,\poly_2$ are all chosen to be some sufficiently big polynomials, this test can be used as the $\class{QCMA}$-verifier we need.\par
First, if a circuit satisfies Eq.~\eqref{eq:umcsp1}, we can prove the verifier succeeds with probability $1-2^{-O(\poly(2^n))}$. 
\begin{enumerate}
\item First, in the standard basis check, by Eq.~\eqref{eq:umcsp1}, the expected ratio of negative sample is at most $1-\alpha\leq \frac{1}{4}\cdot\text{threshold}$ ($\text{threshold}:=2^{-2n-18}(1-\beta)^4$). By the Chernoff bound we have, $\forall a\in [2^n]$,
\begin{align*}
&\Pr[\text{negative ratio}\geq \text{threshold}]\\ = &~ \Pr[\text{negative samples}\geq \text{threshold}\cdot \poly_1(2^n)]\\
\leq &~ 2^{-O(\mathbb{E}[\text{negative samples}])}\tag{Chernoff bound}\\
\leq &~ 2^{-O(\poly_1(2^n)\cdot 2^{-2n-20}(1-\beta)^4)}\tag{threshold $\cdot \poly_1(2^n)\geq 4\cdot \mathbb{E}[\text{negative samples}]$}
\end{align*}
which is $2^{-O(\poly(2^n)})$ when $\poly_1$ is taken to be big enough. (Since $1-\beta=\poly(1/2^n)$)\par
Summing this failure probability for all $a\in [2^n]$ altogether we know with probability is at most 
\begin{align*}
    2^n\cdot 2^{-O(\poly(2^n))}=2^{-O(\poly(2^n))},
\end{align*}
which means it could not pass the first step. 

\item For the coherency check we can apply Eq.~\eqref{eq:umcsp1} directly again and know for each $a,b$, the expected error ratio is $\leq 1-\alpha\leq \frac{1}{4}\cdot\text{threshold}$. (Similarly $\text{threshold}:=2^{-2n-18}(1-\beta)^4$). Thus by the Chernoff bound and similar arguments
$$\forall a,b\in [2^n],~~~\Pr[\text{error ratio}\geq \text{threshold}]\leq 2^{-O(\poly_2(2^n)\cdot 2^{-2n-20}(1-\beta)^4)}$$
which is $2^{-O(\poly(2^n)})$ when $\poly_2$ is taken to be big enough. (Since $1-\beta=\poly(1/2^n)$)\par
Thus summing this failure probability for all $a,b\in [2^n]$ we know this step fails with probability at most 
\begin{align*}
    2^n\cdot 2^n\cdot 2^{-O(\poly(2^n))}=2^{-O(\poly(2^n))}.
\end{align*}
\end{enumerate}
Thus we get the completeness.\par
Then we prove a circuit that satisfies Eq.~(\ref{eq:umcsp2}) will be rejected with probability $1-2^{-O(\poly(2^n))}$. To prove that, we need to understand how the coherency check help us control the form of the states. We will prove the following lemmas step by step.\par
First, we show the success of coherency check implies the ancilla states have to be close to each other:
\begin{lemma}\label{lem:basistosuper}
Suppose for some $a,b\in [2^n], a\neq b$,  the following equations hold:
\begin{align}\label{eq:l22}
    &\|(\bra{a}\otimes I_t)(U^\dagger\otimes I_t)\Circ(\ket{a}\otimes \ket{0^t})\|^2\geq 1-\delta,\notag\\
    &\|(\bra{b}\otimes I_t)(U^\dagger\otimes I_t)\Circ(\ket{b}\otimes \ket{0^t})\|^2\geq 1-\delta,\notag\\
    &\left\|\left(\frac{\bra{a}+\bra{b}}{\sqrt{2}}\otimes I_t\right)(U^\dagger\otimes I)\Circ\left(\frac{\ket{a}+\ket{b}}{\sqrt{2}}\otimes \ket{0^t}\right)\right\|^2\geq 1-\delta.
\end{align}
Define the ancilla states $\ket{\chi_a}$, $\ket{\chi_b}$ via
\begin{equation}\label{eq:8.5.5.2} (U^\dagger\otimes I_t)\Circ(\ket{a}\otimes \ket{0^t})\approx_{\sqrt{\delta}} \ket{a}\otimes \ket{\chi_a}\end{equation}
\begin{equation}\label{eq:8.5.6.2} (U^\dagger\otimes I_t)\Circ(\ket{b}\otimes \ket{0^t})\approx_{\sqrt{\delta}} \ket{b}\otimes \ket{\chi_b}\end{equation}
where the right hand sides are the states from projecting $(U^\dagger\otimes I_t)\Circ(\ket{a}\otimes \ket{0^t})$, and projecting $(U^\dagger\otimes I_t)\Circ(\ket{b}\otimes \ket{0^t})$ on to $\ket{a},\ket{b}$ respectively.

Then we have
\begin{equation}\label{eq:eq11}
    \ket{\chi_a}\approx_{4\delta^{1/4}}\ket{\chi_b}
\end{equation}
\end{lemma}
\begin{proof}
We can evaluate the left hand side of Eq.~\eqref{eq:l22} and get
\begin{align*}
    &\left\|\frac{1}{\sqrt{2}}((\bra{a}+\bra{b})\otimes I_t)((U^\dagger\otimes I)\Circ)\frac{1}{\sqrt{2}}((\ket{a}+\ket{b})\otimes \ket{0^t})\right\|\\
   &\approx_{\sqrt{2\delta}}~\frac{1}{2}\|((\bra{a}+\bra{b})\otimes I_t)(\ket{a}\otimes \ket{\chi_a}+\ket{b}\otimes \ket{\chi_b})\|\tag{By Eqs.~\eqref{eq:8.5.5.2},\eqref{eq:8.5.6.2}}\\
    &=~\frac{1}{2}\|\ket{\chi_a}+\ket{\chi_b}\|\\
    &=~\sqrt{1-\frac{1}{4}\|\ket{\chi_a}-\ket{\chi_b}\|^2}
\end{align*}
Substitute Eq.~\eqref{eq:l22}, we know
\begin{align*}
    \sqrt{1-\frac{1}{4}\|\ket{\chi_a}-\ket{\chi_b}\|^2}\geq \sqrt{1-\delta}-\sqrt{2\delta},\\
    \|\ket{\chi_a}-\ket{\chi_b}\|\leq 2\sqrt{2\sqrt{2\delta(1-\delta)}-\delta}\leq 4\delta^{1/4}.
\end{align*}
The lemma is then proved.
\end{proof}
Furthermore, we can show, when Eq.~\eqref{eq:eq11} holds for all pairs $(a,b)$, the operation $(U^\dagger\otimes I)\Circ$ is indeed close to identity:
\begin{lemma}\label{lem:g7}
Suppose for all $a,b\in [2^n],a\neq b$, Eqs.~\eqref{eq:8.5.5.2},\eqref{eq:8.5.6.2},\eqref{eq:eq11} holds. Then for all $\ket{\psi}\in \C^{2^n}$, 
    \begin{align*}
        \|(\bra{\psi}\otimes I_t)(U^{\dag}\otimes I_t) \Circ\ket{\psi,0^t}\|^2\geq 1-10\cdot 2^{n/2}\delta^{1/4}
    \end{align*}  
\end{lemma}
\begin{proof}
Decompose $\ket{\psi}=\sum_{i\in [2^n]}{c_i}\ket{e_i}$. Take $\ket{\text{aux}}=\ket{\chi_0}$. Then
\begin{align*}
    (U^\dagger\otimes I_t)\Circ(\ket{\psi}\otimes \ket{0^t})
    =&~\sum_{i\in [2^n]}{c_i}(U^\dagger\otimes I_t)\Circ(\ket{e_i}\otimes \ket{0^t})\\
    &\approx_{\sum_ic_i\sqrt{\delta}}~\sum_{i\in [2^n]}c_i\ket{e_i}\otimes \ket{\chi_i}\tag{By Eqs.~\eqref{eq:8.5.5.2},\eqref{eq:8.5.6.2}}\\
   &\approx_{\sum_i4c_i\delta^{1/4}}~\sum_{i\in [2^n]}{c_i}\ket{e_i}\otimes\ket{\text{aux}}\tag{By Eq.~\eqref{eq:eq11}}\\
    =&~\ket{\psi}\otimes \ket{\text{aux}},
\end{align*}
which implies
\begin{align*}
    \|(\bra{\psi}\otimes I_t)(U^{\dag}\otimes I_t) \Circ\ket{\psi,0^t}\|^2\geq &~ (1-5\delta^{1/4}\sum_ic_i)^2\\
    \geq &~ (1-5\cdot 2^{n/2}\delta^{1/4})^2\\
    \geq &~ 1-10\cdot2^{n/2}\delta^{1/4}.
\end{align*}
And the proof is completed.
\end{proof}
 Then we prove a circuit that satisfies Eq.~(\ref{eq:umcsp2}) will be rejected with probability $1-2^{-O(\poly(2^n))}$. 
\begin{enumerate}
    \item After the standard basis check, $\Circ$ has to satisfy the following property, otherwise the verifier will reject with probability $1-2^{-O(\poly(2^n))}$:\par
   \begin{equation}\label{eq:31qcma} \forall a\in [2^n],~~\|(\bra{a}\otimes I_t)(U^\dagger\otimes I_t)\Circ(\ket{a}\otimes \ket{0^t})\|^2\geq 1-2^{-2n}(\frac{1}{11}(1-\beta))^4\end{equation}
   That's because otherwise the standard basis test for some $a\in [2^n]$ will have an expected negative ratio $\geq 2^{-2n}(\frac{1}{11}(1-\beta))^4\geq 4\cdot \text{threshold}$ (recall $\text{threshold}:=2^{-2n-18}(1-\beta)^4$).\par
   A more detailed calculation is as follows.
\begin{align*}
\Pr[\text{negative ratio}< \text{threshold}]= &~ \Pr[\text{negative samples}< \text{threshold}\cdot \poly_1(2^n)]\\
\leq &~ \exp(-O(\mathbb{E}[\text{negative samples}]))\\
\leq &~ \exp\left(-O(\poly_1(2^n)\cdot 2^{-2n}((1-\beta)/11)^4)\right),
\end{align*}
where the second step follows from the Chernoff bound, and the last step follows from 
\begin{align*}
    \text{threshold} \cdot \poly_1(2^n)< \frac{1}{4}\mathbb{E}[\text{negative samples}].
\end{align*}
Thus 
\begin{align*}
    \Pr[\text{negative ratio}\geq \text{threshold}]\geq 1-2^{-O(\poly_1(2^n)\cdot 2^{-2n}((1-\beta)/11)^4)}
\end{align*}
   \item After the coherency check, $\Circ$ has to satisfy the following property, otherwise the verifier will reject with probability $1-2^{-O(\poly(2^n))}$: for all $a,b\in [2^n]$, $a\ne b$, 
   \begin{equation}\label{eq:37qcma}
   \|\frac{1}{\sqrt{2}}((\bra{a}+\bra{b})\otimes I_t)(U^\dagger\otimes I)\Circ(\frac{1}{\sqrt{2}}(\ket{a}+\ket{b})\otimes \ket{0^t})\|^2\geq 1-2^{-2n}(\frac{1}{11}(1-\beta))^4\end{equation}
   The calculation is similar as the first step.
\item And from Eqs.~\eqref{eq:31qcma} and \eqref{eq:37qcma}, Lemma \ref{lem:g7} implies that for all $\ket{\psi}\in \mathbb{C}^{2^n}$,
    \begin{align*}
        \|(\bra{\psi}\otimes I_t)(U^{\dag}\otimes I_t) \Circ\ket{\psi,0^t}\|^2\geq 1-10\cdot 2^{n/2}(2^{-2n}(\frac{1}{11}(1-\beta))^4)^{1/4}>\beta.
    \end{align*} 
However, by the promise this is not possible to be in the no instance.
\end{enumerate}
This completes the proof. 
\end{proof}
\begin{claim}\label{claim:np_umcsp}
$\UMCSP_{\alpha,\beta}$ is in $\class{NP}$ when only linear ancilla qubits are allowed and $1-\alpha< 2^{-2n-20}(1-\beta)^4$ and $1-\beta \geq \poly(1/2^n)$ (for example, $1-\alpha=\mathsf{exp}(-2^n),1-\beta=\poly(1/2^n)$). However, $\UMCSP_{\alpha,\beta}$ is not trivially in $\class{NP}$ in general. 
\end{claim}
\begin{proof}
The certificate is the circuit implementation $\Circ$ that achieves Eq.~\eqref{eq:umcsp1}. Now since the circuit only operates on a polynomial-dimension system, the unitary transformation of the whole circuit can be computed and written down using only a polynomial-time classical computer.\par
The subtlety is to verify whether the unitary computed here satisfies Eq.~\eqref{eq:umcsp1}. We can prove it following the same way as the proof of Theorem \ref{thm:umcspqcma}. Here the quantum space is always polynomially bounded and a classical polynomial time verifier can simulate the protocol in the proof of Theorem \ref{thm:umcspqcma} classically. (One note is the quantum output samples there can be lazy-sampled.) This completes the proof. 
\end{proof}

Next, we showed that $\SMCSP$ has a $\class{QCMA}$ protocol. Note that since $\SMCSP$ is given access to quantum states, it is even not a promise problem under the standard definition. Therefore, we can only say there is a $\class{QCMA}$ protocol for this problem. 

\begin{theorem}\label{thm:SQCMA_UMCSP}
$\SMCSP_{\alpha,\beta}$  with gap $|\alpha-\beta|\geq \poly(s)$ has a $\class{QCMA}$ protocol.
\end{theorem}

\begin{proof}

We use the swap test to check whether the given states and the state generated from the certificate circuit are close. The verifier's algorithm is as follows: 
\begin{algorithm}[H]
\caption{The efficient verifier for $\class{SMCSP}$. }
	\label{alg:smcsp}
    \begin{algorithmic}[1]
	    \Input $s,t\in \N$, $\poly(s)$ copies of $\ket{\psi}$, and quantum circuit $\Circ$. 
	    \State Generate $\poly(s)$ $\ket{\phi} = \Circ\ket{0}$.
	    \State Apply swap test to $\ket{\psi}$ and $\ket{\phi}$. 
        \State \Return ``Yes'' if there are at least $\frac{a+b}{2}$ trials outputs $0$.
        \State \Return ``No'', otherwise. 
    \end{algorithmic}
\end{algorithm}

Given $s,t\in \N$ and $\poly(s)$ copies of $\ket{\psi}$, we first consider the case where there exists a circuit $\Circ$ such that $\|(\bra{\psi}\otimes I_{t})\Circ\ket{0^{n+t}}\|^2 \geq \alpha$. Let $\Circ$ be the certificate. Then, by applying the swap test to $\ket{\psi}$ and $\Circ\ket{0}$, the probability that we get $0$ (which means identical) is $\frac{1}{2}+\frac{|\bra{\psi}\Circ\ket{0}|^2}{2}$, which is at least $\frac{1+\alpha}{2}$ in this case. We denote the probability of outputs $0$ at the $i$-th trial as $X_i$. Then, By the Chernoff inequality, 
\begin{align*}
    \Pr\left[\sum_{i=1}^{\ell}X_i \geq (\frac{1}{2}+\frac{\alpha+\beta}{4})\ell\right]\leq \exp\left(-\frac{(\alpha-\beta)^2 \ell}{16}\right). 
\end{align*}
Since $|\alpha-\beta|\geq \frac{1}{\poly(s)}$, the success probability of Algorithm~\ref{alg:smcsp} in this case is at least $2/3$ by having $\ell = \poly(s)$ trials. Similarly, we can prove the case when there exists no circuit $\Circ$ such that $\|(\bra{\psi}\otimes I_{t})\Circ\ket{0^{n+t}}\|^2 > \beta$. This completes the proof. 

\end{proof}

Given Theorem~\ref{thm:SQCMA_UMCSP}, we can also obtain the following result when given classical descriptions of quantum states.
\begin{restatable}{corollary}{smcspqcma}\label{cor:SMCSP_QCMA}
$\SMCSP$ with classical descriptions of quantum states as inputs is in $\class{QCMA}$.
\end{restatable}

The subtlety is that the verifier needs to construct the state $\ket{\psi}$ given the classical description of $\ket{\psi}$. If the verifier can do this efficiently (in time $\poly(2^n)$), then the rest of the analysis follows the proof for Theorem~\ref{thm:SQCMA_UMCSP}. We leave the proof to Appendix~\ref{appx:UandS}.

For the ease of notation, we will simply denote $\UMCSP_{\alpha,\beta}$ and $\SMCSP_{\alpha,\beta}$ as $\UMCSP$ and $\SMCSP$ and will specify $\alpha$ and $\beta$ when it is necessary in the rest of the section.

\subsection{Reductions for \texorpdfstring{$\class{UMCSP}$}{UMCSP} and \texorpdfstring{$\class{SMCSP}$}{SMCSP}}

In this section, we will show search-to-decision reductions for $\class{UMCSP}$ and $\class{SMCSP}$. To prove the above results, it is easier for us to consider  $\class{UMCSP}$ and $\class{SMCSP}$ as problems for computing the circuit complexity of given unitaries and states.

We first give formal definitions of approximating functions, unitaries, and states and the corresponding quantum circuit complexities. 
\begin{definition}[Approximating $f$ with precision $\delta$]
\label{def:C_f}
We say that a quantum circuit $\Circ$ that approximates a function $f:\mathbb{Z}^n\rightarrow \mathbb{Z}^m$ with precision $\delta$ if for all $x\in \mathbb{Z}^n$, there exists $\epsilon'\leq \epsilon$ such that  
\begin{align}
    \Circ_{f,\delta}\ket{x}\ket{0^t} = \sqrt{1-\epsilon'} \ket{f(x)}\ket{\psi_{f(x)}} + \sqrt{\epsilon'} \ket{\phi_x}. \label{eq:C_f}  
\end{align}
\end{definition}
\begin{definition}[Approximating $U$ with precision $\delta$]
\label{def:c_u_appx}
Let $U$ be as a $2^{n}\times 2^{n}$ unitary. We define $\Circ_{U,\epsilon}$ as the circuit that approximates $U$ with precision $\delta$ such that for all $\ket{\psi}\in \C^{2^n}$ there exists $\delta'\leq \delta$
\begin{align*}
    \Circ_{U,\delta}\ket{\psi}\ket{0^{t}} = \sqrt{1-\delta'}(U\ket{\psi})\otimes\ket{\psi'} + \sqrt{\delta'}\ket{\phi'}.
\end{align*}
Here, the additional $t$ qubits for $\Circ_{U,\delta}$ are ancilla qubits. 
\end{definition}

\begin{definition}[Approximating $\ket{\psi}$ with precision $\delta$]
\label{def:c_s_appx}
Let $\ket{\psi}\in \mathbb{C}^{2^n}$ be a quantum state. We define $\Circ_{\ket{\psi},\epsilon}$ as the circuit that approximates $\ket{\psi}$ with precision $\delta$ 
\begin{align*}
    \Circ_{\ket{\psi},\delta}\ket{0^{n+t}} = \sqrt{1-\delta'}\ket{\psi}\ket{\psi'} + \sqrt{\delta'}\ket{\phi'}
\end{align*}
Here, $\delta'\leq \delta$ and the additional $t$ qubits are ancilla qubits. 
\end{definition}


We use $CC(\cdot,\epsilon)$ to denote the quantum circuit complexity of the minimum quantum circuit that approximates the given Boolean functions, states, or unitaries with precision $\epsilon$. 




\begin{remark}[Upper bounds on $CC(\cdot,\epsilon)$]
\label{remark:complexity_upperbound}
For any universal gate set, any unitary $U$ in $\mathbb{C}^{2^n\times 2^n}$ can be $\epsilon$-approximated by a circuit with size at most $\widetilde{O}(n^2 2^{2n}\log \frac{1}{\epsilon})$~\cite{NC00}. The same upper bound also holds for states. The existence of $2^{O(n)}$ upper bounds implies that $CC(\cdot, \epsilon)$ can be computed efficiently given efficient algorithms for $\SMCSP$ and $\UMCSP$. 
\end{remark}





\subsubsection{Search-to-decision reductions} In the following, we prove search-to-decision reductions for $\UMCSP$ and $\SMCSP$. The main intuition for these reductions is that quantum circuits are reversible, which gives us the ability to do some ``rewinding tricks''. We define the search versions of $\UMCSP$ and $\SMCSP$ as follows:

\begin{definition}[$\class{SearchUMCSP}_{\epsilon}$]\label{def:search_umcsp}
Let $n,t\in\mathbb{N}$. Let $U\in \mathbb{C}^{2^n\times 2^n}$ be a unitary matrix and $\epsilon\in (0,1)$. Let $s$ be the smallest integer such that there exists a quantum circuit $\Circ$ of size $s$ that uses at most $t$ ancilla bits and for all $\ket{\psi}\in \C^{2^n}$,
    \begin{align*}
        \|(\bra{\psi}\otimes I_t)(U^{\dag}\otimes I_t) \Circ\ket{\psi,0^t}\|^2\geq 1-\epsilon.
    \end{align*} 
Given $U$, $t$, and $\epsilon$, the problem is to output a circuit $\Circ'$ of size at most $s$ that uses at most $t$ ancilla bits and for all $\ket{\psi}\in \C^{2^n}$, $\|(\bra{\psi}\otimes I_t)(U^{\dag}\otimes I_t) \Circ'\ket{\psi,0^t}\|^2\geq 1-\epsilon-2^{-cn}$ for every constant $c>0$.
\end{definition}

\begin{definition}[$\class{SearchSMCSP}_{\epsilon,s}$]\label{def:search_smcsp}
Let $n,s,t\in\mathbb{N}$ and $\epsilon\in (0,1)$. Let $\ket{\psi}\in \mathbb{C}^{2^n}$ be a quantum state with the promise that there exists a circuit $\Circ$ of size at most $s$ and $t$ ancilla bits such that     \begin{align*}
    \|(\bra{\psi}\otimes I_{n+t-1}) \Circ\ket{0^{n+t}}\|^2\geq 1-\epsilon. 
\end{align*}
Given $(n,s,t)$ in unary, $\epsilon$, and access to arbitrary many copies of $\ket{\psi}$, the problem is to find a circuit $\Circ'$ of size at most $s$ and $t$ ancilla bits such that $\|(\bra{\psi}\otimes I_{n+t-1}) \Circ'\ket{0^{n+t}}\|^2\geq 1-\epsilon-2^{-cn}$ for every constant $c>0$. 
\end{definition}

\begin{remark}
\label{remark:smcsp_input_s}
In Definition~\ref{def:search_smcsp}, we have included the upper bound $1^s$ (the unary representation) as part of the inputs. This mainly follows from the fact that we are considering problems with copies of quantum states. One may expect that we can find $s$ by using binary search with an efficient algorithm for $\SMCSP$. However, efficient algorithms for $\SMCSP$ with $s= 2^n$ can run in time $\poly(2^n)$, and efficient algorithms for $\class{SearchSMCSP}$ without $1^s$ as part of the inputs need to run in time $\poly(n)$. Hence, this prevents us from finding $s$ efficiently (in time $\poly(n)$) with an efficient algorithm for $\SMCSP$ (in time $\poly(s)$). On the other hand, if we consider the case where $\SMCSP$ and $\class{SearchSMCSP}$ have the classical description of the state (instead of copies of the quantum state) as part of the inputs, then there is no need to have $1^s$ in the inputs of $\class{SearchSMCSP}$ since we can find $s$ via binary search with efficient algorithms for $\SMCSP$.  
\end{remark}

In the following, we show search-to-decision reductions for $\UMCSP$ and $\SMCSP$ when $t=0$ (i.e., no ancilla qubits)\footnote{In general, search-to-decision reductions for $\SMCSP$ and $\UMCSP$ mean that $\class{SearchSMCSP}$ reduces to $\SMCSP$ and $\class{SearchUMCSP}$ reduces to $\UMCSP$ for any $n,s,t\in \N$.}.

\begin{restatable}{theorem}{stdumcsp}\label{thm:search2decision_umcsp}
There exists a search-to-decision reduction for $\UMCSP$ for $t=0$ (i.e., no ancilla qubits). In particular, if there is a time $T(n)$ algorithm for $\UMCSP_{\alpha,\beta}$ where $\alpha>1-2^{-c_1n}$ and $\alpha-\beta\geq2^{-c_2n}$ for every constants $c_1,c_2>0$, then there is a time $\poly(T(n),2^n)$ algorithm for $\class{SearchUMCSP}_{\epsilon}$ where $\epsilon\geq2^{-c_3n}$ for every constant $c_3>0$ and $t=0$.
\end{restatable}

\begin{remark}
Here we require the gap $\alpha-\beta$ in the decision $\UMCSP$ oracle to be at least $\poly(2^{-n})$ because our $\class{QCMA}$ upper bound for $\UMCSP$ (see~\autoref{thm:umcspqcma}) only works in this regime.
\end{remark}

\begin{proof}
Let us first state the reduction in the form of an algorithm with oracle queries to $\UMCSP$ as follows.
\begin{algorithm}[H]
\caption{Search-to-decision reduction for $\class{UMCSP}$. }
	\label{alg:umcsp_s2d}
    \begin{algorithmic}[1]
	    \Input $\epsilon\in(0,1)$, $U\in \C^{2^n\times 2^n}$, and a constant $c_3>0$.
	    \State Let $U_0=U$, $\Delta=2^{-2c_3n}$, $\epsilon_0 = \epsilon$ , and $\epsilon_i=\epsilon_0+i\cdot\Delta$ for all $i\in\N$.
        \State Use the oracle $\UMCSP_{1-\epsilon_0,1-\epsilon_0-\Delta}$ to binary-search $s$, the minimum circuit size of $U$.
	    \State Set $i=1$.
	    \While {$i<s$}
	        \ForAll{gates $h_i$ in $\g$ on all $n$ qubits}
	            \If{$\class{UMCSP}_{1-\epsilon_i,1-\epsilon_i-\Delta}(U_{i-1}h_i^\dagger,s-i) = \text{Yes}$}
    	        \State Set $g_i = h_i$.
    	        \State Let $U_i=g^\dagger_iU_{i-1}$. 
    	        \State Set $i = i+1$.
    	        \State Break.
    	        \EndIf
	        \EndFor
	    \EndWhile
        \State \Return $g_1,\dots,g_s$. 
    \end{algorithmic}
\end{algorithm}

We inductively prove the following claim.
\begin{claim}\label{claim:umcsp_s2d}
For every $0<i<s$, at the $i$-th iteration in line 5, we know that there exists a circuit $\Circ$ of size at most $s-i+1$ such that  $\min_{\ket{\psi}}|\bra{\psi}U_{i-1}^{\dag} \Circ\ket{\psi}|^2\geq 1-\epsilon_i$.
\end{claim}
\begin{proof}
For the base case we consider $i=1$ and note that after line 2 in Algorithm~\ref{alg:umcsp_s2d}, we know that there exists a circuit $\Circ$ of size at most $s$
such that $\min_{\ket{\psi}}|\bra{\psi}U_0^{\dag}\Circ\ket{\psi}|^2\geq 1-\epsilon-\Delta=1-\epsilon_1$. This proves the base case.

Now, suppose the induction statement holds for some $i$, we first claim that the algorithm must go into the if-loop in line 6. Note that by induction hypothesis there exists a circuit $\Circ$ of size at most $s-i+1$ such that $\min_{\ket{\psi}}|\bra{\psi}U_{i-1}^{\dag}\Circ\ket{\psi}|^2\geq 1-\epsilon_i$. Let $g_i$ be the last gate in $\Circ$, we know that $\min_{\ket{\psi}}|\bra{\psi}(U_{i-1}^{\dag}g_i)(g_i^{\dagger} \Circ)\ket{\psi}\|^2\geq 1-\epsilon_i$ and $ g_i^\dagger \Circ$ is a circuit of size at most $s-i+1-1=s-i$. This shows that the algorithm will go into the if-loop in line 6 in the $i$-th iteration. Next, after the algorithm goes into line 6 in the $i$-th iteration, by the correctness of $\UMCSP_{1-\epsilon_i,1-\epsilon_i-\Delta}$, we know that there is a circuit $\Circ'$ ($=\Circ g_i^\dagger$) of size at most $s-i$ such that $\min_{\ket{\psi}}|\bra{\psi}U_{i}^{\dag} \Circ' \ket{\psi}|^2\geq 1-\epsilon_i-\Delta=1-\epsilon_{i+1}$. This completes the induction step and hence proves Claim~\ref{claim:umcsp_s2d}.
\end{proof}
Finally, with the same argument in the proof of Claim~\ref{claim:umcsp_s2d}, we know that
\begin{align*}
    \min_{\ket{\psi}}|\bra{\psi}U^{\dag} g_1\cdots g_s\ket{\psi}|^2\geq 1-\epsilon_s=1-\epsilon-s\cdot2^{-2c_3n}\geq1-\epsilon-2^{-c_3n}
\end{align*} as desired. Also, notice that the algorithm only queries the $\UMCSP$ oracle at most $2^n$ times and hence the running time is $\poly(T(n),2^n)$ where $T(n)$ is the running time of the $\UMCSP$ oracle.
\end{proof}

\begin{restatable}{theorem}{stdsmcsp}\label{thm:s2d_smcsp}
There exists a search-to-decision reduction for $\class{SMCSP}$ for $t=0$. In particular, if there is a time $T(n)$ algorithm for $\SMCSP_{\alpha,\beta}$ where $\alpha>1-2^{-c_1n}$ and $\alpha-\beta\geq2^{-c_2n}$ for every constants $c_1,c_2>0$, then there is a time $\poly(T(n),s)$ quantum algorithm for $\class{SearchSMCSP}_{\epsilon,s}$ where $\epsilon\geq2^{-c_3n}$ for every constant $c_3>0$ and $t=0$.

\end{restatable}

\begin{proof}
The proof is similar to the proof for Theorem~\ref{thm:search2decision_umcsp}. We describe the reduction as follows:
\begin{algorithm}[H]
\caption{Search-to-decision reduction for $\class{SMCSP}$. }
	\label{alg:smcsp_s2d}
    \begin{algorithmic}[1]
	    \Input $s\in \N$, $\epsilon\in(0,1)$, access to copies of $\ket{\psi}$, and a constant $c_3>0$.
	    \State Let $\ket{\psi_0}=\ket{\psi}$, $\Delta=2^{-2c_3n}$, $\epsilon_0=\epsilon$, and $\epsilon_i=\epsilon_0+i\cdot\Delta$ for all $i\in\N$.
	    \State Use the oracle $\SMCSP_{1-\epsilon_0,1-\epsilon-\Delta}$ to binary-search $s^*\leq s$, the minimum circuit size of $\ket{\psi}$.
	    \State Set $i=1$
	    \While {$i<s^*$}
	        \ForAll{gates $h_i$ in $\g$ on all $n+t$ qubits}
	            \If{$\class{SMCSP}_{1-\epsilon_i,1-\epsilon_i-\Delta}(\ket{\psi_i},s^*-i) = \text{Yes}$}
    	        \State Set $g_i = h_i$.
    	        \State Let $\ket{\psi_i}=g^\dagger_i\ket{\psi_{i-1}}$.
    	        \State Set $i = i+1$.
    	        \State Break.
    	        \EndIf
	        \EndFor
	    \EndWhile
        \State \Return $g_1,\dots,g_{s^*}$. 
    \end{algorithmic}
\end{algorithm}
The analysis is similar to the proof of Theorem~\ref{thm:search2decision_umcsp}. 
Notice that given access to the quantum state $\ket{\psi}$, we can uncompute the gates using a quantum computer.
Therefore, the search-to-decision reduction still holds.
\end{proof}

Regarding $\SMCSP$ and $\class{SearchSMCSP}$ which have the classical description of $\ket{\psi}$ as part of the inputs (instead of copies $\ket{\psi}$), we can also obtain the search-to-decision reduction following the same framework. The only difference is that the algorithm uncomputes the gates from the states by matrix-vector multiplication instead of applying the inverse of the gates on the states. The runtime of the matrix-vector multiplication is $\poly(2^n)$.  Note that, as we have mentioned in Remark~\ref{remark:smcsp_input_s}, $\class{SearchSMCSP}$ in this case does not need to have the upper bound $s$ in the inputs.  


\begin{corollary}
There exists a search-to-decision reduction for $\SMCSP$, where the search and the decision problems are given the classical descriptions of the states in inputs. 
\end{corollary}

It is worth noting that Algorithm~\ref{alg:smcsp_s2d} and Algorithm~\ref{alg:umcsp_s2d} do not directly work when considering quantum circuits that are allowed to use ancilla qubits (i.e., $t>0$). This follows from the fact that, based on definitions of $\UMCSP$ and $\SMCSP$, a quantum circuit $\Circ$ that implements the target unitary or state can apply an arbitrary operator on the ancilla qubits, i.e., $C^{\dag} (U\otimes I) \neq I$. In this case, we do not know the unitary of $\Circ$ or the state of $\Circ\ket{0}$, and thus we cannot run Algorithm~\ref{alg:smcsp_s2d} and Algorithm~\ref{alg:umcsp_s2d}. 

\subsubsection{Self-reduction for \texorpdfstring{$\SMCSP$}{SMCSP}} 

In this section, we show that $\SMCSP$ is approximately self-reducible. In other words, one can approximate the circuit complexity of an $n$-qubit state by computing the circuit complexity of an $(n-1)$-qubit state. 

\begin{theorem}
\label{thm:smcsp_self_reduction}

Let $\A_\delta$ be an efficient algorithm for computing $CC(\ket{\phi},\delta)$ for any ($n-1$)-qubit state $\ket{\phi}$. Let $\ket{\psi}$ be any $n$-qubit state. Given $(n,s)$ in unary, $\epsilon\in (0,1)$, and access to copies of $\ket{\psi}$, $CC(\ket{\psi},\epsilon)$ can be approximated efficiently using $\A_\delta$. 
\end{theorem}

Recall that $CC(\cdot,\epsilon)$ denotes the quantum circuit complexity of the minimum quantum circuit that approximates the given states with precision $\epsilon$. 

\begin{proof}
We first fix the gate set to be $CNOT$ and all single-qubit rotations and prove the theorem under this particular gate set. Then, we generalize the theorem to all gate sets by the Solovay-Kitaev Theorem in Theorem~\ref{thm:sk}.  

Let $\ket{\psi}\in \mathbb{C}^{2^n}$ be an arbitrary $n$-qubit quantum state. Without loss of generality, we can represent $\ket{\psi}$ as 
\[
    c_0\ket{0}\ket{\psi_0} + c_1\ket{1}\ket{\psi_1}, 
\]
where $c_0,c_1\in \mathbb{C}$ and $|c_0|^2 + |c_1|^2 = 1$. $\ket{1}$ and $\ket{0}$ are single-qubit states, and $\ket{\psi_0}$ and $\ket{\psi_1}$ are states on $n-1$ qubits and are not orthogonal in general. Our goal is show upper and lower bounds for $CC(\ket{\psi},\epsilon)$ from $CC(\ket{\psi_0},\delta)$ and $CC(\ket{\psi_1},\delta)$.

To prove the upper and the lower bounds, we first estimate $|c_0|^2$ and $|c_1|^2$ to precision $\epsilon/4$ by using quantum amplitude estimation. We denote the estimated values as $|c'_0|^2$ and $|c'_1|^2$ and consider the following two cases.
\begin{enumerate}
    \item $|c'_0|^2\mbox{ or }|c'_1|^2< \epsilon/2$; and
    \item $|c'_0|^2,|c'_1|^2\geq \epsilon/2$. 
\end{enumerate}

\paragraph{Upper bound}  In case that $|c'_0|^2$ (or $|c'_1|^2$) is less than $\frac{\epsilon}{2} $, $|c_1|^2$ (or $|c_0|^2$) must be greater than $1-\frac{3\epsilon}{4}$, which implies that the square of the inner product of $\ket{\psi}$ and $\ket{1}\ket{\psi_1}$ (or  $\ket{0}\ket{\psi_0}$ ) is at least $1-\frac{3\epsilon}{4}$. Therefore, 
\begin{align*}
CC(\ket{\psi},\epsilon) \leq CC(\ket{\psi_1},\epsilon/4) \mbox{ or } CC(\ket{\psi_0},\epsilon/4).     
\end{align*}

In case that both $|c'_0|^2$ and $|c'_1|^2$ are at least $\frac{\epsilon}{2}$, Let $\Circ_0 = \Circ_{\ket{\psi_0},\epsilon}$ and $\Circ_1 = \Circ_{\ket{\psi_1},\epsilon}$. Then, there exists $\Circ^*$ that approximates $\ket{\psi}$ with precision $\epsilon$ as follows: 

\begin{align*}
    \ket{0^{n}} &\xrightarrow{R\otimes I_{n-1}}\quad c_0\ket{0}\ket{0^{n-1}} +  c_1\ket{1}\ket{0^{n-1}}\\
    &\xrightarrow{control-\Circ_{1}}\quad c_0\ket{0}\ket{0^{n-1}} +  c_1\ket{1}\Circ_1\ket{0^{n-1}}\\
    &\xrightarrow{X\otimes I_{n-1}}\quad c_0\ket{1}\ket{0^{n-1}} +  c_1\ket{0}\Circ_1\ket{0^{n-1}}\\
    &\xrightarrow{control-\Circ_{0}}\quad c_0\ket{1}\Circ_0\ket{0^{n-1}} +  c_1\ket{0}\Circ_1\ket{0^{n-1}}\\
    & \xrightarrow{X\otimes I_{n-1}}\quad c_0\ket{0}\Circ_0\ket{0^{n-1}} +  c_1\ket{1}\Circ_1\ket{0^{n-1}} 
\end{align*}

Here $R$ is a single-qubit rotation gate that rotates $\ket{0}$ to $c_0\ket{0}+c_1\ket{1}$. Since our gate set includes all single-qubit rotations, the cost of $R$ is just 1. For $control-\Circ_0$ and $control-\Circ_1$, we can think of it as every gate in $\Circ_i$ is controlled by an additional qubit, i.e., $R$ becomes $control-R$ and $\mathsf{CNOT}$ becomes $\mathsf{Toffoli}$ gate. By the composition methods in~\cite{NC00}, we can implement these control gates with only constant multiplicative overhead. Hence, $|\Circ^*|\leq k\cdot (|\Circ_0|+|\Circ_1|)+3$ for some constant $k$, and we can conclude that
\begin{align*}
CC(\ket{\psi},\epsilon) \leq k\cdot (CC(\ket{\psi_0},\epsilon)+CC(\ket{\psi_1},\epsilon))+3.
\end{align*}

\paragraph{Lower bound} Let $\Circ$ be the minimum quantum circuit that approximates $\ket{\psi}$ with precision $\epsilon$.

When $|c'_0|^2$ and $|c'_1|^2$ are both at least $\epsilon/2$,  $|c_0|^2$ and $|c_1|^2$ are at least $\epsilon/4$ where $|c_0'|^2$ is the estimated value of $|c_0|^2$. Intuitively, we can obtain $\ket{\psi_0}$ or $\ket{\psi_1}$ by parallelly applying $\Circ$ on $O(\frac{1}{\epsilon})$-many $\ket{0^n}$ states and measuring the first qubits of all the outputs states in the computational basis. By deferring all these measurements toward the end of the computation, we obtain 
\[
    CC(\ket{\psi_i},\epsilon') \leq k^*(CC(\ket{\psi},\epsilon)+h)  
\]
for $i=0,1$, $h=O(1)$, and $k^* = O(1/\epsilon)$. Here $\epsilon\leq \epsilon' \leq (1-\frac{\epsilon}{4})^{k^*} +\epsilon$. The additional constant cost $h$ is from the overhead of deferring measurements. 

When $|c'_0|^2$ or $|c'_0|^2$ is at least $1-\epsilon/2$, the circuit for $\ket{\psi}$ is already a good approximation for $\ket{\psi_1}$ following the same reason for proving the upper bound in the same case. This implies that
\[
    CC(\ket{\psi_i},4\epsilon)\leq CC(\ket{\psi},\epsilon). 
\] 

\paragraph{The reduction} 
The algorithm is as follows: 
\begin{enumerate}
    \item Estimating $|c_0|$ and $|c_1|$ with precision $\epsilon/4$. 
    \item Approximate $CC(\ket{\psi},\epsilon)$ according to $|c'_0|$ and $|c'_1|$. 
    \begin{itemize}
        \item When $|c'_0|^2\mbox{ or } |c'_1|^2\leq \frac{\epsilon}{2}$, compute $CC(\ket{\psi_i},\epsilon/4)$ and $CC(\ket{\psi_i},4\epsilon)$ for $i=0,1$. Then,
        \begin{align*}
            CC(\ket{\psi_i},4\epsilon)\leq CC(\ket{\psi},\epsilon) \leq CC(\ket{\psi_i},\epsilon/4).
        \end{align*}
        \item When $|c'_0|^2,|c'_1|^2\geq \epsilon/2$, compute $CC(\ket{\psi_i},\epsilon')$ and $CC(\ket{\psi_i},\epsilon)$ for $i=0,1$. Then,
        \begin{align*}
            \frac{1}{k^*}\cdot \max_{i=0,1}~(CC(\ket{\psi_i},\epsilon')) - h \leq CC(\ket{\psi},\epsilon) \leq k\cdot (CC(\ket{\psi_0},\epsilon)+CC(\ket{\psi_1},\epsilon))+3
        \end{align*}
    \end{itemize}
\end{enumerate}

For the running time of the reduction, we can estimate $|c_0|^2$ and $|c_1|^2$ with precision $\epsilon/4$ in time $\poly(1/\epsilon)$ using quantum amplitude estimation. In case that $|c'_0|^2$ (or $|c'_1|^2$) is less than $\frac{\epsilon}{2} $, we only need to compute $CC(\ket{\psi_1},\epsilon/4)$ by having many enough copies of $\ket{\psi_1}$, which can be efficiently obtained by measuring $\ket{\psi}$. In case that both $|c'_0|^2$ and $|c'_1|^2$ are at least $\frac{\epsilon}{2}$, $|c_0|$ and $|c_1|$ must be at least $\frac{\epsilon}{4}$. Then, we can still obtain sufficiently many copies of $\ket{\psi_0}$ and $\ket{\psi_1}$ in time $\poly(\frac{1}{\epsilon})$ to compute $CC(\psi_0,\epsilon)$ and $CC(\psi_1,\epsilon)$.

Finally, we generalize the results above to arbitrary universal gate set by applying the Solovay-Kitaev Theorem. This gives upper bounds multiplicative overhead $\polylog \frac{CC(\ket{\psi_i},\delta)}{\epsilon} $ and lower bounds multiplicative overhead $\polylog^{-1} \frac{CC(\ket{\psi_i},\delta)}{\epsilon} $, where the choices of $i$ and $\delta$ depend on the cases.

\end{proof}

\begin{remark}
\label{remark:state_equivalence}
Theorem~\ref{thm:smcsp_self_reduction} also holds when the problem is given the classical description of the quantum state. When considering the version with classical descriptions of states, the reduction becomes even simpler since $c_0$ and $c_1$ can be easily computed from the input.  
\end{remark}




\subsubsection{Reducing \texorpdfstring{$\MQCSP$}{MQCSP} to \texorpdfstring{$\UMCSP$}{UMCSP}} In the following, we present a reduction from $\class{MQCSP}$ to $\class{UMCSP}$. We first introduce a unitary that trivially encode a given Boolean function.

\begin{definition}[Trivial unitary encoding of Boolean functions ($U_{f}$)]\label{def:U_f}
Let $f:\mathbb{Z}^n\rightarrow \mathbb{Z}^m$. We define $U_{f}$ as a $2^{n+m}\times 2^{n+m}$ unitary such that for all $x\in \mathbb{Z}^n$
\begin{align*}
    U_{f}\ket{x}\ket{0} = \ket{x}\ket{f(x)} 
\end{align*}
\end{definition}



Obviously, given the truth table of a function $f:\{0,1\}^n\rightarrow \{0,1\}^m$, one can compute $U_f$ in time $\poly(2^n)$. Then, one might expect that the circuit complexity of $f$ is equal to of $U_f$ (in Definition~\ref{def:U_f}). However, this is not the case in general since there are many unitaries that can compute $f$ without the form of $U_f$. In the following lemma, we show that one can give both upper and lower bounds for $CC(f)$ by the quantities $CC(U_f,\epsilon)$ and $CC(U_f,2\epsilon)$
\begin{lemma}\label{lemma:B2U}
\begin{align*}
    \frac{CC(U_f,2\epsilon)}{2} - m \leq CC(f,\epsilon) \leq CC(U_f,\epsilon)
\end{align*}
\end{lemma}
\begin{proof}
It is easy to see that given $\T(f)$, one can compute $U_f$ in time $2^{O(n+m)}$ which is polynomial in $|T(f)|=2^{n+m}$. 

We first consider the case where $CC(f)$ and $CC(U_f)$ can be computed with probability $1$. We can prove the first inequality as follows: 
\begin{align}
    \ket{x}\ket{0} \xrightarrow{\Circ_{f}} & e^{-i\theta_x}\ket{f(x)}\ket{\psi_x}\label{eq:global_phase}\\
    \xrightarrow{copy}& e^{-i\theta_x}\ket{f(x)}\ket{f(x)}\ket{\psi_x}\notag\\
    \xrightarrow{\Circ_{f}^{\dag}}& \ket{f(x)}\ket{x}\ket{0},\notag  
\end{align}
where $e^{-i\theta}$ are the global coefficient that $C_f$ might have for each $\theta_x$. $C_f^\dag(copy)C_f$ perfectly computes $U_f$ on all $x\in \{0,1\}^n$ without any global coefficient. This implies that for all $\ket{\psi}\in \C^{2^n}$, $C_f^\dag(copy)C_f$ computes $U_f\ket{\psi}$ perfectly. The cost for applying this circuit is $2CC(f)+m$. 
Therefore, we can conclude that $CC(U_f)\leq 2CC(f)+m$. The second inequality is true since a circuit for implementing $U_f$ is also a circuit for $f$ by definition.  Note that the global phase in Eq.~\eqref{eq:global_phase} can be absorbed into the second register; however, we write it down here to help explain why $C_f^\dag(copy)C_f$ implements $U_f$ not just only on the computational basis, but on all the states.    

In the following, we consider the case where we allow $U_f$ and $f$ to be computed with probability at least some thresholds. 
\begin{align}
    \ket{x}\ket{0} \xrightarrow{\Circ_{f,\epsilon}}  &~\sqrt{1-\epsilon} \ket{f(x)}\ket{\psi_{f(x)}} + \sqrt{\epsilon} (\sum_{y\neq f(x)} c_y\ket{y}\ket{\phi'_{x,y}})\notag\\
    \xrightarrow{Copy} &~\sqrt{1-\epsilon} \ket{f(x)}\ket{f(x)}\ket{\psi_{f(x)}} + \sqrt{\epsilon} (\sum_{y\neq f(x)} c_y\ket{y}\ket{y}\ket{\phi'_{x,y}})\notag\\
    = &~\ket{f(x)}(\sqrt{1-\epsilon} \ket{f(x)}\ket{\psi_{f(x)}} + \sqrt{\epsilon} (\sum_{y\neq f(x)} c_y\ket{y}\ket{\phi'_{x,y}}))\notag\\
    &+ 
    \sqrt{\epsilon}(\sum_{y\neq f(x)}c_y\ket{y}\ket{y}\ket{\phi'_{x,y}}-\sum_{y\neq f(x)}c_y\ket{f(x)}\ket{y}\ket{\phi'_{x,y}})\notag\\
    \xrightarrow{\Circ_{f,\epsilon}^\dag}&~ \ket{f(x)}\ket{x}\ket{0} + \ket{\psi'_x}.\label{eq:last} 
\end{align}
Since $\ipro{f(x),x,0}{\psi'_x} = -\epsilon$ and $\ipro{\psi'_x}{\psi'_x}=2\epsilon$, we have that
\begin{align*}
    \ket{\psi'_x} = -\epsilon\ket{f(x),x,0} + \sqrt{2\epsilon-\epsilon^2} \ket{\psi''_x}. 
\end{align*}
Therefore, we can rewrite Eq.~\eqref{eq:last} as 
\begin{align*}
    (1-\epsilon)\ket{f(x),x,0} + \sqrt{2\epsilon-\epsilon^2}\ket{\psi''_x}, 
\end{align*}
which implies that the circuit $C^\dag_{f,\epsilon}(Copy)C_{f,\epsilon}$ can compute $U_f$ with probability $(1-\epsilon)^2<1-2\epsilon$, i.e., $CC(U_f,2\epsilon) \leq 2CC(f,\epsilon)+m$. $CC(f,\epsilon) \leq CC(U_f,\epsilon)$ is also trivial by the definition.  

\end{proof}

We describe an algorithm to approximate $CC(f)$ given an oracle to $\class{UMCSP}$. 
\begin{algorithm}[H]
\caption{A reduction from $\class{MQCSP}$ to $\class{UMCSP}$}
	\label{alg:B2U}
    \begin{algorithmic}[1]
	    \Input Given $\T(f)$ for $f:\{0,1\}^n\rightarrow \{0,1\}^m$
        \State Construct $U_f$.
        \State Use $\class{UMCSP}$ oracle to compute $s=CC(U_f)$. 
        \State \Return $(\frac{s}{2}-m, s)$. 
    \end{algorithmic}
\end{algorithm}

\begin{theorem}\label{thm:b2u}
$\class{MQCSP}[s/2-1, s] \leq \class{UMCSP}$. 
\end{theorem}

\begin{proof}
By Lemma~\ref{lemma:B2U}, $CC(f,\epsilon)$ is between $\frac{CC(U_f,2\epsilon)}{2}-1$ and $CC(U_f,\epsilon)$ when $f$ is a Boolean function. To compute $CC(U_f,\epsilon)$, we can use the oracle for $\UMCSP_{1-\epsilon, \beta}$, where $\beta\leq 1-\epsilon-\frac{1}{\poly}$. For $CC(U_f,2\epsilon)$, we use the oracle for $\UMCSP_{1-2\epsilon, \beta'}$, where $\beta'\leq 1-2\epsilon-\frac{1}{\poly}$. This completes the proof. 
\end{proof}


\begin{remark}
One may expect that we can use Algorithm~\ref{alg:B2U} and $\class{NP}$-hardness result about $\class{multiMCSP}$ to prove $\class{NP}$-hardness of $\class{UMCSP}$. However, since the reduction for the multioutput MCSP generates functions with exponential-size output string, it make the first inequality in Lemma~\ref{lemma:B2U} fail. Therefore, whether $\class{UMCSP}$ is $\class{NP}$-hard or not is still open. 
\end{remark}


\subsection{Applications of SMCSP and UMCSP}\label{sec:smcsp_app}

In this part, we give applications of $\UMCSP$ and $\SMCSP$ to other fields in computer science and physics. For $\SMCSP$, we focus on the version with multiple quantum states as inputs.    

\subsubsection{Applications of UMCSP}

A question Aaronson raised in~\cite{aaronson2016complexity} is whether there exists an efficient quantum process that generates a family of unitaries that are indistinguishable from random unitaries given the full description of the unitary. Obviously, if we can solve $\UMCSP$ efficiently, we can distinguish truly random unitaries from unitaries generated from efficient quantum process. 
\begin{theorem}
\label{thm:pru}
If $\UMCSP$ has efficient (quantum) algorithms, then there is no efficient quantum process that generates a family of unitaries indistinguishable from random unitaries given the full description of the unitary.
\end{theorem}

Besides, some results about $\MQCSP$ in Section~\ref{sec:main-connections} also hold for $\UMCSP$ by Theorem~\ref{thm:b2u} and Algorithm~\ref{alg:B2U}. In the following, we list some results that trivially holds. 

\begin{corollary}
\label{cor:umcsp_owf}
If $\UMCSP\in \class{BQP}$, then there is no $\class{qOWF}$. 
\end{corollary}

\begin{corollary}
\label{cor:umcsp_io}
If there exists a quantum-secure $\iO$, then $\UMCSP\in \class{BQP}$ implies $\class{NP} \subseteq \class{coRQP}$. 
\end{corollary}
\begin{corollary}\label{cor:umcsp_amp}
Assume $\class{UMCSP}\in \class{BQP}$. Then, there exists a $\class{BQP}$ algorithm that, given the truth-table of an $n$-variable Boolean function of quantum circuit complexity $2^{\Omega(n)}$, output $2^{\Omega(n)}$ Boolean functions on $m=\Omega(n)$ variables each, such that all of the output functions have quantum circuit complexity greater than $\frac{2^m}{(c+2)m}$ for any $c>0$.
\end{corollary}
Corollary~\ref{cor:umcsp_owf}, Corollary~\ref{cor:umcsp_io} and Corollary~\ref{cor:umcsp_amp} hold since we use the $\MQCSP$ oracle as a distinguisher to distinguish functions whose sizes have a large gap, i.e., functions with quantum circuit complexity $\poly(n)$ from functions with quantum circuit complexity $2^{\Omega(n)}$. As the $\UMCSP$ oracle can solve $\class{MQCSP}[\frac{s}{2}-1,s]$, the existence of efficient algorithms for $\UMCSP$ also implies the same results.  

\begin{corollary}\label{cor:umcsp_bqe}
If $\class{UMCSP}\in \class{BQP}$, then $\class{BQE}\not\subset\class{BQC}[n^k]$ for all constant $k\in\mathbb{N}$.
\end{corollary}
Corollary~\ref{cor:umcsp_bqe} holds because for the gap version of $\class{MQCSP}$ with a constant gap, it gives a \emph{promise} $\class{BQP}$-natural property, which is defined in \cite{agg20}. Suppose we have an efficient quantum algorithm for solving $\class{MQCSP}[2^{\epsilon n}/2-1, 2^{\epsilon n}]$ for small constant $\epsilon$, then it will reject any function with quantum circuit complexity less than $2^{\epsilon n}/2$ and will accept another large subset of functions with quantum circuit complexity larger than $2^{\epsilon n}$. Then, we can use the technique in \cite{agg20} to construct the hard language ${\cal L}$ from the quantum $\mathsf{PRG}$ (Theorem~\ref{thm:bqp_PRG}) and promise quantum natural property. The remaining proof of Theorem~\ref{thm:ckt lb from MQCSP in BQP fixed} will work after this adaptation.

\subsubsection{Pseudorandom states} An efficient algorithm for $\SMCSP$ gives an efficient distinguisher for separating states with large circuit complexity from states with small circuit complexity given many copies of the state. Obviously, this gives us a way to distinguish random states from states that are generated from some efficient process.    

\begin{definition}[Pseudorandom states (PRS) (\cite{JLF18})]
Let $\kappa$ be the security parameter. Let $K$ be the key space and $\mathcal{H}$ be the state space both parameterized by $\kappa$. A family of quantum states $\{\ket{\psi_k}\}_{k\in K}\subset \mathcal{H}$ is pseudorandom if the following properties hold. 
\begin{enumerate}
    \item \textbf{Efficiency:} There is a quantum polynomial-time algorithm G that given $k\in K$, can generate $\ket{\psi_k}$. 
    \item \textbf{Indistinguishability:} For all quantum polynomial-time algorithm $\A$ and any $m=\poly(\kappa)$
    \begin{align*}
        |\Pr_k[\A(\ket{\psi_k})=1] - \Pr_{\ket{\psi}\leftarrow \mu}[\A(\ket{\psi})=1]|\leq \negl(\kappa), 
    \end{align*}
    where $\mu$ is the Haar measure on $\mathcal{H}$. 
\end{enumerate}
\end{definition}

\begin{theorem}\label{thm:smcsp_no qOWF}
If $\SMCSP\in \class{BQP}$, then there is no $\class{PRS}$ and $\class{qOWF}$. 
\end{theorem}
\begin{proof}

Let $\ket{\psi}$ be the state and $\A$ be the algorithm to distinguish whether $\ket{\psi}$ is a truely random state or from a particular efficient algorithm. In the definition of $\class{PRS}$, $\A$ knows the algorithm for constructing the $\class{PRS}$ (but it does not know the key.) Therefore, $\A$ also knows the circuit complexity $s$ for generating the $\class{PRS}$ $\ket{\psi}$. Suppose $\ket{\psi}$ is an $n$-qubit $\class{PRS}$ generated by a quantum circuit with size $s$, by solving $\SMCSP$ with size parameter $s$ and $\poly(s)$ copies of $\ket{\psi}$, the adversary can distinguish $\ket{\psi}$ from a Haar random state with high probability since a Haar random state has complexity exponential in $n$. 

Finally, by~\cite{JLF18}, there exist $\class{PRS}$ assuming the existence of $\class{qOWF}$. Since we can break any $\class{PRS}$ scheme by solving $\SMCSP$, we can also invert any $\class{qOWF}$ by solving $\SMCSP$. 

\end{proof}

\subsubsection{Estimating the wormhole volume} 

Integrating general relativity and quantum mechanics into a comprehensive theorem for quantum gravity is one of the most challenging physics problems. The AdS/CFT correspondence plays an important role in this line of research. The AdS/CFT correspondence conjectures the duality between the Anti-de Sitter space (i.e., the bulk) and a conformal field theory (i.e., the boundary). In particular, it conjectures the dictionary maps from wormholes and operators in the bulk to quantum states and operators on the boundary. One fascinating puzzle in Ads/CFT correspondence is about the volume of the wormhole. The volume of the wormhole grows steadily with time; what is the quantity of the corresponding quantum state on the boundary that has this feature? Susskind proposed the \textit{Complexity=Volume Conjecture}~\cite{susskind2014}. It states that the wormhole volume equals the quantum circuit complexity of the corresponding quantum state times some constant $c$. In the following, we give a brief description of the Complexity=Volume Conjecture and related backgrounds. One can see~\cite{susskind2014,BFV19} for detailed discussions.

\paragraph{AdS/CFT Correspondence} AdS/CFT correspondence conjectures a dual map $\Phi$ between wormholes (AdS side) and quantum systems (CFT side). The setting we consider here is wormholes with two-sided blackholes. Under this setting, the CFT side is divided into left and right systems denoted by Hamiltonians $H_L$ and $H_R$, where the left and right CFT systems are on $n$ qubits (compatible with the entropy of the wormhole $2^n$). We denote the whole system (with both left and right systems) as $H = H_L + H_R$. An early model of AdS/CFT goes under the ER=EPR slogan: the wormhole (Einstein-Rosen Bridge) is dual to maximally entangled (EPR) pair. The corresponding state is usually called the thermal field double (TFD) state $\ket{TFD}$~\cite{Maldacena_2013_ER_EPR} 
\begin{align}
\ket{TFD} = \frac{1}{\sqrt{2^n}}\sum_{i} e^{-E_i/\beta} \ket{i}_L\ket{i}_R,
\end{align}
where $\ket{i}_L$ and $\ket{i}_R$ are energy eigenstates of $H_L$ and $H_R$.

The quantum state after time-$t$ evolution is  
\begin{align*}
    \ket{TFD(t)} = e^{-iHt}\ket{TFD}.
\end{align*}
Recall the dual map $\Phi$ between a wormhole (AdS side) and a quantum system (CFT side), one can represent the wormhole after time $t$ as $\Phi(e^{-iHt}\ket{TFD})$ (and view $\Phi(\ket{TFD})$ as the wormhole at time $0$).

The statement of Complexity=Volume Conjecture can be stated as follows:
\begin{conjecture}[Complexity=Volume Conjecture~\cite{susskind2014}]
Consider a wormhole and its corresponding CFT system $H$, for some suitable $\epsilon$, $c$, and $0\leq t\leq O(2^n)$, 
\begin{align*}
CC_{\epsilon}(\ket{TFD}, \ket{TFD(t)}) = c\cdot Volume(\Phi(e^{-iHt}\ket{TFD})),
\end{align*}
where $CC_{\epsilon}(\ket{TFD}, \ket{TFD(t)})$ is the circuit complexity for constructing $\ket{TFD(t)}$ from $\ket{TFD}$ with at most $\epsilon$ error. 
\end{conjecture}

The $\SMCSP$ oracle gives a way to identify the quantum circuit complexity of the given state. This implies that if the dictionary map between the wormhole and the quantum state is efficient, one can estimate the wormhole volume in two ways. 1) Apply the dictionary map to transfer the wormhole to the corresponding state and then apply the $\SMCSP$ oracle for the circuit complexity, which gives the wormhole volume. 2) As it is hard to imagine mapping wormholes to states, one can view the $\SMCSP$ oracle as a POVM and then uses the dictionary map to transfer the POVM to the corresponding operators in the bulk to measure the volume. This gives the following lemma. 
\begin{theorem}\label{thm:smcsp_wormwhole}
Assuming the Volume=Complexity Conjecture, if the dictionary map can be computed in quantum polynomial time and $\SMCSP\in \class{BQP}$, then one can estimate the wormhole volume in quantum polynomial time when the volume is at most polynomially large.  
\end{theorem}
Here, we require the volume is at most polynomially large. This follows from the fact that we need a upper bound polynomial in $n$ for doing binary search to find the circuit complexity with an efficient $\SMCSP$ algorithm. If the upper bound is $2^{O(n)}$, the running time of the $\SMCSP$ algorithm can be $\poly(2^n)$. Therefore a quantum polynomial-time algorithm for $\SMCSP$ in this case would not imply a quantum polynomial-time algorithm for estimating the wormhole's volume. 

Besides, recall that the wormhole is initially described by $\ket{TFD}$. So, we also need to modify the definition of $\SMCSP$ to allow such an initial state.


Bouland et al. in~\cite{BFV19} used this correspondence in a reverse way. In particular, they showed that if the dictionary map and simulating the state in the bulk are efficient (i.e., the quantum Extended Church-Turing thesis holds for quantum gravity), then one can efficiently distinguish certain $\class{PRS}$ from Haar random state by mapping the state to the wormhole in the bulk and do the simulation in the bulk to estimate the volume. Following this idea, we can also conclude that if there is a quantum polynomial time algorithm for estimating the wormhole's volume, then one can compute the circuit complexity of the corresponding quantum state efficiently assuming the  the Volume=Complexity Conjecture and that the dictionary map is efficient\footnote{Note that this does not give an efficient algorithm for solving $\SMCSP$ in general since it can only solve $\SMCSP$ for CFT states.}.




\subsubsection{Succinct state tomography} In the following, we show that solving $\SMCSP$ can help to have a succinct answer to state tomography for states which are generated from a polynomial-size circuit without any measurement.   
\begin{definition}[Succinct state tomography]
\label{def:sst}
Let $\ket{\psi}$ be an $n$-qubit quantum state that is generated from a quantum circuit $\Circ$ of size $s$ without using measurement and ancilla qubits. Given $\poly(n)$ copies of $\ket{\psi}$ and an upper bound $s' $ where $s\leq s'\leq \poly(n)$, the problem is to output a succinct description (e.g., $\Circ$) of $\ket{\psi}$.  
\end{definition}
\begin{theorem}\label{thm:smcsp_succinct tomo}
Succinct state tomography in Def.~\ref{def:sst} reduces to $\SMCSP$.
\end{theorem}
\begin{proof}
Obviously, succinct state tomography reduces to the search version of $\SMCSP$. By the search-to-decision reduction in Theorem~\ref{thm:s2d_smcsp}, we can solve succinct state tomography by solving $\SMCSP$.
\end{proof}

%% file: main-ack.tex
\section{Acknowledgment}
We are grateful to Scott Aaronson and Boaz Barak for helpful discussions and valuable comments on our manuscript. We would like to thank Lijie Chen, Kai-Min Chung, Matthew Coudron, Yanyi Liu, and Fang Song for useful discussions.

\ifdraft
NHC's research is support by
the U.S. Department of Defense and NIST through the Hartree Postdoctoral Fellowship at QuICS and by NSF through IUCRC Planning Grant Indiana University: Center for Quantum Technologies (CQT) under award number 2052730. 

RZ's research is supported by NSF Grant CCF-1648712 and Scott Aaronson's Vannevar Bush Faculty Fellowship from the US Department of Defense.

CNC's research is supported by Boaz Barak's NSF awards CCF 1565264 and CNS 1618026.

JZ's research is supported by Adam Smith's NSF awards 1763786.
\fi

%% file: SZK.tex
\section{Proof for the hardness of \texorpdfstring{$\MQCSP$}{MQCSP}}

\qcmamqcsp*
\begin{proof}
The certificate is still the classical description of a quantum circuit $\Circ$ that has size at most $s$ and operates on at most $n+t$ qubits. The verifier first implements $\Circ$. Then, the verifier repeats evaluating $\Circ\ket{x,0^t}$ and measuring the first qubit $\ell = \poly(2^n)$ times. We denote the measurement outcomes of the $\ell$ trials as binary random variables $X_1,\dots,X_{\ell}$ which are all independent. Finally, the verifier checks if for all $x\in\{0,1\}^n$, there are at least $\frac{\alpha+\beta}{2}$ of the outcomes are consistent with $f(x)$.   

For the yes instance, we have the promise that $\|(\bra{f(x)}\otimes I_{n+t-1})\Circ\ket{x,0^t}\|^2\geq \alpha$ for all $x\in \{0,1\}^n$. Let $X = \sum_{i=1}^n X_i$. By using the second statement of Chernoff inequality, we have that $\Pr[X\leq \frac{(\alpha+\beta)\ell}{2}] \leq \exp\left(-\frac{(\alpha+\beta)^2\ell}{8\alpha}\right)$. 
By setting $\ell = \poly(2^n)$, we obtain $\Pr[X\leq \frac{(\alpha+\beta)\ell}{2}] \leq e^{-\poly(2^n)}$.
This implies that $\Pr[X\geq \frac{(\alpha+\beta)\ell}{2}\mbox{ for all } x\in\{0,1\}^n] \geq 1 - e^{-\poly(2^n)}$. For the no instance, we can do the similar analysis using Chernoff bound and show that there exists $x\in \{0,1\}^n$ such that $\Pr[X\geq \frac{(\alpha+\beta)\ell}{2}]$ is negligible. 
\end{proof}

\label{sec:proof_szk}
\mcspszkhard*
\begin{proof}[Proof of Theorem~\ref{thm:szk_mcsp}]

Let $(n,C_0,C_1)$ be a $\class{PIID}$ instance, where $C_0, C_1:\{0,1\}^m \rightarrow \{0,1\}^{m'}$ of size $n^k$. For $b=0,1$ and $x\in \{0,1\}^m$, we let $f_b(x) = C_b(x)$. Then, similar to the proof for Theorem~\ref{thm:mqcsp_qowf}, the idea is using $f_b$ to construct a pseudorandom generator $\hat{G}$ and break $\hat{G}$ by applying the $\class{MQCSP}$ oracle. Specifically, the algorithm is as follows: 

\begin{algorithm}[H]
\caption{A PPT algorithm $\A$ for $\class{PIID}$ with $\class{MQCSP}$ oracle}
	\label{alg:szk_mcsp}
    \begin{algorithmic}[1]
	    \Input $C_0,C_1$ of size $n^k$ and $m$-qubit input. 
	    \State Pick $x$ uniformly randomly from $\{0,1\}^m$.
	    \State Compute $f_0(x)$. 
        \State Use $f_0(x)$ to generate a pseudorandom string $G_{f_0(x)}(r)$ as in Lemma~\ref{qowf2qprg}.
        \State Use $G_{f_0(x)}(r)$ to generate the truth table $\T(g) = \hat{G}(r)$ as in Lemma~\ref{lem:qprg2lqprg}. 
        \State Apply the inverting algorithm $\A^{\MQCSP}_{inv}$ with access to function $f_1$ in Theorem~\ref{thm:mqcsp_qowf} to invert $f_1$ for $x'$. Note that the function used in the inverting algorithm is $f_1$ instead of $f_0$.
        \State \Return ``Yes'' if $C_0(x) = C_1(x')$; ``No'' if $C_0(x)\neq C_1(x')$.  
    \end{algorithmic}
\end{algorithm}

In Algorithm~\ref{alg:szk_mcsp}, we do not explicitly describe the inverting algorithms $\A_{inv}$. However, based on Theorem~\ref{thm:mqcsp_qowf}, such algorithms must exist. 

Then, when $(C_0,C_1)$ is a no instance, i.e., $\Pr_{x\in \{0,1\}^{m}}[\exists y\in \mathsf{Im}(C_0)\mbox{ such that }C_1(x) = y]\leq \frac{1}{2^n}$, the probability that there exists $x'$ such that $C_1(x') = C_0(x)$ over $x$ is at most $1/2^n$. In this case, Algorithm~\ref{alg:szk_mcsp} outputs ``Yes'' with probability at most $1/2^n$.  

When $(C_0,C_1)$ is a yes instance, $C_0$ and $C_1$ has statistical distance $1/2^n$ over $x\in \{0,1\}^m$. Then, the success probability of the algorithm $\A$ in Algorithm~\ref{alg:szk_mcsp} is
\begin{align*}
    \Pr[\A(C_0,C_1) = \text{``Yes''}] &= \Pr_x[ f_1(\A^{\MQCSP}_{inv}(f_1,f_0(x)))=f_0(x) ]\\
    &= \sum_{y\in \{0,1\}^{m'}}\Pr_x[f_0(x) = y]\Pr_x[f_1(\A^{\MQCSP}_{inv}(f_1,y))=y| y]
\end{align*}
Note that if we compute $f_1(x)$ (instead of $f_0(x)$) at step 2 in Algorithm~\ref{alg:szk_mcsp}, then the success probability of $\A$ is 
\begin{align*}
    \Pr_x[ f_1(\A^{\MQCSP}_{inv}(f_1,f_1(x)))=f_1(x) ]
    &= \sum_{y\in \{0,1\}^{m'}}\Pr_x[f_1(x) = y]\Pr_x[f_1(\A^{\MQCSP}_{inv}(f_1,y))=y|f_1(x) = y] \\
    &\geq 1/\poly(n). 
\end{align*}
The last inequality follows from Theorem~\ref{thm:mqcsp_qowf}. The $\MQCSP$ oracle can break $\hat{G}$ due to the fact that the construction of $\hat{G}$ is a small classical circuit and thus also a small quantum circuit. Therefore, we can use the $\MQCSP$ oracle to distinguish it from a truely random string.  

The difference between these two probabilities above is 
\begin{align*}
    &\Pr_x[ f_1(\A^{\MQCSP}_{inv}(f_1,f_0(x)))=f_0(x) ] -  \Pr_x[ f_1(\A^{\MQCSP}_{inv}(f_1,f_1(x)))=f_1(x) ] \\
    &= \sum_y\Pr_x[f_1(\A^{\MQCSP}_{inv}(f_1,y))=y|f_1(x) = y](\Pr_x[f_0(x) = y] - \Pr_x[f_1(x) = y])\\
    &\leq \sum_y (\Pr_x[f_0(x) = y] - \Pr_x[f_1(x) = y]) \leq \frac{1}{2^n}. 
\end{align*}
The last inequality follows from the definition of statistical distance. Therefore, Algorithm~\ref{alg:szk_mcsp} succeeds with probability at least $1/\poly(n) - 2^{-n}$ for a ``Yes'' instance. Finally, we can amplify the success probability for the yes instance to $2/3$ by repetition. Thus, $\class{PIID}\in \class{BPP}^{\class{MQCSP}}$. 

\end{proof}

%% file: learning.tex
\section{Learning Theory}\label{app:learning}

In this section, we provide the details of Section~\ref{sec:main learning} on the connection between learning theory and $\MQCSP$. 

\subsection{PAC learning}
Let us recall the definition of PAC learning.

\paclearning*

The following theorem shows that efficient PAC-learning for $\class{BQP/poly}$ is equivalent to efficient algorithms for $\MQCSP$. Here, $\class{BQP/poly}$ is defined as $\bigcup_{s\leq \poly(n)} \class{BQC}(s)$ 

\learningpac*

\begin{proof}
\mbox{}
\begin{itemize}
    \item The key ingredient to show $\MQCSP\in\class{BPP}$ implies efficient PAC learning for $\class{BQP/poly}$ is the ``learning from a natural property'' framework by~\cite{CIKK16}. First, note that $\class{BQP/poly}$ is a circuit class that contains $\class{P/poly}$ and hence can implement both the \textit{Nisan Wigderson generator} and the \textit{Direct Product + Goldreich-Levin amplification}. Second, $\MQCSP\in\class{BPP}$ implies there is a $\class{BPP}$-natural property against $\class{BQP/poly}$. Finally, by Theorem 5.1 of~\cite{CIKK16}, there is a randomized algorithm that $(1/\poly(n),\delta)$-PAC learns $f\in\class{BQP/poly}$ under the uniform distribution with membership queries for every $\delta>0$ in quasipolynomial time.
    
    \item Let $ALG$ be a randomized algorithm that $(1/\poly(n),\delta)$-PAC learns $f\in\class{BQP/poly}$ under the uniform distribution with membership queries for some $\delta>0$. We design the following randomized algorithm for $\MQCSP[\poly(n),\omega(\poly(n)),t(n),\tau]$ where $t(n)$ is the number of ancilla bits that will be determined later. For every $\tau>0$, let $\epsilon=\tau/2$.
    
    \begin{algorithm}[H]
	\caption{A quantum algorithm for $\MQCSP[\poly(n),\omega(\poly(n)),t(n),\tau]$}
	\label{alg:DP reconstruction}
    \begin{algorithmic}[1]
	    \Input The truth table $T$ of a $n$-variate Boolean function $f$.
	    \For{$i=1,\dots,10\ceil*{\log1/\delta}$}
	        \State Run $ALG$ and supply the membership query with the truth table $T$. Let $C_i$ be the output of $ALG$.
	        \State Uniformly and independently sample $x_1,\dots,x_\ell\in\{0,1\}^n$ where $\ell=\ceil*{100\log(1/\delta)/\epsilon^2}$.
	        \If{$|\{j\in[\ell]: C_i(x_j)\neq f(x_j)\}|<\frac{\epsilon}{10}\cdot \ell$}
	        \State Break and output ``Yes''.
	        \EndIf
	    \EndFor
	    \State Output ``No''.
    \end{algorithmic}
    \end{algorithm}
    Let us analyze the correctness of the above algorithm. First, if $f$ is an Yes instance, i.e., there exists a polynomial size quantum circuit $C$ that computes $f$, then due to the correctness of $ALG$, $\Pr_{C_i}[|\{x\in\{0,1\}^n: C_i(x)\neq f(x)\}|<2^n/\poly(n)]>\delta$ for each $i$. Namely, with probability at least $9/10$, there exists an $i\in[10\ceil*{\log1/\delta}]$ such that $|\{x\in\{0,1\}^n: C_i(x)\neq f(x)\}|<2^n/\poly(n)$. For this specific $i$, by Chernoff bound, with probability at least $9/10$ the algorithm will go to line 5 and output ``Yes''. That is, the above algorithm accepts an Yes instance with probability at least $2/3$ as desired.

    Next, if $f$ is a No instance, i.e., for every polynomial size quantum circuit $C$, we have $|\{x\in\{0,1\}^n: C(x)\neq f(x)\}|\geq\tau\cdot2^n>\epsilon\cdot 2^n$. For each $i\in[10\ceil*{\log1/\delta}]$, $C_i$ is a polynomial size circuit and hence by Chernoff bound, the algorithm goes to line 5 with probability at most $2^{-\Omega(\epsilon^2\ell)}$. Due to the choice of $\ell$, we know that the algorithm will output ``No'' with probability at least $2/3$. That is, the above algorithm rejects an No instance with probability at least $2/3$ as desired.

    Finally, the running time of the algorithm is $\poly(\text{Time}(ALG),1/\delta,1/\epsilon,n,m)$ where the dependency on $\poly(n,m)$ is for calculating $C_i(x_j)$ using the quantumness. Note that this running time is $\poly(2^n)$ and hence we conclude that $\MQCSP[\poly(n),\omega(\poly(n)),t(n)]\in\class{BQP}$.
    
    When the number of ancilla bits is $O(n)$, note that we can calculate $C_i(x_j)$ in $\poly(2^n)$ time and hence $\MQCSP[\poly(n),\omega(\poly(n)),t(n)]\in\class{BPP}$
\end{itemize}
\end{proof}

\subsection{Quantum learning}
As it could be the case that $\MQCSP$ might have non-trivial quantum algorithm, it is also of interest to study the connection to quantum learning. 

\quantumlearning*

It turns out that efficient quantum learning for a circuit class $\class{C}$ is equivalent to efficient quantum algorithm for its corresponding $\MCSP$, i.e., $\class{C}$-$\class{MCSP}$.


\learningquantum*

\begin{proof}
\mbox{}
\begin{itemize}
    \item 
The key idea is to quantize the ``learning from a natural property'' framework~\cite{CIKK16}. Let us start with three important lemmas from~\cite{agg20}.
    
    \begin{lemma}[Corollary of Lemma~4.3 and Lemma~4.4 in~\cite{agg20}]\label{lem:quantum NW}
    Let $L,s_D:\mathbb{N}\rightarrow\mathbb{N}$ be constructive functions and $\gamma\in(0,1)$ with $1\leq L(n)\leq2^n$ for every $n\in\mathbb{N}$. There exists an algorithm $A_{NW}$ on input $1^n$ and $1^L$ outputs $\textsf{code}(C_{NW})$ for a quantum circuit $C_{NW}$ in time $S(n)=\poly(n,L(n),s_D(n))$ with the following properties. In the following, we abbreviate $L=L(n)$ and $s_D=s_D(n)$.
    
    There exists a constant $c>0$ and an oracle function $NW^{\mathcal{O}}:\{0,1\}^m\rightarrow\{h:\{0,1\}^{\log L}\rightarrow\{0,1\}\}$ where $m=cn^2$ and $\text{size}(NW^{\mathcal{O}}(z))=\poly(n,\text{size}(\mathcal{O}))$ for all $z\in\{0,1\}^m$. Let $g:\{0,1\}^n\rightarrow\{0,1\}$. Suppose there is a quantum circuit $D$ of size at most $s_D$ with
    \[
        \left|\Pr_{z\in\{0,1\}^m,D}[D(NW^g(z))=1]-\Pr_{y\in\{0,1\}^L}[D(y)=1]\right|\geq\gamma\, .
    \]
    Then $C_{NW}$ on input $\textsf{code}(D)$ and with oracle access to $g$, outputs $\textsf{code}(C)$ for a quantum circuit $C$ of size $O(L^2\cdot s_D)$. With probability $\Omega(\gamma/L^2)$ over the output measurement of $C_{NW}$, we have
    \[
        \Pr_{x\in\{0,1\}^n,C}[C(x)=g(x)]\geq\frac{1}{2}+\frac{\gamma}{2L} \, .
    \]
    \end{lemma}

    \begin{lemma}[Lemma~4.5 in~\cite{agg20}]\label{lem:quantum DP}
    Let $k,s:\mathbb{N}\rightarrow\mathbb{N}$ be constructive functions and $\gamma>0$.
    There exists an algorithm $A_{GL}$ such that on input $1^n$ and $1^{k(n)}$ outputs a circuit $C_{GL}$ of size $\poly(n,k(n),s(n))$ in time $\poly(n,k(n),s(n))$ with the following properties. In the following, we abbreviate $k=k(n)$ and $s=s(n)$.
    
    Let $f:\{0,1\}^{kn}\rightarrow\{0,1\}^k$. Suppose there is a quantum circuit $C$ of size at most $s$ satisfying
    \[
        \E_{x\in\{0,1\}^{kn}}\E_{r\in\{0,1\}^k}[|(\bra{f(x)\cdot r}\otimes I)C\ket{x,r,0^m}|^2]\geq\frac{1}{2}+\gamma \, .
    \]
    Then $C_{GL}$ on input $\textsf{code}(C)$ outputs $\textsf{code}(G^{\mathcal{O}})$ for a quantum oracle circuit $G^{\mathcal{O}}$ of size $O(kn)$ such that
    \[
        \E_{x,G^C}[|(\bra{f(x)}\otimes I)G^C\ket{x,0^{k+m+1}}|^2]\geq\frac{\gamma^3}{2} \, .
    \]
    \end{lemma}
        
    \begin{lemma}[Theorem in~4.28~\cite{agg20}]\label{lem:quantum GL}
    Let $k,s:\mathbb{N}\rightarrow\mathbb{N}$ be constructive functions and $\epsilon,\delta\in(0,1)$.
    There exists a constant $c\geq1$ and an algorithm $A_{IJKW}$ such that on input $1^n$ and $1^{k(n)}$ outputs a circuit $C_{IJKW}$ of size $\poly(n,k(n),s(n),\log1/\delta,1/\epsilon)$ in time $\poly(n,k(n),s(n),\log1/\delta,1/\epsilon)$ with the following properties. In the following, we abbreviate $k=k(n)$ and $s=s(n)$.
    
    Let $g:\{0,1\}^{n}\rightarrow\{0,1\}$. Suppose $k$ is an even integer with
    \[
        k\geq c\cdot\frac{1}{\delta}\left[\log\frac{1}{\delta}+\log\frac{1}{\epsilon}\right] \, ,
    \]
    and suppose $G$ is a quantum circuit of size at most $s$ defined over $S_{n,k}:=\{S\subset\{0,1\}^n:|S|=k\}$ with $k$ output bits with
    \[
        \E_{B\sim S_{n,k},G}[G(B)=g^k(B)]\geq\epsilon \, .
    \]
    Then $C_{IJKW}$ on input $\textsf{code}(G)$ outputs $\textsf{code}(C)$ for a quantum circuit $C$ of size $\poly(n,k,s,\log(1/\delta),1/\epsilon)$ such that
    \[
        \E_{x\sim\{0,1\}^n,C}[C(x)=g(x)]\geq1-\delta \, .
    \]
    \end{lemma}
    
    Now, we are ready to describe our quantum learning algorithm for $\class{C}$.    
    \begin{algorithm}[H]
	\caption{A quantum learning algorithm for $\class{C}$}\label{alg:quantum learning}
    \begin{algorithmic}[1]
        \Input $1^n$, quantum oracle access to $n$-variate $f\in\class{C}$, and parameters $\delta\in(0,1)$.
        \State Let $L=\poly(n)$, $\epsilon=1/\poly(n)$, and $k=\ceil*{c\cdot\frac{1}{\delta}(\log\frac{1}{\delta}+\log\frac{1}{\epsilon})}$.
        \State $C_{NW}\gets A_{NW}(1^{kn+k})$; $C_{GL}\gets A_{GL}(1^{n},1^k)$; $C_{IJKW}\gets A_{IJKW}(1^{n},1^k)$.
        \State Let $\textsf{code}(D)$ be the description of a quantum circuit solving $\class{C}$-$\class{MCSP}$ with truth table size $L$.
        \State Use the oracle access to $f$ to build an oracle access to $NW^g$ where $g:\{0,1\}^{kn}\times\{0,1\}^k\rightarrow\{0,1\}$ with $g(x_1,\dots,x_k,r_1,\dots,r_k)=\oplus_{i=1}^k(r_i\cdot f(x_i))$ for every $x_1,\dots,x_k\in\{0,1\}^n$ and $r_1,\dots,r_k\in\{0,1\}$.
        \State $\textsf{code}(\tilde{C})\gets C_{NW}^{g}(\textsf{code}(D))$
        \State $\textsf{code}(G^{\mathcal{O}})\gets C_{GL}(\textsf{code}(\tilde{C}))$.
        \State $C\gets C_{IJKW}(\textsf{code}(G^{\tilde{C}}))$.
        \State Output $C$.
    \end{algorithmic}
    \end{algorithm}
    
    Let us analyze the correctness and running time of Algorithm~\ref{alg:quantum learning} simultaneously. Let $f:\{0,1\}^n\rightarrow\{0,1\}\in\class{C}$ be the function we want to learn. Let $g:\{0,1\}^{kn}\times\{0,1\}^k\rightarrow\{0,1\}$ be $g(x_1,\dots,x_k,r_1,\dots,r_k)=\oplus_{i=1}^k(r_i\cdot f(x_i))$ for every $x_1,\dots,x_k\in\{0,1\}^n$ and $r_1,\dots,r_k\in\{0,1\}$. Observe that if $\text{size}(f)=\poly(n)$, then $\text{size}(NW^g)=\poly(n)=\poly(\log L)$.
        
    Next, if $\class{C}$-$\class{MCSP}\in\class{BQP}$, then there exists a quantum algorithm $D$ running in time $\poly(L)$ with
    \[
        \left|\Pr_{z\in\{0,1\}^m,D}[D(NW^g(z))=1]-\Pr_{y\in\{0,1\}^L}[D(y)=1]\right|\geq\frac{1}{3} \, .
    \]
    By Lemma~\ref{lem:quantum NW}, $C^g_{NW}(\textsf{code}(D))$ outputs the description of a quantum circuit $C$ of size $O(L^2\cdot\text{size}(D))=\poly(n)$ in time $\poly(L,\textsf{size}(D))$ such that with probability $\Omega(1/L^2)$, 
    \[
        \Pr_{\substack{x_1,\dots,x_r\in\{0,1\}^n\\r_1,\dots,r_k\in\{0,1\},C}}[C(x_1,\dots,x_k,r_1,\dots,r_k)=g(x_1,\dots,x_k,r_1,\dots,r_k)]\geq\frac{1}{2}+\frac{1}{6L} \, .
    \]
    Next, by Lemma~\ref{lem:quantum DP}, $C_{GL}(\textsf{code}(C))$ outputs the description of an oracle quantum circuit $G^{\mathcal{O}}$ of size $O(kn\cdot\text{size}(C))=\poly(n)$ in time $\poly(n,k)$ such that 
    \[
        \E_{x_1,\dots,x_k,G^C}\left[|(\bra{f^k(x_1,\dots,x_k)}\otimes I)G^C\ket{x,0^{k+m+1}}|^2\right]\geq\Omega\left(\frac{1}{L^3}\right)=\frac{1}{\poly(n)} \, .
    \]
    Finally, by Lemma~\ref{lem:quantum GL}, $C_{IKJW}(\textsf{code}(G))$ outputs the description of a quantum circuit $C$ of size $\poly(n,k,\text{size}(G),\log(1/\delta),1/\epsilon)=\poly(n,1/\delta,1/\epsilon)=\poly(n)$ in time $\poly(n)$ such that
    \[
        \E_{x\sim\{0,1\}^n,C}[C(x)=g(x)]\geq1-\delta \, .
    \]
    We conclude that there is a polynomial time $(1/3,\delta)$-quantum learning algorithm for $\class{C}$.
    
    \item Let $ALG$ be a $(\epsilon,\delta)$-quantum learning algorithm for $\class{C}$ for some $\epsilon,\delta\in(0,1/2)$. We design the following quantum algorithm for $\class{C}$-$\class{MCSP}[\poly(n),\omega(\poly(n)),\tau]$. For every $\tau>0$, let $\epsilon=\tau/4$ and $\epsilon'=\tau/2$.
    
    \begin{algorithm}[H]
	\caption{A quantum algorithm for $\class{C}$-$\class{MCSP}[\poly(n),\omega(\poly(n)),\tau]$}
    \begin{algorithmic}[1]
	    \Input The truth table $T$ of a $n$-variate Boolean function $f$.
	    \For{$i=1,\dots,10\ceil*{\log1/\delta}$}
	        \State Run $ALG$ and supply quantum oracle access to $f$ using the truth table $T$. Let $C_i$ be the output of $ALG$.
	        \State Uniformly and independently sample $x_1,\dots,x_\ell\in\{0,1\}^n$ where $\ell=\ceil*{100\log(1/\delta)/\epsilon^2}$.
	        \If{$\sum_{j\in[\ell]}|(\bra{f(x_j)}\otimes I)U\ket{x,0^m}|^2\geq(1-\frac{\epsilon+\epsilon'}{2})\cdot \ell$}
	        \State Break and output ``Yes''.
	        \EndIf
	    \EndFor
	    \State Output ``No''.
    \end{algorithmic}
    \end{algorithm}
    Let us analyze the correctness of the above algorithm. First, if $f$ is an Yes instance, i.e., there exists a polynomial size quantum circuit $C$ that computes $f$, then due to the correctness of $ALG$, $\Pr_{C_i}[\mathbb{E}_{x\in\{0,1\}^n}[|(\bra{f(x)}\otimes I)U\ket{x,0^m}|^2]>1-\epsilon]>\delta$ for each $i$. Namely, with probability at least $9/10$, there exists an $i\in[10\ceil*{\log1/\delta}]$ such that $\mathbb{E}_{x\in\{0,1\}^n}[|\{x\in\{0,1\}^n: |(\bra{f(x)}\otimes I)U\ket{x,0^m}|^2\}|]\geq1-\epsilon$. For this specific $i$, by Chernoff bound, with probability at least $9/10$ the algorithm will go to line 5 and output ``Yes''. That is, the above algorithm accepts an Yes instance with probability at least $2/3$ as desired.

    Next, if $f$ is an No instance, i.e., for every polynomial size quantum circuit $C$, at least $\tau$ fraction of $x\in\{0,1\}^n$ has $|(\bra{f(x)}\otimes I)U\ket{x,0^m}|^2\leq1/2$. Hence, by the choice of $\epsilon'$, we have    $\mathbb{E}_{x\in\{0,1\}^n}[|(\bra{f(x)}\otimes I)U\ket{x,0^m}|^2]<(1-\epsilon')$. For each $i\in[10\ceil*{\log1/\delta}]$, $C_i$ is a polynomial size circuit and hence by Chernoff bound, the algorithm goes to line 5 with probability at most $2^{-\Omega(\epsilon^2m)}$. Due to the choice of $m$, we know that the algorithm will output ``No'' with probability at least $2/3$. That is, the above algorithm rejects an No instance with probability at least $2/3$ as desired.

    Finally, the running time of the algorithm is $\poly(\text{Time}(ALG),1/\delta,1/\epsilon,n,m)$ where the dependency on $\poly(n,m)$ is for calculating $C_i(x_j)$ using the quantumness. Note that this running time is polynomial in the size of the truth table and hence we conclude that $\class{C}$-$\class{MQCSP}[\poly(n),\omega(\poly(n)),\tau]\in\class{BQP}$.
\end{itemize}
\end{proof}

%% file: circuit_lowerbound.tex
\section{Proofs in Section~\ref{sec:main circuit lbs}}\label{sec:ckt lb}
In this section, we provide some missing proofs in Section~\ref{sec:main circuit lbs}. 

\subsection{Proof for Theorem~\ref{thm:ckt lb from MQCSP in BQP max}}\label{sec:mqcsp_bpe_qcma}

The goal of this section is to prove Theorem~\ref{thm:ckt lb from MQCSP in BQP max}.

\cktlbmax*

\begin{proof}
We follow the proof of a classical result in \cite[Theorem 10]{KC00}.

We first determine the maximum quantum circuit complexity for all Boolean functions using an $\class{MQCSP}$ oracle. For each $s=2^{O(n)}, 2^{O(n)}-1,\cdots$, decide if there exists a function $f_s$ such that $\mathrm{qCC}(f_s)\geq s$. The first $s$ we meet such that $f_s$ exists is the maximum quantum circuit complexity. It can be achieved by a $\class{QCMA}$ algorithm with input $1^s$, by the assumption $\class{MQCSP}\in \class{BQP}$. Hence, in classical $2^{O(n)}$ time with query access to a $\class{QCMA}$ oracle, we can find the maximum quantum circuit complexity $s_\star$ with high probability.

Then, we can construct the truth table by guessing bit-by-bit. We start from the empty truth table $T=\emptyset$. We first try to choose the first bit $T_1 = 0$ and decide if $T$ can be extended to a truth table with quantum circuit complexity $s_\star$, which can be done by a $\class{QCMA}$ oracle query. If the answer is ``No'', we set $T_1=1$. Then, we iterate over all bits of $T$. It is easy to see that in $O(2^n)$ time we can construct $T$ with high probability.

Therefore, we get a $\class{BPE^{QCMA}}$ algorithm for the maximum quantum circuit complexity problem, which immediately gives a $\class{BPE^{QCMA}}$ algorithm for computing such hard functions. By Claim~\ref{claim:gatecount}, this function has quantum circuit complexity at least $\Omega(2^{n}/n)$. Hence, by a padding argument for quantum circuits, we obtain a polynomial lower bound for $\class{BQP^{QCMA}}$. 
\end{proof}

\input{magnification}

%% file: magnification.tex
\subsection{Proof of Quantum Antichecker Lemma}\label{sec:magnify}
The goal of this section is to prove Lemma~\ref{lem:antichecker}.

\antichecker*


\begin{proof}
The proof follows \cite{chopr20}. 

Let  $\lambda\in (0,1)$ and $f$ be a Boolean function with $n$ input bits that is hard for $2^{n^\lambda}$-size quantum circuits.

For $i\geq 0$ and $s\in [0,1]$, define the predicate: 
\begin{align*}
    &P_f(y_1,\dots,y_i)[s] = 1 ~~\Longleftrightarrow\\
    &\leq s\text{ fraction of all quantum circuits of size $2^{n^\lambda}/2n$ compute}~f~\text{correctly on}~y_1,\dots,y_i.
\end{align*}
We also define the function:
\begin{align*}
    R_f(y_1,\dots,y_i):=\# \left\{\text{quantum circuits of size}
    ~2^{n^\lambda}/2n~\text{compute}~f~\text{correctly on}~y_1,\dots,y_i\right\}.
\end{align*}

Then, we construct $y_1,\dots,y_{2^{O(n^\lambda)}}$ iteratively. It is easy to see that $P_f(\cdot)[1]=1$. Suppose we already have $y_1,\dots,y_{i-1}$ such that $P_f(y_1,\dots,y_{i-1})[(1-1/4n)^{i-1}]=1$ holds. We want to find $y_i$ such that $P_f(y_1,\dots,y_{i})[(1-1/4n)^{i}]=1$. We will construct a formula $F$ of size $2^{O(n^\lambda)}$ such that if $R_f(y_1,\dots,y_{i-1})\geq 2n^2$, then
\begin{align*}
    & P_f(y_1,\dots,y_{i-1})\left[(1-1/4n)^{i-1}\right] = 1\\
    \Rightarrow~&\exists y_i~F(y_1,\dots,y_i,f(y_1),\dots,f(y_i))=1\\
    \Rightarrow~&P_f(y_1,\dots,y_i)\left[(1-1/4n)^i\right]=1.
\end{align*}

We first show how to find $y_i$ given this formula $F$. The idea is to use Valiant-Vazirani Isolation Lemma. Let $r$ be uniformly chosen from $\{2,n+1\}$ and let $h:\{0,1\}^n\rightarrow \{0,1\}^r$ be uniformly chosen from a pairwise independent hash family $\mathcal{H}_{n,r}$. Consider the following predicate
\begin{align*}
    &~F^{r,h}(y_1,\dots,y_{i-1},z,f(y_1),\dots,f(y_{i-1}),f(z)):=\\
    &~F(y_1,\dots,y_{i-1},z,f(y_1),\dots,f(y_{i-1}),f(z)) \wedge h(z)=0^r.
\end{align*}
The quantum circuit size of $F^{r,h}$ is $2^{O(n^\lambda)}$.

By the Isolation Lemma, for fixed $y_1,\dots,y_{i-1}$, with probability at least $1/8n$, there is a unique $z$ such that
\begin{align*}
    F^{r,h}(y_1,\dots,y_{i-1},z,f(y_1),\dots,f(y_{i-1}),f(z))=1.
\end{align*}

If we sample $2^{O(n^\lambda)}$ many tuples of $(r,h)$, then the probability that none of those $(r,h)$ will lead to unique solution of $F^{r,h}$ is less than  $2^{-2^{O(n^\lambda)}/8n}\leq 2^{-2^{O(n^\lambda)}}$ by choosing proper constant. On the other hand, the total number of all possible $y_1,\dots,y_{i-1},f(y_1),\dots,f(y_{i-1})$ is at most $2^{2^{O(n^\lambda)}}$. It means that there exists a set ${\cal R}$ of $2^{O(n^\lambda)}$ tuples of $(r,h)$ such that for any $y_1,\dots,y_{i-1},f(y_1),\dots,f(y_{i-1})$, there exists an $(r,h)\in {\cal R}$ that makes $F^{r,h}$ have unique solution. Note that ${\cal R}$ can be hard-wired into the circuit $C_n$. Hence, the $j$-th bit of the antichecker $y_i$ can be computed by the following formula of size $2^{n+O(n^\lambda)}$:
\begin{align}\label{eq:formula_y_i}
    \bigvee_{z\in \{0,1\}^n} z_j \wedge F^{r,h}(y_1,\dots,y_{i-1},z,f(y_1),\dots,f(y_{i-1}),f(z)).
\end{align}
Then, we need to select an $(r,h)$ from ${\cal R}$ that gives the unique $y_i$. This task is in $\class{QCMA}$, and by assumption, $\class{QCMA}\subseteq \class{BQC}[\poly]$. So, we just need to apply a $2^{O(n^\lambda)}$-size quantum circuit. 
Once we have $y_i$, $f(y_i)$ can be obtained from $\textsf{tt}(f)$ via an Address function, which can be implemented by a circuit of size $2^{n+O(\log n)}$.

By repeating this process, we can get $y_1,\dots,y_{2^{O(n^\lambda)}}$ and $f(y_1),\dots,f(y_{2^{O(n^\lambda)}})$ by a $2^{n+O(n^\lambda)}$ circuit. Then, we need to check $R_f(y_1,\dots,y_{2^{O(n^\lambda)}})\geq 2n^2$. Deciding whether $R_f(y_1,\dots,y_i)\geq 2n^2$ is in $\class{QCMA}\subseteq \class{BQC}[\poly]$ with input $(y_1,\dots,y_i,f(y_1),\dots,f(y_i), 1^{2^{O(n^\lambda)}})$ since the witness is $2n^2$ quantum circuits each of size $2^{n^{\lambda}}/2n$, which can be represented by a $2^{O(n^\lambda)}$ binary string. The witness can be checked by simulating the quantum circuits. Therefore, there exists a $2^{O(n^\lambda)}$ quantum circuit for it. When $R_f(y_1,\dots,y_i)\leq 2n^2$, the $2n^2$ circuits of size $2^{n^\lambda}/2n$ can be generated by an $\class{QCMA}^{\class{coQCMA}}$ algorithm. And since $\class{QCMA}\subseteq \class{BQC}[\poly]$, by uncomputing the garbage, we can show that $\class{QCMA}^{\class{coQCMA}}\subseteq \class{BQC}[\poly]$ and this step can be done by a $2^{O(n^\lambda)}$ quantum circuit. For each circuit, by exhaustively searching, we can find an $n$-bit string that witness the error. The circuit size of this step is $2^{n+O(n^\lambda)}$.

In order to construct $F$, we use a result in \cite{ops19} (Lemma 23) showing that if $P_f(y_1,\dots,y_{i-1})[(1-1/4n)^{i-1}]=1$ and $R_f(y_1,\dots,y_{i-1})\geq 2n^2$, then
\begin{align}\label{eq:ops_lem}
    \exists y_i~P_f(y_1,\dots,y_{i})\left[(1-1/4n)^{i-1}(1-1/2n)\right]=1.
\end{align}
The proof is by a standard counting argument, and by examining the proof, we find that it also holds for quantum circuits. 

By Eq.~\eqref{eq:ops_lem}, we know that there exists a $y_i$ such that $\leq(1-1/4n)^{i-1}(1-1/2n)<(1-1/4n)^i$ fraction of circuits of size $2^{n^\lambda}/2n$ that can compute $f$ on $y_1,\dots,y_i$. The remaining task is to find a witness (which is $F$) that can certify $P_f(y_1,\dots,y_{i})\left[(1-1/4n)^{i}\right]=1$. We can use an approximate counting with linear hash functions to construct $F$. More specifically, by \cite{jer09}, the witness is a set of matrices $A_1,\dots,A_{2^{O(n^\lambda)}}$ defining an injective map from the Cartesian power of the set of all circuits of size $2^{n^\lambda}/2n$ that compute $f$ on $y_1,\dots,y_i$ to the same Cartesian power of $(1-1/4n)^i$ fraction of the set of all circuits of size $2^{n^\lambda}/2n$. The existence of these matrices can be decided by an $\class{QCMA}^{\class{coQCMA}}$ algorithm,
which can also be decided by a $2^{O(n^\lambda)}$ quantum circuit, by our assumption. 
\end{proof}

\subsection{Quantum Impagliazzo-Wigderson generator}\label{sec:mqcsp_amp}
The goal of this section is to prove Lemma~\ref{lem:quantum_IW_prg}.

\quantumiwprg*


Before giving the proof, we first recall some necessary definitions and lemmas in the previous work.

\begin{lemma}[A variant of Lemma 4.29 in \cite{agg20}]\label{lem:random_reducible}
Let $L:\{0,1\}^*\rightarrow \{0,1\}$ be a language that is randomly reducible to the language $L'$. For every $n$, suppose we have the description of a quantum circuit $U$ such that 
  \begin{align*}
    \E_{x \in \{0,1\}^n} \left[ \|\Pi_{L'(x)} U\ket{x,0^ q} \|^2 \right]\geq
    1-\frac{1}{n^{k}},
  \end{align*}
  for some $k \geq 2b + a$.
  
There is a $O(|U| \cdot \poly(n))$-size quantum circuit $\widetilde{U}$ that satisfies
\begin{align*}
    \|\widetilde{\Pi}_{x} \widetilde{U}\ket{0,x,0^{q^*}}\|^2 \geq 1- 2^{-2n+1} \qquad \text{ for every } x\in \{0,1\}^n,
\end{align*}
where $\widetilde{\Pi}_{x} = \ket{L(x)}\bra{L(x)} \otimes \ket{x}\bra{x} \otimes \ket{0^{q^*}}\bra{0^{q^*}}$ and $q^* = \poly(n)$.
\end{lemma}

\begin{definition}[Expander walks]\label{def:expander_walk}
Let ${\cal G}$ be a graph with vertex set $\{0,1\}^n$ and degree $16$.  Let the expander walk generator $\mathsf{EW}:\{0,1\}^n\times [16]^k\rightarrow \{0,1\}^{nk}$ such that $\mathsf{EW}(v,d):=(v_1,\dots,v_k)$, where $v_1=v$ and $v_{i+1}$ is the $d_i$-th neighbor of $v_i$ in ${\cal G}$. 
\end{definition}

\begin{definition}[Nearly disjoint subsets]\label{def:nds}
Let $\Sigma=\{S_1,\dots,S_k\}$ be a family of subsets of $[m]$ of size $n$. We say $\Sigma$ is $\gamma$-disjoint if $|S_i\cap S_j|\leq \gamma n$ for any $i\ne j$. 

For $r\in \{0,1\}^m$, $S\subseteq [m]$, let $r|_S$ be the restriction of $r$ to $S$. Then, for a $\gamma$-disjoint $\Sigma$, $\mathsf{ND}^\Sigma:\{0,1\}^m\rightarrow \{0,1\}^{nk}$ is defined by $\mathsf{ND}^\Sigma (r):=r|_{S_1},\dots,r|_{S_k}$.
\end{definition}

\begin{definition}[$M$-restrictable]\label{def:restrictable}
We say $G_n:\{0,1\}^m\rightarrow \{0,1\}^{nk}$ is $M$-restrictible if there exists a polynomial-time computable function $h:[n]\times \{0,1\}^n\times \{0,1\}^m\rightarrow \{0,1\}^m$ such that 
\begin{itemize}
    \item For any $i\in [n],x\sim \{0,1\}^n, \alpha\sim \{0,1\}^m$, $h(i,x,\alpha)$ is uniformly distributed.
    \item For any $i,x,\alpha$, let $G(h(i,x,\alpha)):=x_1,\dots,x_k$. Then, we have $x_i=x$.
    \item For any $i, j\ne i$, for any $\alpha$, there exists a set $S\subseteq \{0,1\}^n$, $|S|\leq M$ such that for any $x$, $x_j\in S$.
\end{itemize}
\end{definition}

\begin{definition}[$(k',q,\delta)$-hitting]\label{def:hitting}
We say $G_n:\{0,1\}^m\rightarrow \{0,1\}^{nk}$ is $(k',q,\delta)$-hitting if for any sets $H_1,\dots,H_k\subseteq\{0,1\}^n$, $|H_i|\geq \delta 2^n$, we have 
\begin{align*}
    \Pr[|\{i:x_i\in H_i\}|<k']<q.
\end{align*}
\end{definition}

\begin{proof}[Proof of Lemma~\ref{lem:quantum_IW_prg}]
We follow the proof in \cite{iw97}. We first assume that there exists a function $f_0:\{0,1\}^n\rightarrow \{0,1\}$ such that the quantum circuit complexity of $f_0$ is $2^{\Omega(n)}$. We may assume that $f_0\in \class{BQE}$. Then, encoding the truth table of $f_0$ by a locally list-decodable code, we obtain a function $f_1:\{0,1\}^{O(n)}\rightarrow \{0,1\}$ such that $f_1\in \class{BQE}$, and for any quantum circuit ${\cal B}_1$ of size less than $2^{\Omega(n)}$,
    \begin{align*}
        \E_{x\sim \{0,1\}^{O(n)},{\cal B}_1}[B_1(x)=f_1(x)]:=\E_{x\sim \{0,1\}^{O(n)},{\cal B}_1}[\|\Pi_{f_1(x)}{\cal B}_1\ket{x,0}\|]\leq 1-n^{-O(1)}.
    \end{align*}
The properties of $f_1$ can be proved by Lemma~\ref{lem:random_reducible}.

Then, by Lemma~\ref{lem:quantum GL} with $k=\poly(n), \epsilon=O(1), \delta=\frac{1}{\poly(n)}$, we have a function $f_2=f_1^{\otimes k} : \{0,1\}^{kn}\rightarrow \{0,1\}^k$ such that for any quantum circuit ${\cal B}_2$ of size less than $2^{\Omega(n)}$, 
\begin{align*}
    \E_{x\in \{0,1\}^{nk},{\cal B}_2}[{\cal B}_2 (x) = f_2(x)]\leq O(1).
\end{align*}

We can apply the quantum  Goldreich-Levin Theorem (Lemma~\ref{lem:quantum DP}) to $f_2$ and get a function $f_3:\{0,1\}^{n}\rightarrow \{0,1\}$ (scaling the input size) such that for any quantum circuit ${\cal B}_3$ of size less than $2^{\Omega(n)}$, 
\begin{align*}
    \E_{x\in \{0,1\}^{n},{\cal B}_3}[{\cal B}_3 (x) = f_2(x)]\leq \frac{2}{3}.
\end{align*}

The remaining thing is to ``quantize'' the direct-product generator defined by \cite{iw97} using $f_3$. More specifically, we say $G$ is a $(s,s',\epsilon,\delta)$ \emph{quantum direct-product generator} if $G:\{0,1\}^{m}\rightarrow \{0,1\}^{nk}$ such that for every Boolean function $g$ that is $\delta$-hard for any quantum circuit of size $s$, we have $g^{\otimes }\circ G$ is $\epsilon$-hard for any quantum circuit of size $s'$. The main result of \cite{iw97} is the construction of $(2^{\Omega(n)},2^{\Omega(n)},2^{-\Omega(n)}, \frac{1}{3})$ direct-product generator. We first briefly describe the construction and then show that it also works for quantum circuits.

The direct-product generator in \cite{iw97} is constructed from the expander random walks (Definition~\ref{def:expander_walk}) and nearly disjoint subsets (Definition~\ref{def:nds}). They defined the direct-product generator $\mathsf{XG}(r,r',v,d):=\mathsf{EW}(v,d)\oplus \mathsf{ND}^\Sigma(r')$, where $\Sigma\subseteq[m]$ is selected by $r$ such that $|r|=O(n), |r'|=m=O(n), |v|=n, |d|=O(n)$. They proved that $\mathsf{XG}$ is $2^{\Omega(n)}$-restrictible and $(O(n),2^{-\Omega(n)},1/3)$-hitting. It's easy to see that the restrictible and hitting properties are pure combinatorial and circuit independent, which means that they also hold for quantum circuits. Then, they proved that these combinatorial properties imply $\mathsf{XG}$ is also a direct product generator. This step, however, need to be reproved for quantum circuits.

\begin{claim}\label{clm:H_prob}
Let $s>0$, $G(r):\{0,1\}^m\rightarrow \{0,1\}^{nk}$ be a $(\rho k, q, \delta)$-hitting, $M$-restrictible pseudo-random generator, where $q>2^{-\rho k/3}, s>2Mnk$. Then, $G$ is a $(s, \Omega(sq^2 n^{-O(1)}), O(q), \delta)$-quantum direct product generator.
\end{claim}
\begin{proof}
Let $\epsilon = (4\delta / \rho + 1) q$. Suppose there is a quantum circuit ${\cal C}$ such that 
\begin{align*}
    \E_{x\sim \{0,1\}^m, {\cal C}}\left[{\cal C}(x) = g^{\otimes k}\circ G(x)\right]\geq \epsilon.
\end{align*}

Then, we construct a quantum circuit ${\cal F}$ of size $O(|C|+kMn)$ such that for any $H\subseteq \{0,1\}^n,|H|\geq \delta 2^n$, 
\begin{align*}
    \E_{y\sim H, {\cal F}}\left[{\cal F}(y) = g(y)\right]\geq \frac{1}{2}+\frac{q}{2}.
\end{align*}
We use the same construction as \cite{iw97}. Let $i\sim [k], \alpha_0\sim \{0,1\}^m$. Let $x_1,\dots,x_k$ be the output of $G(h(i, x,\alpha_0))$. For each $j\ne i$, we non-uniformly construct a table of $g(x_j)$ for any $x_j$ that is a possible output of $G(h(i,x,\alpha_0))$ for different $x$. Since $G$ is $M$-restrictible, each table has at most $M$ values. Then, on input $y\in \{0,1\}^n$, the circuit ${\cal F}$ simulates ${\cal C}$ on $h(i, y, \alpha_0)$ and let $c_1,\dots,c_k$ be the output. Then, for $y_1,\dots,y_k:=G(h(i,y,\alpha))$, ${\cal F}$ counts the number of indices $j\ne i$ such that $c_i\ne g(y_i)$ using the tables. Let $t$ be the number. Then, with probability $2^{-t}$, $\cal F$ outputs $c_i$; otherwise, $\cal F$ outputs a random bit.

For analysis of quantum circuits, as in \cite{agg20}, we first consider ${\cal C}$ being an inherently probabilistic circuit. For any $H\subseteq \{0,1\}^n$, let $y\sim H$ uniformly at random. Then, for any $y_1,\dots,y_k\in (\{0,1\}^n)^k$, 
\begin{align}\label{eq:prob_F_G}
    \Pr_{y\sim H} [y_1,\dots,y_k~\text{generated by }{\cal F}]=\frac{u}{\delta k} \cdot \Pr_{r\sim \{0,1\}^m}[y_1,\dots,y_k~\text{generated by }G(r)].
\end{align}
where $u$ is the number of $y_i\in H$. Since $\E_{r, {\cal C}}[{\cal C}(r)=g^{\otimes k}(G(r))]\geq \epsilon$, for a random $r$, the probability that $u\geq \rho k$ and ${\cal C}(r)=g^{\otimes k}(G(r))$ is at least $\epsilon - q$, by the hitting property of $G$. 
Hence, the probability that $u\geq \rho k$ and ${\cal C}$ succeeds for $y_1,\dots,y_k$ generated by $F$ on a random $x\in H$ is $(\epsilon - q)\cdot \rho/\delta = 4q$, since each $(y_1,\dots,y_k)$ has at least $\rho/\delta$ of its probability under $G(r)$ by Eq.~\eqref{eq:prob_F_G}. Then, we can compute the expected success probability of ${\cal F}$ on $y\in H$ given $u\geq \rho k$ by Theorem 3.2 in \cite{iw97}, which is
\begin{align*}
    \E_{y\sim H}[{\cal F}(y)=g(y) ~|~ u\geq \rho k] \geq \frac{1}{2} + q.
\end{align*}
Since $u\geq \rho k$ has probability at least $1-q$, the overall success probability is at least $(1+q)/2$.
Finally, by Lemma 2.7 in \cite{agg20}, we can change the inherently probabilistic circuit by a quantum circuit and the result still holds. 

Hence, ${\cal F}$ has expected probability $(1+q)/2$ on $1-\delta$ fraction of inputs. Then, we can take $O(n/q^2)$ copies and take the majority of them, which gives a circuit of size $O((|{\cal C}|+kM)n/q^2)\leq s$ if $|{\cal C}|=\Omega(sq^2 n^{-O(1)})$, and has success probability at least $1-\delta$. The Claim is then proved.
\end{proof}
By Claim~\ref{clm:H_prob}, we know that $\mathsf{XG}:\{0,1\}^{O(n)}\rightarrow \{0,1\}^{n^2}$ is a $(2^{\Omega(n)}, 2^{\Omega(n)}, 2^{-\Omega(n)}, 1/3)$-quantum direct product generator.

Finally, feeding the output of $\mathsf{XG}$ to the quantum Nisan-Wigderson generator (Lemma~\ref{lem:quantum NW}) $C_{NW}$ gives the desired quantum pseudo-random generator, which completes the proof of the lemma. 
\end{proof}

%% file: qeth_hardness.tex
\section{Quantum fine-grained hardness based on \textsf{QETH}}\label{sec:fine grained}
In this section, we will show that $2n\times 2n$ bipartite permutation independent set problem is hard under $\textsf{QETH}$. 
\qethperm*

More specifically, We ``quantize'' the fine-grained reduction in \cite{lms11}. The reduction chain is as follows:
\begin{align*}
    & \textsf{3-SAT} \leq_{FG} \textsf{3-Coloring} \leq_{FG} n\times n ~\textsf{Clique} \leq_{FG} n \times n ~\textsf{Permutation Clique}\\
    &\leq_{FG} n\times n ~\textsf{Permutation Independent Set}\\
    & \leq_{FG} 2n\times 2n ~\textsf{Bipartite Permutation Independent Set}
\end{align*}

We first define some intermediate fine-grained problems.
\begin{definition}[$n\times n$ \textsf{Clique} problem]
Given a graph on the vertex set $[n]\times [n]$, decide if there exists $i_1,\dots,i_n \in [n]$ such that the subgraph on $(1,i_1),\dots,(n, i_n)$ forms an $n$-clique.
\end{definition}

\begin{definition}[$n\times n$ \textsf{Permutation Clique/Independent Set} problem]
Given a graph on the vertex set $[n]\times [n]$, decide if there exists a permutation $\pi\in {\cal S}_n$ such that the subgraph on $(1,\pi(1)),\dots,(n, \pi(n))$ forms an $n$-clique/independent set.
\end{definition}

The following claims shows that the aforementioned reductions work for quantum lower bounds.
\begin{claim}\label{clm:qeth_to_3col}
Under \textsf{QETH}, there is no $2^{o(n)}$-time quantum algorithm for \textsf{3-Coloring}, where $n$ is the number of vertices in the input graph. 
\end{claim}
\begin{proof}
By the NP-complete proof of \textsf{3-Coloring}, we know that a 3-CNF formula with $n$ variables and $m$ clauses can be reduced to a \textsf{3-Coloring} instance in time $O(n+m)$. Hence, a $2^{o(n)}$-time quantum algorithm for \textsf{3-Coloring} implies a $2^{o(n)}$-time quantum algorithm for \textsf{3-SAT}, which implies that \textsf{QETH} fails.
\end{proof}

\begin{claim}\label{clm:3col_to_kclique}
If $n\times n$ \textsf{Clique} can be solved in $2^{o(n\log n)}$ time quantumly, then 3-\textsf{Coloring} can be solved in $2^{o(n)}$ time quantumly.
\end{claim}
\begin{proof}
We use the reduction given by \cite{lms11}. Let $G$ be an instance of 3-\textsf{Coloring} with $n$ vertices. The reduction can produce a graph $H$ with vertices $[k]\times [k]$ such that $n\leq k\log_3 k-k$. Then, $G$ is 3-colorable if and only if $H$ is a ``Yes'' instance of $k\times k$ \textsf{Clique}. The reduction takes $\poly(k)$-time classically.

Hence, if there exists a quantum algorithm for $k\times k$ \textsf{Clique} in time $2^{o(k\log k)}$, then it gives a quantum algorithm for 3-\textsf{Coloring} that runs in time $2^{o(n)}$.
\end{proof}

\begin{claim}\label{clm:cli_to_perm_cli}
If $n\times n$ \textsf{Permutation Clique/Independent Set} can be solved in $2^{o(n\log n)}$ time quantumly, then $n\times n$ \textsf{Clique} can also be solved in $2^{o(n\log n)}$ time quantumly.
\end{claim}
\begin{proof}
By \cite{lms11}, there is a reduction from $n\times n$ \textsf{Clique} to $n\times n$ \textsf{Permutation Clique} that takes $2^{O(n\log\log n)}=2^{o(n\log n )}$ time classically.  Hence, the reduction also works for quantum $2^{o(n\log n)}$-time lower bound. 

Note that $n\times n$ \textsf{Permutation Clique} and $n\times n$ \textsf{Permutation Independent Set} are equivalent problem, since we can reduce them by taking the complement graph.
\end{proof}

\begin{claim}\label{clm:ind_to_perm_ind}
If $2n\times 2n$ \textsf{Bipartite Permutation Independent Set} can be solved in $2^{o(n\log n)}$ quantumly, then $n\times n$ \textsf{Permutation Independent Set} can be solved in $2^{o(n\log n)}$ time quantumly.
\end{claim}
\begin{proof}
By \cite{lms11}, the classical reduction takes time $O(n^2)$. Hence, it also works for quantum algorithms.
\end{proof}

Finally, we can prove the \textsf{QETH}-hardness of $2n\times 2n$ bipartite permutation independent set problem:

\begin{proof}[Proof of Lemma~\ref{lem:qeth_perm_ind_set}]
It follows from Claim~\ref{clm:qeth_to_3col}, \ref{clm:3col_to_kclique}, \ref{clm:cli_to_perm_cli} and \ref{clm:ind_to_perm_ind}.
\end{proof}

%% file: UandS_appx.tex
\section{Proofs for Corollary~\ref{cor:SMCSP_QCMA}}
\label{appx:UandS}

\smcspqcma*

\begin{lemma}
\label{lem:c2q_state}
Given $v = [v_0,\dots,v_{2^n-1}]$ for $v_i \in \mathbb{C}$ for $i=0,\dots,2^n-1$, there exists a quantum circuit such that the state $\ket{v}$ can be computed in time $\poly(2^n)$ with $\ipro{i}{v} = v_i$. 
\end{lemma}
\begin{proof}
We show that one can use single-qubit rotations to construct $\ket{v}$. 

We first prepare $\ket{0^{n+1}}$. Then, we do a single-qubit rotation on the first qubit such that 
\begin{align*}
    \ket{0^{n+1}} \rightarrow 
    \sqrt{\frac{\sum_{i=0}^{2^{n-1}-1}|v_i|^2}{\sum_{i=0}^{2^{n}-1} |v_i|^2}} \ket{0}\ket{0^{n}} + \sqrt{\frac{\sum_{i=2^{n-1}}^{2^{n}-1}|v_i|^2}{\sum_{i=0}^{2^{n}-1} |v_i|^2}} \ket{1}\ket{0^{n}}. 
\end{align*}

Then, let the first qubit be the control qubit and apply the controlled rotation to rotate the second qubit to be 
\begin{align*}
    &\sqrt{\frac{\sum_{i=0}^{2^{n-2}-1}|v_i|^2}{\sum_{i=0}^{2^{n-1}-1} |v_i|^2}} \ket{0} + \sqrt{\frac{\sum_{i=2^{n-2}}^{2^{n-1}-1}|v_i|^2}{\sum_{i=0}^{2^{n-1}-1} |v_i|^2}} \ket{1}, \mbox{if the first qubit is } \ket{0}, \\
    &\sqrt{\frac{\sum_{i=2^{n-1}}^{2^{n-1}+2^{n-2}-1}|v_i|^2}{\sum_{i=2^{n-1}}^{2^{n}-1} |v_i|^2}} \ket{0} + \sqrt{\frac{\sum_{i=2^{n-1}+2^{n-2}}^{2^{n}-1}|v_i|^2}{\sum_{i=2^{n-1}}^{2^{n}-1} |v_i|^2}} \ket{1}, \mbox{if the first qubit is } \ket{1}.
\end{align*}

By doing these controlled rotations in sequence, we can obtain 
$\ket{|v|}$ where $\ipro{v}{i} = |v_i|$ for all $i$. Let $v_j = e^{-i\theta_j}|v_j|$ without loss of generality. Then, condition on $j$, we do the following rotation on the ($n+1$)-th qubit: 
\begin{align*}
    \ket{0} \rightarrow e^{-i\theta_j}\ket{0}
\end{align*}
for all $j$. This gives $\ket{v}$. 

Finally, we use at most $2^{O(n)}$ (control) rotations. By Remark~\ref{remark:complexity_upperbound}, each controlled rotation can be implemented with at most $2^{O(n)}$ overhead. Hence, the verifier can construct $\ket{v}$ in time $\poly(2^n)$.

\end{proof}

\begin{proof}[Proof of Corollary~\ref{cor:SMCSP_QCMA}]

Following Lemma~\ref{lem:c2q_state}, we can make $\poly(n,s,t)$ copies of the state in polynomial time. Then, following the proof for Theorem~\ref{thm:SQCMA_UMCSP}, the problem is in $\class{QCMA}$.

\end{proof}

%% file: QCircuit.tex
\section{Quantum Circuit Class}

In this section, we will show some properties of the quantum circuit $\class{QC}(s,\g)$. Note that $\g$ considered in this paper are universal gate set with constant fan-in. So, the results here are also for constant fan-in universal gate sets.

\begin{claim}\label{claim:gatecount}
For $n\in \N$, there exists a constant $c$ such that a random Boolean function $f:\{0,1\}^n\rightarrow \{0,1\}$ has quantum circuit complexity greater $\frac{2^n}{(c+1)n}$ with probability at least $1-2^{\frac{2^n}{c+1}}$.    
\end{claim}
\begin{proof}
For any $s$-gate and ($n+t$)-qubit quantum circuit (where $n+t\leq s$), there are at most
\begin{align*}
    \binom{n+qs+t}{q}^s |\g|^s \leq 2^{cs\log s}
\end{align*}
possible circuits for some constant $c$ large enough, where $\g$ is the quantum gate set, and $q$ is the maximum number of qubits for any gate in $\g$ can operate on. Let $s = \frac{2^n}{(c+1)n}$. Then the number of circuits of size $s$ is at most $2^{cs\log s} < 2^{\frac{c}{c+1}\cdot 2^n}$. 

There are $2^{2^n}$ Boolean functions from $\{0,1\}^n$ to $\{0,1\}$. Suppose we pick one function uniformly randomly, then for every fixed quantum circuit $\Circ$ and input $x\in \{0,1\}^n$, the probability that $\|(\bra{f(x)}\otimes I_{n+t-1})\Circ\ket{x,0^t}\|\geq \frac{1}{2}$ is $\frac{1}{2}$. Therefore, the probability that a fixed quantum circuit can compute $f(x)$ for all $x\in \{0,1\}^n$ is at most $\frac{1}{2^{2^n}}$. By using union bound, the probability that there exists $\Circ$ of size $\frac{2^n}{(c+1)n}$ that can compute $f$ is at most $\frac{2^{\frac{c}{c+1}\cdot 2^n}}{2^{2^n}} = 2^{\frac{2^n}{c+1}}$. 
\end{proof}







\begin{claim}\label{clm:bqc_in_dspace}
For $s=\poly(n)$ and $\g$ a gate set that contains only constant fan-in gates, $\class{BQC}(s,\g)$ is in $\class{DSPACE}(O(s^2))/O(s^2)$. 
\end{claim}

\begin{proof}
The proof follows from the idea of showing $\class{BQP}\subset \class{PSPACE}$. Let $L\in \class{BQC}(s,\g)$ and $\{\Circ_n\}$ be the quantum circuit family in $\class{QC}(s,\g)$ that can solve $L$. Then, we show that there is a $O(s^2)$-space TM $T$ with $O(s\log s)$-bit advice that can simulates $\Circ_n$. 

Let $\Circ_n$ be the advice to $T$. We first calculate the number of bits needed to represent $s$-gate circuit. For each gate, we need $O(\log s)$ to specify its wires and $2^a$ register to record the corresponding unitary, where $a$ is the maximum fan-in of gates in $\g$.  Note that a unitary $U$ may has entries that cannot be written down in bounded bits. Therefore, we let the precision to every entry in $U$ be $\epsilon = \frac{1}{c2^s}$ for some constant $c$ large enough, which requires number of bits $\log \frac{1}{\epsilon} = O(s)$. The total number of bits required for each gate is $O(s)$. and thus the number bits for the circuit is $O(s^2)$.    

Now, suppose $\Circ_n = U_sU_{s-1}\cdots U_{1}$. For any $x\in \{0,1\}^n$ the probability that $\Circ_n$ accepts is 
\begin{align*}
    \sum_{y\in A}|\bra{y}U_sU_{s-1}\cdots U_{1}\ket{x}|^2, 
\end{align*}
where $A:=\{y: y\mbox{ has the first bit as }1 \}$. Then, the TM $T$ computes each branch one-by-one. for any $y\in A$
\begin{align}
    \bra{y}U_sU_{s-1}\cdots U_{1}\ket{x} &= \sum_{z_1,\dots,z_{s-1}\in \{0,1\}}\bra{y}U_s\opro{z_{s-1}}{z_{s-1}}U_{s-1}\opro{z_{s-2}}{z_{s-2}}\cdots \opro{z_{1}}{z_{1}}U_{1}\ket{x}.\label{eq:branch}
\end{align}
Note that $U_i$ is a constant-dimensional unitary and $x$ and $z_j$'s are vectors with exactly one non-zero entry. So, computing $\bra{z_{j}}U_j\ket{z_{j-1}}$ only requires $O(s)$ (since the entries in $U$ takes $O(s)$ space for the precision). Then, since we can also compute $\bra{z_{j}}U_j\ket{z_{j-1}}$ one by one, the space required for each branch in Eq.~\eqref{eq:branch} is just $O(s)$. Therefore, the space we need is at most $O(s^2)$ (including the space for the advice). 

Note that our calculation in Eq.~\eqref{eq:branch} will have error since our precision to each entry in the unitary is $\epsilon = \frac{1}{c2^s}$. Let $\tilde{U}_s\tilde{U}_{s-1}\cdots \tilde{U}_{1}$ be what we really compute. Then, 
\begin{align*}
    \sum_{y\in A}|\bra{y}U_sU_{s-1}\cdots U_{1}\ket{x}|^2 - \sum_{y\in A}|\bra{y}\tilde{U}_s\tilde{U}_{s-1}\cdots \tilde{U}_{1}\ket{x}|^2 \leq O(2^{s+n}\epsilon). 
\end{align*}
By setting $\epsilon = \frac{1}{c2^s}$ for some constant $c$ large enough, $T$ can solve $L$ with probability at least $2/3$ by having an amplified version of $\Circ_n$ at first (e.g., parallel repetition).  

\end{proof}

\begin{claim}[Diagonalization for quantum circuits]\label{clm:qc_diag}
For every $k\in \N_+$, there exists a language $L_k\in \class{PSPACE}$ but $L_k\notin \class{BQC}[n^k]$ for sufficiently large $n$.
\end{claim}
\begin{proof}
By Claim~\ref{clm:bqc_in_dspace}, we know that $\class{BQC}[n^k]$ is contained in $\class{DSPACE}[n^{2k}]/n^{2k}$. By a nonuniform almost everywhere hierarchy for space complexity (Lemma 11 in \cite{os16}), we know that $\class{DSPACE}[n^{3k}]\not\subset \class{DSPACE}[n^{2k}]/n^{2k}$ for sufficiently large $n$. Hence, we can find a language $L_k\notin \class{BQC}[n^k]$.
\end{proof}

\begin{claim}[$\class{BQC}$ size hierarchy]
For $n> 0$, let $s(n)=o(\frac{2^n}{n})$. Then, there exists a Boolean function $f$ in $\class{BQC}[s(n)]\backslash \class{BQC}[s(n)-O(n)]$, i.e., $f$ can be computed by an $s(n)$-size quantum circuit but not computed by any $(s(n)-O(n))$-size quantum circuit.
\end{claim}

\begin{proof}
The proof is very similar to the argument for classical circuits. By Claim~\ref{claim:gatecount}, we can find a function $g$ that requires quantum circuit of size $2^n/cn$ for some $c>1$. Suppose there are $t$ inputs $x_1,\dots,x_t$ such that $g(x_i)=1$ for $i\in [t]$. Then, we construct a series of functions $g_i$ for $i=0,1,\cdots,t$ such that $g_i(x)=1$ if and only if $x\in \{x_1,\dots,x_i\}$. It's easy to see that the following properties are satisfied:
\begin{itemize}
    \item $g_0\in \class{BQC}[0]$ and $g_t \in \class{BQC}[2^n/cn]$.
    \item For $0\leq i<t$, the difference of the quantum circuits size of $g_i$ and $g_{i+1}$ is at most $O(n)$. It follows since $g_i$ and $g_{i+1}$ are only different at $x_i$.
\end{itemize}
Hence, there exists an $i>0$ such that the quantum circuit size of $g_i$ is at most $s(n)$ but lager than $s(n)-O(n)$, since $s(n)=o(2^n/cn)$.

\end{proof}